\documentclass[sort]{msc}
\usepackage{xspace,amsmath,amsfonts,bbm,ifthen,mathtools}

\usepackage{tikz}
\usetikzlibrary{arrows, automata, shapes, calc, fadings}
\tikzset{->, auto, >=stealth', font=\small}
\tikzset{state/.style={shape=circle, draw, fill=white, initial text=,
    inner sep=.5mm, minimum size=1.5mm}}
\tikzset{accepting/.style=accepting by arrow}
\tikzset{state with output/.style={shape=rectangle split, rectangle
    split parts=2, draw, fill=white,
    initial text=, inner sep=1mm}}

\usepackage[matrix,arrow,cmtip,curve]{xy}
\SelectTips{cm}{}
\SilentMatrices

\newcommand\pomsetwop[4]{
  \vcenter{\xymatrix@1@R=#1@C=#2@M=#3{#4}}%
}
\newcommand\pomset[2][1.3]{%
  \left(%
    \pomsetwop{0ex}{#1em}{2pt}{#2}%
  \right)%
}
\newcommand\ipomset[2][1.5]{%
  \hspace*{.6em}%
  \left(%
    \hspace*{-#1em}%
    \pomsetwop{2.7ex}{#1em}{1pt}{#2}%
    \hspace*{-#1em}%
  \right)%
  \hspace*{.6em}%
}

\newcommand*\ie{\textit{i.e.,}\xspace}
\newcommand*\para{\mathrel{\|}}
\newcommand*\twotwo{\textsf{2+2}}
\newcommand*\etal{\textit{et al.}\xspace}
\newcommand*\bang{\mathord{!}}
\newcommand*\bbang{\mathord{!\hspace{-.2ex}!}}
\newcommand*\intord{\dashrightarrow} 
\newcommand*\face{\triangleleft}
\newcommand*\ecaf{\triangleright}
\newcommand*\dotsquare{%
  \ensuremath{\text{%
      $\square$\llap{\raisebox{.2ex}{$\cdot$\hspace*{.55ex}}}}}}
\newcommand*\pobj[1]{\square^{#1}}
\newcommand*\incomp{\mathrel{\|}}
\newcommand*\cat[1]{\text{\textup{\textsf{#1}}}}
\newcommand*\Pos{\cat{Pos}}
\newcommand*\Poms{\cat{Poms}}
\newcommand*\exec{%
  \raisebox{1pt}{%
    \begin{tikzpicture}[x=.8ex,y=1ex,-]
      \draw (0,0) -- (1,0) -- (1,1) -- (2,1);
    \end{tikzpicture}}}
\newcommand*\dotDelta{\ensuremath{\text{%
      $\Delta$\llap{\raisebox{-.5pt}{$\cdot\,$}}}}}
\newcommand*\tPos{\dotDelta}
\newcommand*\op{\textup{\textsf{op}}}
\newcommand*\Set{\cat{Set}}
\newcommand*\mcal[1]{\mathcal{#1}}
\newcommand*\rest[1]{{}_{| #1}}
\newcommand*\yoneda{\mathbf{y}}
\newcommand*\ineda{\mathbf{i}}
\newcommand*\jneda{\mathbf{j}}
\newcommand*\Inj{\cat{Inj}}
\newcommand*\pcSet{\Set^{\square^\op}}
\newcommand*\ev{\textup{\textsf{ev}}}
\newcommand*\eveq{\sim_\ev}
\newcommand*\subid[3]{{}_{#1}#2_{#3}}
\newcommand*\subsu{\sqsubseteq}
\newcommand*\id{\textup{\textsf{id}}}
\newcommand*\iPoms{\cat{iPoms}}
\newcommand*\iiPoms{\cat{iiPoms}}
\newcommand*\Track[1]{\cat{Track}(#1)}
\newcommand*\precchu{<_{\textup{C}}}
\newcommand\red[1]{\textcolor{red}{#1}}
\let\po\vec
\newcommand*\Real{\mathbbm{R}}
\newcommand*\dTop{\cat{dTop}}
\newcommand*\georel[1]{\mathopen| #1\mathclose|}
\newcommand*\intimg[1]{\mathopen] #1\mathclose[}
\newcommand*\carr{\textup{\textsf{carr}}}
\newcommand*\oil{\mathopen]}
\newcommand*\oir{\mathclose[}
\newcommand*\bigoil{\bigl]}
\newcommand*\bigoir{\bigr[}
\newcommand*\IFF{\Longleftrightarrow}
\newcommand*\down{\mathord{\downarrow}}
\newcommand*\cf{\textit{cf.}\xspace}
\newcommand*\biggdown{\mathord{\bigg\downarrow}}
\newcommand*\from{\leftarrow}
\def\X{\mathcal{X}}
\def\Y{\mathcal{Y}}

\theoremstyle{remarkstyle}\newtheorem{example}[therm]{Example}

\begin{document}

\jnlPage{1}{00}
\jnlDoiYr{2020}
\doival{10.1017/xxxxx}

\lefttitle{U.\ Fahrenberg, C.\ Johansen, G.\ Struth and K.\ Ziemiański}
\righttitle{Languages of Higher-Dimensional Automata}
\title{Languages of Higher-Dimensional Automata}

\begin{authgrp}
\author{Uli Fahrenberg}
\affiliation{{\'E}cole Polytechnique, Palaiseau, France}

\author{Christian Johansen}
\affiliation{Norwegian University of Science and Technology, Norway}

\author{Georg Struth}
\affiliation{University of Sheffield, UK}

\author{~~Krzysztof Ziemia{\'n}ski}
\affiliation{University of Warsaw, Poland}
\end{authgrp}

\begin{abstract}
  We introduce languages of higher-dimensional automata (HDAs) and
  develop some of their properties.  To this end, we define a new
  category of precubical sets, uniquely naturally isomorphic to the
  standard one, and introduce a notion of event consistency.  HDAs are
  then finite, labeled, event-consistent precubical sets with
  distinguished subsets of initial and accepting cells.  Their
  languages are sets of interval orders closed under subsumption; as a
  major technical step we expose a bijection between interval orders
  and a subclass of HDAs.  We show that any finite subsumption-closed
  set of interval orders is the language of an HDA, that languages of
  HDAs are closed under binary unions and parallel composition, and
  that bisimilarity implies language equivalence.

\end{abstract}


\maketitle

\section{Introduction}

Higher-dimensional automata (HDAs) are a formalism for modeling and
reasoning about behaviors of concurrent systems, introduced by Pratt
\cite{Pratt91-geometry} and van Glabbeek \cite{Glabbeek91-hda}.  Like
Petri nets \cite{book/Petri62}, event structures
\cite{DBLP:journals/tcs/NielsenPW81}, configuration structures
\cite{DBLP:conf/lics/GlabbeekP95, DBLP:journals/tcs/GlabbeekP09},
asynchronous transition systems \cite{Bednarczyk87-async,
  DBLP:journals/cj/Shields85}, and similar approaches
\cite{pratt95chu, DBLP:journals/acta/GlabbeekG01, Pratt03trans_cancel,
  P15jlamp_STstruct}, they form a model of non-interleaving
concurrency as they differentiate between interleaving and ``truly''
concurrent computations, \ie $a\| b\ne a. b+ b. a$ (using CCS notation
\cite{book/Milner89}).  Van Glabbeek
\cite{DBLP:journals/tcs/Glabbeek06} has shown that HDAs generalize
``the main models of concurrency proposed in the literature'',
including those mentioned above.

\begin{figure}[bp]
  \centering
  \begin{tikzpicture}
    \begin{scope}[x=1.5cm, state/.style={shape=circle, draw,
        fill=white, initial text=, inner sep=1mm, minimum size=3mm}]
      \node[state, black] (10) at (0,0) {};
      \node[state, rectangle] (20) at (0,-1) {$\vphantom{b}a$};
      \node[state] (30) at (0,-2) {};
      \node[state, black] (11) at (1,0) {};
      \node[state, rectangle] (21) at (1,-1) {$b$};
      \node[state] (31) at (1,-2) {};
      \path (10) edge (20);
      \path (20) edge (30);
      \path (11) edge (21);
      \path (21) edge (31);
      \node[state, black] (m) at (.5,-1) {};
      \path (20) edge[out=15, in=165] (m);
      \path (m) edge[out=-165, in=-15] (20);
      \path (21) edge[out=165, in=15] (m);
      \path (m) edge[out=-15, in=-165] (21);
    \end{scope}
    \begin{scope}[xshift=4cm]
      \node[state] (00) at (0,0) {};
      \node[state] (10) at (-1,-1) {};
      \node[state] (01) at (1,-1) {};
      \node[state] (11) at (0,-2) {};
      \path (00) edge node[left] {$\vphantom{b}a$\,} (10);
      \path (00) edge node[right] {\,$b$} (01);
      \path (10) edge node[left] {$b$\,} (11);
      \path (01) edge node[right] {\,$\vphantom{b}a$} (11);
    \end{scope}
  \end{tikzpicture}
  \qquad\qquad
  \begin{tikzpicture}
    \begin{scope}[xshift=8cm]
      \path[fill=black!15] (0,0) to (-1,-1) to (0,-2) to (1,-1);
      \node[state] (00) at (0,0) {};
      \node[state] (10) at (-1,-1) {};
      \node[state] (01) at (1,-1) {};
      \node[state] (11) at (0,-2) {};
      \path (00) edge node[left] {$\vphantom{b}a$\,} (10);
      \path (00) edge node[right] {\,$b$} (01);
      \path (10) edge node[left] {$b$\,} (11);
      \path (01) edge node[right] {\,$\vphantom{b}a$} (11);
    \end{scope}
    \begin{scope}[x=1.5cm, state/.style={shape=circle, draw,
        fill=white, initial text=, inner sep=1mm, minimum size=3mm},
      xshift=10.5cm]
      \node[state, black] (10) at (0,0) {};
      \node[state, rectangle] (20) at (0,-1) {$\vphantom{b}a$};
      \node[state] (30) at (0,-2) {};
      \node[state, black] (11) at (1,0) {};
      \node[state, rectangle] (21) at (1,-1) {$b$};
      \node[state] (31) at (1,-2) {};
      \path (10) edge (20);
      \path (20) edge (30);
      \path (11) edge (21);
      \path (21) edge (31);
    \end{scope}
  \end{tikzpicture}
  \bigskip\bigskip
  \caption{Petri net and HDA models distinguishing interleaving (left)
    from non-interleaving (right) concurrency.  Left: Petri net and
    HDA models for $a. b+ b. a$; right: HDA and Petri net models for
    $a\para b$.}
  \label{fi:int-conc}
\end{figure}
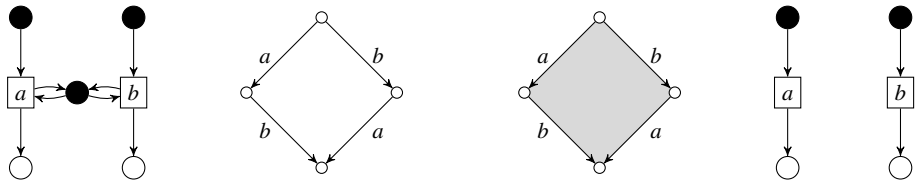

HDAs extend finite automata with additional structure that
distinguishes interleavings from concurrency.  As an example, Figure
\ref{fi:int-conc} shows Petri net and HDA models for a system with two
events, labeled $a$ and $b$.  The Petri net and HDA on the left model
the (mutually exclusive) interleaving of $a$ and $b$ as either $a. b$
or $b. a$; those on the right model concurrent execution of $a$ and
$b$.  In the HDA, this independence is indicated by a filled-in
square.

HDAs thus have states and transitions like finite automata, but may
also contain squares, cubes, and higher-dimensional cubical
structures.  A square stands for the concurrent execution of two
events; a cube for the concurrent execution of three events; and so
on.

\begin{figure}
  \centering
  \begin{tikzpicture}[x=1.5cm, y=1.5cm]
    \path[fill=black!15] (0,0) to (2,0) to (2,1) to (1,1) to
    (0,1);
    \node[state, initial] (00) at (0,0) {};
    \node[state] (10) at (1,0) {};
    \node[state] (20) at (2,0) {};
    \node[state] (01) at (0,1) {};
    \node[state] (11) at (1,1) {};
    \node[state, accepting] (21) at (2,1) {};
    \path (00) edge node[below] {$\vphantom{d}c$} (10);
    \path (10) edge node[below] {$d$} (20);
    \path (01) edge node[above] {$c$} (11);
    \path (11) edge node[above] {$d$} (21);
    \path (00) edge node[left] {$\vphantom{y}a$} (01);
    \path (10) edge (11);
    \path (20) edge node[right] {$\vphantom{y}a$} (21);
  \end{tikzpicture}
  \bigskip\medskip
  \caption{HDA which executes $a$ in parallel with $c. d$.  Initial
    and accepting cells marked with incoming and outgoing arrows.}
  \label{fi:hda-a|cd}
\end{figure}
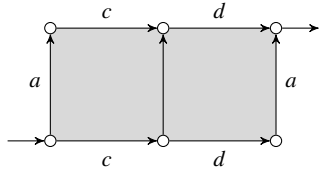

This paper is concerned with \emph{languages} of HDAs.  Like languages
related to other formalisms for concurrency, these need to account for
both the sequential and the concurrent nature of computations. Their
elements will therefore be finite \emph{pomsets} or \emph{partial
  words} \cite{DBLP:journals/ipl/Winkowski77}.  As an example, Figure
\ref{fi:hda-a|cd} displays an HDA consisting of two squares, with
three events labeled $a$, $c$, and $d$.  Here the $a$-labeled event is
executed concurrently to the sequence $c. d$, so that the language of
this HDA will contain the pomset
\begin{equation}
  \label{eq:a|cd}
  \pomset[.4]{ & a \\ c \ar[rr] && d}.
\end{equation}
(It will contain other elements; but in a sense to be made precise
below, they are all generated by this one pomset.)

Partial words and pomsets have been introduced by Winkowski
\cite{DBLP:journals/ipl/Winkowski77} and have a long history as
semantics for concurrent systems \cite{Pratt86pomsets,
  DBLP:books/sp/Vogler92}.  The subclass of \emph{interval orders},
introduced by Fishburn \cite{journals/mpsy/Fishburn70}, has seen
abundant attention in concurrency theory and distributed systems
\cite{DBLP:journals/dc/Lamport86,
  DBLP:journals/jacm/Lamport86a, 
  DBLP:journals/dc/Vogler91, DBLP:books/sp/Vogler92,
  DBLP:journals/tcs/JanickiK93, DBLP:journals/toplas/HerlihyW90,
  DBLP:journals/iandc/JanickiY17, DBLP:conf/RelMiCS/FahrenbergJST20}.
A pomset is an interval order precisely if it is
\emph{$\twotwo$-free}, that is, does not contain an induced subpomset
of the form
\begin{equation*}
  \twotwo= \pomset{\bullet \ar[r] & \bullet \\ \bullet \ar[r] & \bullet}.
\end{equation*}

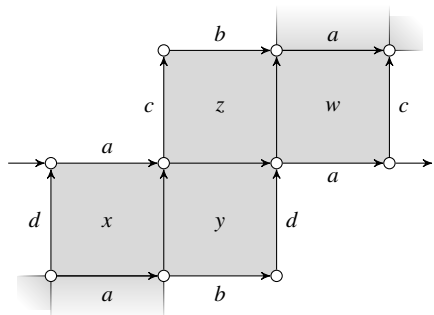
\begin{figure}[bp]
  \centering
  \begin{tikzpicture}[x=1.5cm,y=1.5cm]
    \path[fill=black!15] (0,0) to (2,0) to (2,1) to (3,1) to (3,2) to
    (1,2) to (1,1) to (0,1);
    \fill[-, fill=black!15, path fading=east, fading
    transform={rotate=45}] (3,2) -- (3.3,2) -- (3.3,2.3) -- (3,2.3) --
    (3,2);
    \filldraw[-, fill=black!15, path fading=north] (2,2.4) -- (2,2) --
    (3,2) -- (3,2.4);
    \draw[-, path fading=east] (3,2) -- (3.3,2);
    \fill[-, fill=black!15, path fading=west, fading
    transform={rotate=45}] (0,0) -- (-.3,0) -- (-.3,-.3) -- (0,-.3) --
    (0,0);
    \filldraw[-, fill=black!15, path fading=south] (0,-.4) -- (0,0) --
    (1,0) -- (1,-.4);
    \draw[-, path fading=west] (0,0) -- (-.3,0);
    \node[state] (00) at (0,0) {};
    \node[state] (10) at (1,0) {};
    \node[state] (20) at (2,0) {};
    \node[state, initial] (01) at (0,1) {};
    \node[state] (11) at (1,1) {};
    \node[state] (21) at (2,1) {};
    \node[state, accepting] (31) at (3,1) {};
    \node[state] (12) at (1,2) {};
    \node[state] (22) at (2,2) {};
    \node[state] (32) at (3,2) {};
    \node at (.5,.5) {$\vphantom{d}x$};
    \node at (1.5,.5) {$\vphantom{d}\smash[b]{y}$};
    \node at (1.5,1.5) {$z$};
    \node at (2.5,1.5) {$w$};
    \path (00) edge node[below] {$\vphantom{b}a$} (10);
    \path (10) edge node[below] {$b$} (20);
    \path (01) edge node[above] {$a$} (11);
    \path (11) edge (21);
    \path (21) edge node[below] {$a$} (31);
    \path (12) edge node[above] {$b$} (22);
    \path (22) edge node[above]{$a$} (32);
    \path (00) edge node[left] {$d$} (01);
    \path (10) edge (11);
    \path (11) edge node[left] {$c$} (12);
    \path (20) edge node[right] {$d$} (21);
    \path (21) edge (22);
    \path (31) edge node[right] {$c$} (32);
  \end{tikzpicture}
  \bigskip
  \caption{HDA which generates infinite set of pomsets (bottom left
    and top right edges identified).}
  \label{fi:hda-loop}
\end{figure}

We will show that languages of HDAs are sets of interval orders, and
that any interval order may be generated by an HDA.  For another
example, the HDA in Figure \ref{fi:hda-loop} has a two-dimensional
\emph{loop} created by identifying the horizontal edges in the
bottom-left and top-right of the automaton (together with their
corresponding faces).  Its language includes the infinite set
\begin{equation}
  \label{eq:hda-loop-lang}
  \bigg\{ \pomset{ a \ar[r] & b \ar[r] & a},%
  \pomset{
    a \ar[r] \ar[dr] & b \ar[r] \ar[dr] & a \ar[r] & b \ar[r] & a\\
    & c \ar[r] \ar[urr] & d \ar[urr]},\dotsc \bigg\},
\end{equation}
where the second pomset is obtained by traversing the squares $z$,
$w$, $x$ and $y$ in that order.  We will only be concerned with
\emph{finite} HDAs in this paper, yet as the above example shows,
languages of finite HDAs may well be infinite.

A precursor to this work is van Glabbeek's
\cite{DBLP:journals/tcs/Glabbeek06}, which introduces \emph{tracks} in
HDAs (there called paths) and then defines their observable content in
terms of \emph{ST-traces}.  We have shown in
\cite{DBLP:conf/RelMiCS/FahrenbergJST20} that there is a bijective
correspondence between ST-traces and interval orders.  Another
precursor is Fajstrup \etal's \cite{DBLP:journals/tcs/FajstrupRG06},
where the authors define computations as \emph{directed paths} through
geometric cubical complexes.  We introduce languages based on van
Glabbeek's tracks and languages based on Fajstrup \etal's directed
paths, and show that they define the same objects.

Grabowski \cite{DBLP:journals/fuin/Grabowski81} has introduced a
notion of \emph{smoothing} for pomsets which is nowadays mostly called
\emph{subsumption} \cite{DBLP:journals/tcs/Gischer88,
  DBLP:conf/concur/FanchonM02}: a pomset $P$ subsumes a pomset $Q$ if $Q$
is at least as ordered as $P$.  Sets of pomsets closed under
subsumption are generally called \emph{weak}
\cite{DBLP:journals/fuin/Grabowski81, DBLP:conf/concur/FanchonM02}.
We show that languages of HDAs are weak sets of interval orders.

\begin{figure}
  \centering
  \begin{tikzpicture}[x=1.5cm, y=1.5cm]
    \begin{scope}
      \path[fill=black!15] (0,0) to (2,0) to (2,1) to (1,1) to
      (0,1);
      \node[state, initial] (00) at (0,0) {};
      \node[state] (10) at (1,0) {};
      \node[state] (20) at (2,0) {};
      \node[state] (01) at (0,1) {};
      \node[state] (11) at (1,1) {};
      \node[state, accepting] (21) at (2,1) {};
      \path (00) edge node[below] {$\vphantom{d}c$} (10);
      \path (10) edge node[below] {$d$} (20);
      \path (01) edge node[above] {$c$} (11);
      \path (11) edge node[above] {$d$} (21);
      \path (00) edge node[left] {$a$} (01);
      \path (10) edge (11);
      \path (20) edge node[right] {$a$} (21);
      \draw[-, very thick, orange] (0,0) -- (2,1);
      \node[font=\normalsize] at (1,-.7) {$\pomset[.4]{ & a \\ c
          \ar[rr] && d}$};
    \end{scope}
    \begin{scope}[shift={(3.2,0)}]
      \path[fill=black!15] (0,0) to (2,0) to (2,1) to (1,1) to
      (0,1);
      \node[state, initial] (00) at (0,0) {};
      \node[state] (10) at (1,0) {};
      \node[state] (20) at (2,0) {};
      \node[state] (01) at (0,1) {};
      \node[state] (11) at (1,1) {};
      \node[state, accepting] (21) at (2,1) {};
      \path (00) edge node[below] {$\vphantom{d}c$} (10);
      \path (10) edge node[below] {$d$} (20);
      \path (01) edge node[above] {$c$} (11);
      \path (11) edge node[above] {$d$} (21);
      \path (00) edge node[left] {$a$} (01);
      \path (10) edge (11);
      \path (20) edge node[right] {$a$} (21);
      \draw[-, very thick, orange] (0,0) -- (1,1) -- (2,1);
      \node[font=\normalsize] at (1,-.7) {$\pomset{a\ar[dr] \\ c\ar[r] & d}$};
    \end{scope}
    \begin{scope}[shift={(6.4,0)}]
      \path[fill=black!15] (0,0) to (2,0) to (2,1) to (1,1) to
      (0,1);
      \node[state, initial] (00) at (0,0) {};
      \node[state] (10) at (1,0) {};
      \node[state] (20) at (2,0) {};
      \node[state] (01) at (0,1) {};
      \node[state] (11) at (1,1) {};
      \node[state, accepting] (21) at (2,1) {};
      \path (00) edge node[below] {$\vphantom{d}c$} (10);
      \path (10) edge node[below] {$d$} (20);
      \path (01) edge node[above] {$c$} (11);
      \path (11) edge node[above] {$d$} (21);
      \path (00) edge node[left] {$a$} (01);
      \path (10) edge (11);
      \path (20) edge node[right] {$a$} (21);
      \draw[-, very thick, orange] (0,0) -- (1,0) -- (2,1);
      \node[font=\normalsize] at (1,-.7) {$\pomset{ & a \\ c\ar[r]\ar[ur] & d}$};
    \end{scope}
    \begin{scope}[shift={(0,-2.5)}]
    \begin{scope}
      \path[fill=black!15] (0,0) to (2,0) to (2,1) to (1,1) to
      (0,1);
      \node[state, initial] (00) at (0,0) {};
      \node[state] (10) at (1,0) {};
      \node[state] (20) at (2,0) {};
      \node[state] (01) at (0,1) {};
      \node[state] (11) at (1,1) {};
      \node[state, accepting] (21) at (2,1) {};
      \path (00) edge node[below] {$\vphantom{d}c$} (10);
      \path (10) edge node[below] {$d$} (20);
      \path (01) edge node[above] {$c$} (11);
      \path (11) edge node[above] {$d$} (21);
      \path (00) edge node[left] {$a$} (01);
      \path (10) edge (11);
      \path (20) edge node[right] {$a$} (21);
      \draw[-, very thick, orange] (0,0) -- (0,1) -- (2,1);
      \node[font=\normalsize] at (1,-.6) {$\pomset{a\ar[r] & c\ar[r] &
          d}$};
    \end{scope}
    \begin{scope}[shift={(3.2,0)}]
      \path[fill=black!15] (0,0) to (2,0) to (2,1) to (1,1) to
      (0,1);
      \node[state, initial] (00) at (0,0) {};
      \node[state] (10) at (1,0) {};
      \node[state] (20) at (2,0) {};
      \node[state] (01) at (0,1) {};
      \node[state] (11) at (1,1) {};
      \node[state, accepting] (21) at (2,1) {};
      \path (00) edge node[below] {$\vphantom{d}c$} (10);
      \path (10) edge node[below] {$d$} (20);
      \path (01) edge node[above] {$c$} (11);
      \path (11) edge node[above] {$d$} (21);
      \path (00) edge node[left] {$a$} (01);
      \path (10) edge (11);
      \path (20) edge node[right] {$a$} (21);
      \draw[-, very thick, orange] (0,0) -- (1,0) -- (1,1) -- (2,1);
      \node[font=\normalsize] at (1,-.6) {$\pomset{c\ar[r] & a\ar[r] &
          d}$};
    \end{scope}
    \begin{scope}[shift={(6.4,0)}]
      \path[fill=black!15] (0,0) to (2,0) to (2,1) to (1,1) to
      (0,1);
      \node[state, initial] (00) at (0,0) {};
      \node[state] (10) at (1,0) {};
      \node[state] (20) at (2,0) {};
      \node[state] (01) at (0,1) {};
      \node[state] (11) at (1,1) {};
      \node[state, accepting] (21) at (2,1) {};
      \path (00) edge node[below] {$\vphantom{d}c$} (10);
      \path (10) edge node[below] {$d$} (20);
      \path (01) edge node[above] {$c$} (11);
      \path (11) edge node[above] {$d$} (21);
      \path (00) edge node[left] {$a$} (01);
      \path (10) edge (11);
      \path (20) edge node[right] {$a$} (21);
      \draw[-, very thick, orange] (0,0) -- (2,0) -- (2,1);
      \node[font=\normalsize] at (1,-.6) {$\pomset{c\ar[r] & d\ar[r] &
          a}$};
    \end{scope}
    \end{scope}
  \end{tikzpicture}
  \smallskip
  \caption{Directed paths in HDA of Figure \ref{fi:hda-a|cd} together
    with corresponding pomsets.}
  \label{fi:hda-a|cd-dpaths}
\end{figure}
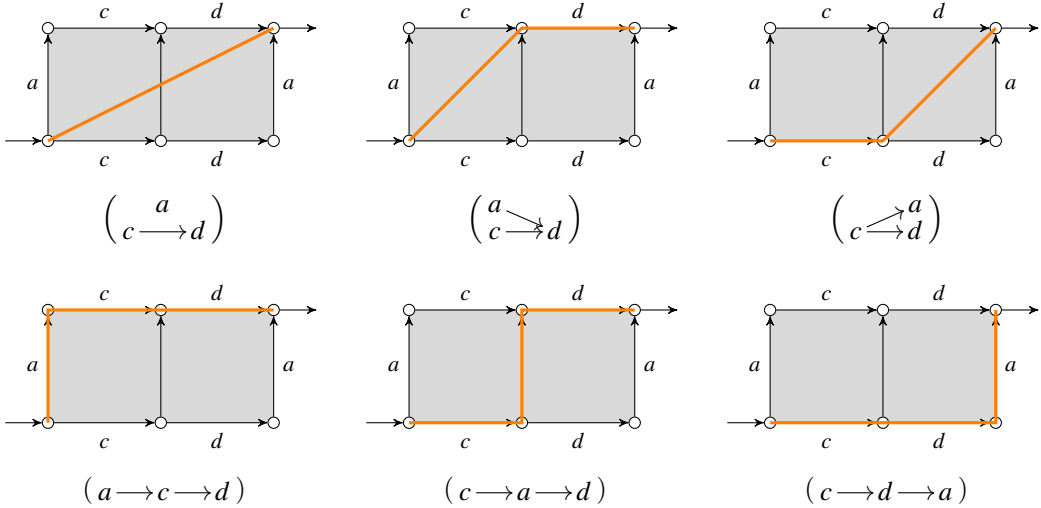

Figure \ref{fi:hda-a|cd-dpaths} exhibits six directed paths through
the HDA $A$ of Figure \ref{fi:hda-a|cd} together with the
corresponding pomsets.  The language of $A$ consists precisely of
these six pomsets; it is also the weak closure of the pomset in
\eqref{eq:a|cd} corresponding to the first directed path displayed.
The language of the HDA in Figure \ref{fi:hda-loop} is the weak
closure of the infinite set in \eqref{eq:hda-loop-lang}.

We finish the paper by showing that languages of HDAs are closed under
binary union and parallel composition, and further that
\emph{bisimilarity} of HDAs \cite{DBLP:conf/fossacs/Fahrenberg05,
  DBLP:journals/tcs/Glabbeek06} implies language equivalence.  A
comprehensive treatment of regular operations on HDAs and their
languages is left for future work.

We start this paper with an overview section which introduces the main
concepts and results without going into too much technical detail.  In
order to properly define and develop languages of HDAs, we first
introduce a new base category for precubical sets, identify a new
subclass of \emph{event consistent} precubical sets, and make clear
the relationship between tracks and interval orders.  This is also why
we define HDAs only on page \pageref{de:hda}.

We detail the main technical contributions of this paper at the end of
the overview Section \ref{se:overview}.  Afterwards, we introduce
precubical sets, event consistency, and HDAs in Section \ref{se:hda}.
Section \ref{se:ipoms} is concerned with pomsets with interfaces,
their gluing composition, and representations of interval orders.  The
connection between interval orders and tracks in precubical sets is
made in Section \ref{se:tracks}, and directed paths are introduced in
Section \ref{se:geo}.  Section \ref{se:lang} concludes the paper by
defining languages of HDAs and developing some basic properties.

\section{Overview}
\label{se:overview}

\begin{figure}
  \centering
  \begin{tikzpicture}[x=.9cm,y=.9cm]
    \path[fill=black!15] (0,0) to (2,0) to (2,2) to (0,2) to (0,0);
    \node[state] (00) at (0,0) {};
    \node[state] (10) at (2,0) {};
    \node[state] (01) at (0,2) {};
    \node[state] (11) at (2,2) {};
    \path (00) edge (01);
    \path (00) edge (10);
    \path (01) edge (11);
    \path (10) edge (11);
    \node at (1,1.02) {$x$};
    \node at (-.4,1.05) {$\delta_1^0 x$};
    \node at (2.4,1.05) {$\delta_1^1 x$};
    \node at (1,-.25) {$\delta_2^0 x$};
    \node at (1,2.25) {$\delta_2^1 x$};
    \node at (-1,-.35) {$\delta_1^0 \delta_2^0 x= \delta_1^0
      \delta_1^0 x$};
    \node at (-1,2.35) {$\delta_1^0 \delta_2^1 x= \delta_1^1
      \delta_1^0 x$};
    \node at (3,-.35) {$\delta_1^1 \delta_2^0 x= \delta_1^0
      \delta_1^1 x$};
    \node at (3,2.35) {$\delta_1^1 \delta_2^1 x= \delta_1^1
      \delta_1^1 x$};
  \end{tikzpicture}
  \bigskip
  \caption{%
    \label{fi:2cubefaces-full}
    A square $x$ with its four elementary faces $\delta_1^0 x$,
    $\delta_1^1 x$, $\delta_2^0 x$, $\delta_2^1 x$ and four corners.
  }
\end{figure}
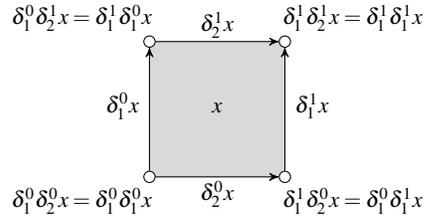

HDAs are built on \emph{precubical sets} \cite{thesis/Serre51,
  book/Grandis09}, a generalization of directed graphs to higher
dimensions.  To be precise, a precubical set consists of a graded set
$X= \bigcup_{ n\ge 0} X_n$ of $n$-cells together with \emph{elementary
  face maps} $\delta_{ i, n}^\nu: X_n\to X_{ n- 1}$,
$i\in\{ 1,\dotsc, n\}$, $\nu\in\{ 0, 1\}$ that specify boundaries of
$n$-cells.  These are required to satisfy the \emph{precubical
  identities}
\begin{equation*}
  \delta_{ i, n- 1}^\nu \delta_{ j, n}^\mu= \delta_{ j- 1, n- 1}^\mu
  \delta_{ i, n}^\nu,
\end{equation*}
for every $i< j\le n$, which identify common elementary faces of
elementary faces.  Figure \ref{fi:2cubefaces-full} shows an example of
a $2$-cell with all its faces; we will generally omit parentheses for
elementary faces and the subscript $n$ and thus, for example, write
$\delta_1^0 x$ instead of $\delta_{ 1, n}^0( x)$.

A precubical set $X$ with $X_1= \emptyset$, hence $X_i= \emptyset$ for
all $i\ge 1$, is simply a set (of $0$-cells or \emph{points}).  A
one-dimensional precubical set $X$, with $X_2= \emptyset$, is a
directed graph.  $1$-cells are generally called \emph{edges},
$2$-cells, \emph{squares}, and $3$-cells, \emph{cubes}.  Modifying the
standard setting \cite{book/Grandis09}, we introduce precubical sets
as presheaves over a category of linearly ordered sets with suitable
morphisms (Definition \ref{de:widepc}).  From a technical point of
view this does not matter, as our ``large'' category of precubical sets
is uniquely isomorphic to the standard one (Proposition
\ref{pr:precubidot}); yet it clarifies the relation between ordered
sets, presimplicial sets and precubical sets.  This simplifies later
developments.

An HDA is a tuple $( X, I, F, \lambda)$ with $X$ a precubical set,
$I, F\subseteq X$ subsets of initial and accepting cells, and
$\lambda$ a labeling on $X$.  This labeling is generated by a function
$\lambda_1: X_1\to \Sigma$, into an alphabet $\Sigma$, which satisfies
$\lambda_1( \delta_1^0 x)= \lambda_1( \delta_1^1 x)$ and
$\lambda_1( \delta_2^0 x)= \lambda_1( \delta_2^1 x)$ for every
$x\in X_2$; but we will extend it to a precubical morphism
$\lambda: X\to \bang \Sigma$ into a special labeling object
$\bang \Sigma$ (Definition \ref{de:labelob}).

One-dimensional HDAs are equivalent to ordinary finite automata, with
$0$-cells as states and $1$-cells as transitions.  Two-dimensional
HDAs are equivalent to asynchronous transition systems, with the
$2$-cells denoting independence of events.

Most formalisms for non-interleaving concurrency have a notion of
\emph{events}: unique occurrences of actions in space and time.
HDAs, on the other hand, do not have a well-defined notion of event
\cite{Pratt00Sculptures, DBLP:journals/corr/FahrenbergJTZ18}.  Going
back to the example in Figure \ref{fi:int-conc}, we see that the Petri
nets on each side of the figure have two events each, induced by their
transitions and labeled $a$ and $b$, respectively.  In the two HDAs on
the other hand, every label appears twice, and there is no immediate
conception of events.  For the HDA on the right, we may deduce from
the presence of the square, which indicates \emph{two} events running
concurrently, that there are indeed precisely two events in the
system; but on the left, there might as well be four.

We make the notion of event identification precise in Definition
\ref{de:labelob} and identify a subclass of \emph{event consistent}
precubical sets: precubical sets $X$ that admit an equivalence
relation $\sim$ on $X_1$ such that for all $x\in X_2$,
$\delta_1^0 x\sim \delta_1^1 x$, $\delta_2^0 x\sim \delta_2^1 x$, and
$\delta_1^0 x\not\sim \delta_2^0 x$ (Lemma \ref{le:evcons}).  The
equivalence classes of the smallest such equivalence are called the
\emph{universal events} of $X$: the largest possible identification of
events which is consistent with the structure of the HDA.

In the example in Figure \ref{fi:int-conc}, the HDA on the left has
four universal events, whereas the one on the right has two.  (An
example of a precubical set which is not event consistent is shown in
Figure \ref{fi:pcsselflev} on page \pageref{fi:pcsselflev}.)

Any labeling factors uniquely through the universal events
(Proposition \ref{prp:EventsAreUniversalLabeling}), so that we could
have written this paper only with \emph{unlabeled} (but event
consistent) HDAs in mind and then added labels as an afterthought,
much in the spirit of \cite{WinskelN95-Models}.  For sake of
readability we have refrained from doing so.

\begin{figure}
  \begin{equation*}
    \ipomset{ & & a \ar[r] & \\ & b \ar[r] & c &}\quad*\quad
    \ipomset{ \ar[r] & a & \\ & d &}\quad=\quad
    \ipomset{ & & a & & \\ & b \ar[r] & c \ar[r] & d &}
  \end{equation*}
  \bigskip
  \caption{Two ipomsets and their gluing composition (interfaces
    marked with incoming and outgoing arrows).}
  \label{fi:ipoms-glue}
\end{figure}

In Section \ref{se:ipoms} we recapitulate the notion of \emph{pomset
  with interfaces} (\emph{ipomset}) from
\cite{DBLP:conf/RelMiCS/FahrenbergJST20}.  An ipomset
$( P, \mathord<_P, \lambda_P, S_P, T_P)$ consists of a labeled partial
order $( P, \mathord<_P, \lambda_P)$ together with subsets
$S_P, T_P\subseteq P$ of minimal and maximal elements which designate
starting and terminating interfaces.  Ipomsets may be \emph{glued}
along their interfaces: if $( Q, \mathord<_Q, \lambda_Q, S_Q, T_Q)$ is
another ipomset such that $P\cap Q= T_P= S_Q$, then $P* Q$ is the
ipomset
\begin{equation*}
  ( P\cup Q, \mathord<_P\cup \mathord<_Q\cup( P\setminus T_P)\times(
  Q\setminus S_Q), \lambda_P\cup \lambda_Q, S_P, T_Q),
\end{equation*}
with the order defined by those of $P$ and $Q$ together with imposing
that every event not in $Q$ precedes every event not in $P$.
Hence events in the overlap $P\cap Q$ are continued across the gluing
composition; Figure \ref{fi:ipoms-glue} shows an example.

\begin{figure}
  \centering
  \begin{tikzpicture}[x=1.3cm, y=1.3cm]
    \begin{scope}
      \path[fill=black!20] (0,0) to (2,0) to (2,1) to (0,1);
      \node[state, initial] (00) at (0,0) {};
      \node[state] (10) at (1,0) {};
      \node[state] (20) at (2,0) {};
      \node[state] (01) at (0,1) {};
      \node[state] (11) at (1,1) {};
      \node[state] (21) at (2,1) {};
      \path (00) edge node[below] {$\vphantom{d}b$} (10);
      \path (10) edge node[below] {$\vphantom{d}c$} (20);
      \path (01) edge (11);
      \path (11) edge (21);
      \path (00) edge node[left] {$a$} (01);
      \path (10) edge (11);
      \path (20) edge (21);
      \path (2.05,.5) edge (2.4,.5);
    \end{scope}
    \begin{scope}[shift={(3.5,0)}]
      \path[fill=black!20] (0,0) to (1,0) to (1,1) to (0,1);
      \node[state] (00) at (0,0) {};
      \node[state] (10) at (1,0) {};
      \node[state] (01) at (0,1) {};
      \node[state, accepting] (11) at (1,1) {};
      \path (00) edge node[below]{$\vphantom{d}d$} (10);
      \path (01) edge (11);
      \path (00) edge (01);
      \path (-.4,.5) edge (-.05,.5);
      \path (10) edge node[right] {$a$} (11);
    \end{scope}
    \begin{scope}[shift={(6,0)}]
      \path[fill=black!20] (0,0) to (3,0) to (3,1) to (0,1);
      \node[state, initial] (00) at (0,0) {};
      \node[state] (10) at (1,0) {};
      \node[state] (20) at (2,0) {};
      \node[state] (30) at (3,0) {};
      \node[state] (01) at (0,1) {};
      \node[state] (11) at (1,1) {};
      \node[state] (21) at (2,1) {};
      \node[state, accepting] (31) at (3,1) {};
      \path (00) edge node[below]{$\vphantom{d}b$} (10);
      \path (10) edge node[below]{$\vphantom{d}c$} (20);
      \path (20) edge node[below]{$\vphantom{d}d$} (30);
      \path (00) edge node[left]{$a$} (01);
      \path (10) edge (11);
      \path (20) edge (21);
      \path (30) edge (31);
      \path (01) edge (11);
      \path (11) edge (21);
      \path (21) edge (31);
    \end{scope}
  \end{tikzpicture}
  \bigskip
  \caption{Three HDAs corresponding to ipomsets of Figure
    \ref{fi:ipoms-glue}.}
  \label{fi:ipoms-glue-hda}
\end{figure}
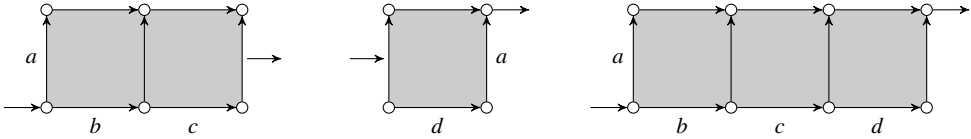

We extend ipomsets with an \emph{event order}: a second strict order,
denoted $\intord_P$, which is required to be linear on
$<_P$-antichains.  This allows us to assign which interfaces are
identified in gluing compositions and also establishes a close
relation between interval-ordered ipomsets and a subclass of HDAs, see
Definition \ref{de:pobj}.  As an example, Figure
\ref{fi:ipoms-glue-hda} shows the three HDAs corresponding to the
ipomsets of Figure \ref{fi:ipoms-glue}; we show in Lemma
\ref{le:pobjJoin} that gluing compositions of interval ipomsets
correspond to pushouts of their induced HDAs.

Most papers in concurrency theory define pomsets as isomorphism
classes of labeled partial orders.  We find it convenient to work
directly with labeled partial orders instead and consider properties
up to isomorphism.  As any isomorphic ipomsets are uniquely isomorphic
(Lemma \ref{le:CanonicalIso}), the difference is without significance.

One central mathematical insight on which this paper is built is that
both precubical sets and interval ipomsets can be obtained by gluing
linear orders, \ie precubical sets glued as presheaves, and interval
ipomsets as gluing compositions of discrete ipomsets.  An ipomset
$( P, \mathord<_P, \mathord{\intord_P}, \lambda_P, S_P, T_P)$ is
\emph{discrete} if $<_P$ is trivial, thus $\intord_P$ is a linear
order.  We always think of $<_P$ as a \emph{precedence} order; hence
all events are concurrent in a discrete ipomset, and the event order
$\intord_P$ is used as a book-keeping device.  Seen as linear
$\intord_P$-ordered sets, discrete ipomsets form our base category for
precubical sets; as trivial $<_P$-ordered sets, we may glue them into
interval orders.

The subclass of precubical sets that correspond to interval ipomsets
under this identification is comprised of \emph{tracks}: sequences of
cells connected at intermediate faces.  This notion generalizes paths
in finite automata to higher dimensions.  Figure \ref{fi:track} shows
an example, a track consisting of six cells $x_1,\dotsc, x_6$ ($x_4$
is the central cube in the figure) with face relations
\begin{equation*}
  x_1\face x_2\ecaf x_3\face^3 x_4\ecaf^2 x_5\face x_6,
\end{equation*}
where $x_1\face x_2$ denotes that $x_1$ is a lower face of $x_2$, and
$x_4\ecaf^2 x_5$ that $x_5$ is an upper face of an upper face of $x_4$
(in anticipation of notation introduced later on).  We show in Section
\ref{se:tracks} how tracks give rise to interval ipomsets, but also
how interval ipomsets can be converted into tracks.

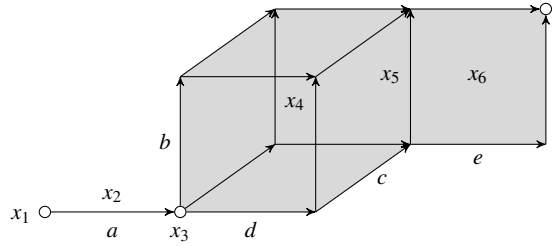
\begin{figure}
  \centering
  \begin{tikzpicture}[x=1.8cm, y=1.8cm]
    \path[fill=black!15] (1,0) -- (2,0) -- (2.7,.5) -- (3.7,.5) --
    (3.7,1.5) -- (1.7,1.5) -- (1,1);
    \node[state] (0) at (0,0) {};
    \node[state] (1) at (1,0) {};
    \node[coordinate] (1a) at (1,1) {};
    \node[coordinate] (2) at (2,0) {};
    \node[coordinate] (2a) at (2,1) {};
    \node[coordinate] (3) at (2.7,.5) {};
    \node[coordinate] (3a) at (1.7,.5) {};
    \node[coordinate] (4) at (2.7,1.5) {};
    \node[coordinate] (4a) at (1.7,1.5) {};
    \node[coordinate] (5) at (3.7,.5) {};
    \node[state] (6) at (3.7,1.5) {};
    \node[left] at (0.south west) {$x_1$};
    \path (0) edge node[above] {$x_2$} node[below] {$\vphantom{d}a$} (1);
    \node[below] at (1.south) {$x_3$};
    \path (1) edge node[below] {$d$} (2);
    \node at (1.85,.8) {$x_4$};
    \path (1) edge node[left] {$b$} (1a);
    \path (1) edge (3a);
    \path (1a) edge (2a);
    \path (1a) edge (4a);
    \path (2) edge node[below, pos=.7] {$c$} (3);
    \path (2) edge (2a);
    \path (2a) edge (4);
    \path (3) edge node[left] {$x_5$} (4);
    \node at (3.2,1) {$x_6$};
    \path (3) edge node[below] {$e$} (5);
    \path (3a) edge (3);
    \path (3a) edge (4a);
    \path (4a) edge (4);
    \path (4) edge (6);
    \path (5) edge (6);
  \end{tikzpicture}
  \medskip
  \caption{A track in a precubical set.}
  \label{fi:track}
\end{figure}

In Section \ref{se:geo} we give a \emph{geometric} interpretation of
executions in HDAs, following \cite{DBLP:journals/tcs/FajstrupRG06}
and subsequent related work \cite{DBLP:books/sp/FajstrupGHMR16,
  DBLP:journals/aaecc/Ziemianski17, journals/act/Ziemianski20}.
Precubial sets may be realized geometrically as \emph{directed spaces}
\cite{book/Grandis09}, and executions of HDAs may then be seen as
\emph{directed paths} through their geometric realization.  We
introduce the \emph{interval arrangement} of a directed path, which
tracks the events that are active during different phases of the
execution, and use this to define labels of directed paths.

In Section \ref{se:lang} we show that languages of HDAs defined by
directed paths are the same as languages defined by tracks.  We also
see that languages of HDAs are weak sets of interval orders, and that
any \emph{finite} weak set of interval orders may be generated by an
HDA.

In summary, the main contributions of this paper are as follows:
\begin{itemize}
\item New Definitions \ref{de:widepc} and \ref{de:precubi} of
  precubical sets as presheaves over a category $\dotsquare$.  This
  has linearly ordered sets as objects, and the morphism are pairs
  $(f, \epsilon)$ of a poset map $f$ and a function $\epsilon$ which
  partitions elements not in the image of $f$ into two classes.  This
  is similar to constructions in \cite{DBLP:conf/types/BezemCH13,
    DBLP:journals/apal/Awodey18}; the standard base category $\square$
  \cite{book/Grandis09} of precubical sets is uniquely isomorphic to
  the skeleton of $\dotsquare$.
\item The identification of a new subclass of \emph{event consistent}
  precubical sets in Definitions \ref{de:labelob} and \ref{de:evcons}
  and the introduction of universal events for such precubical
  sets.
\item The exposition of a bijection, in Definitions
  \ref{de:label-track} and \ref{de:pobj}, between interval-ordered
  ipomsets and HDA tracks.  The first of these definitions introduces
  the \emph{label} of a track in an HDA $X$, which forms the basis on
  which we define track-based languages of HDAs; the second defines the
  \emph{track object} $\pobj{P}$ pertaining to an interval ipomset
  $P$.  These notions unite in the important Proposition
  \ref{pr:Representability}: $P$ is contained in the language of $X$
  precisely if there is an HDA morphism from $\pobj{P}$ into $X$.
\item The notion of \emph{interval arrangement} of a directed path
  through the geometric realization of an HDA and the subsequent
  Definition \ref{de:label-dpath} of labels of directed paths.
\item Definition \ref{de:lang} of the language of an HDA, the closure
  properties (under binary union and parallel composition) in Theorems
  \ref{th:hdlang-union} and \ref{th:hdapara}, and Theorems
  \ref{th:bisim-lang} and \ref{th:stbis-lang} that bisimilarity
  implies language equivalence.  We expect that together with
  Proposition \ref{pr:Representability} this may form the basis of a
  theory of regular pomset languages, but leave this for future work.
\end{itemize}

\section{Precubical Sets and Higher Dimensional Automata}
\label{se:hda}

In this section we introduce precubical sets and HDAs, but we start
with order-theoretic definitions, mainly to fix notation.  We will
return to posets and posets with interfaces in the next section.

A \emph{poset} is a pair $( P, <)$ consisting of a set $P$ and a
\emph{strict} partial order $<$ on $P$.  We henceforth assume tacitly
that the set $P$ is finite.  For any \emph{alphabet}, \ie finite set
$\Sigma$, a \emph{pomset} is a triple $( P, <, \lambda)$ with
$( P, <)$ a poset and $\lambda: P\to \Sigma$ the \emph{labeling} of
$P$.  If the order is linear, \ie a total relation in which $x= y$,
$x< y$, or $y< x$ for all $x, y\in P$, then we will speak of
\emph{linear posets} and \emph{linear pomsets} (and generally denote
linear po(m)sets by $S$, $T$, $U$ instead of $P$, $Q$, $R$).

Elements $x, y\in P$ of a po(m)set $P$ are \emph{comparable} if
$x= y$, $x< y$, or $y< x$; otherwise they are \emph{incomparable},
denoted $x\incomp y$.
An element $x\in P$ is \emph{minimal} if there exists no $y\in P$ with
$y< x$; and \emph{maximal} if there is no $y\in P$ with $x< y$.

A subset $Q\subseteq P$ of a po(m)set $P$ is an \emph{antichain} if
its elements are pairwise incomparable.  A \emph{maximal antichain} is
one which is not a proper subset of any other antichain.  The sets of
minimal, respectively maximal elements of $P$ are both maximal
antichains.

A function $f: P\to Q$ between posets $P$, $Q$ is a \emph{poset
  map}\footnote{We use \emph{map} and \emph{morphism} interchangeably
  for the arrows in a category.}  if $x<_P y$ implies $f( x)<_Q f( y)$
for all $x, y\in P$.  By irreflexivity, $f$ is \emph{injective} on
comparable elements: if $x<_P y$ or $y<_P x$, then $f( x)\ne f( y)$.
If $P$ is linearly ordered, then $f$ must be injective.

A function $f: P\to Q$ between pomsets $P$, $Q$ is a \emph{pomset map}
if it is a poset map that preserves the labeling, \ie
$\lambda_Q\circ f= \lambda_P$.  Posets and poset maps form the
category $\Pos$, and pomsets and pomset maps form the category
$\Poms$.  Isomorphism in these and all subsequent categories will be
denoted $\cong$.

\subsection{Precube Categories}

Precubical sets are usually defined as presheaves over a small
skeletal category $\square$, see Definition \ref{de:precubecat} below.
We find it more convenient to work with a large version of $\square$,
denoted $\dotsquare$ and defined below, which as objects has all
linear posets.

\begin{definition}
  \label{de:widepc}
  The \emph{large precube category} $\dotsquare$ consists of the
  following data:
  \begin{itemize}
  \item objects are linear posets $( S, \intord)$;
  \item morphisms $S\to T$ in $\dotsquare( S, T)$ are pairs
    $(f, \epsilon)$, where $f: S\to T$ is a poset map and
    $\epsilon: T\to \{0,\exec,1\}$ a function such that
    $f(S)= \epsilon^{-1}(\exec)$;
  \item the composition of morphisms $(f,\epsilon):S\to T$ and
    $(g,\zeta):T\to U$ is $(g\circ f,\eta)$, where
    \begin{equation*}
      \eta(u)=
      \begin{cases}
        \epsilon(g^{-1}(u)) & \text{for $u\in g(T)$}, \\
        \zeta(u) & \text{otherwise}.
      \end{cases}
    \end{equation*}
  \end{itemize}
\end{definition}

The function $\epsilon$ distinguishes events that have not yet started
(labelled by $0$) from those that have finished (labelled by $1$) and
those that are executing (labelled by $\exec$).  This notation is
inspired by Chu spaces \cite{pratt95chu}; see also
\cite{DBLP:journals/corr/FahrenbergJTZ18} for the relation between
HDAs and Chu spaces.  For every morphism $(f,\epsilon)$, the
isomorphism $f: S\to \epsilon^{ -1}( \exec)\subseteq T$ is unique; the
map $f$ is therefore determined by $\epsilon$.

\begin{figure}
  \centering
  \begin{tikzpicture}
    \begin{scope}
      \node[font=\normalsize] at (0,0) {$S$};
      \node (s1) at (0,-1) {$\bullet$};
      \node (s2) at (0,-2) {$\bullet$};
      \node (s3) at (0,-3) {$\bullet$};
      \path (s1) edge[dotted] (s2);
      \path (s2) edge[dotted] (s3);
    \end{scope}
    \begin{scope}[shift={(3,0)}]
      \node[font=\normalsize] at (0,0) {$T$};
      \node (t1) at (0,-1) {$\exec$};
      \node (t2) at (0,-2) {$\exec$};
      \node (t3) at (0,-3) {$0$};
      \node (t4) at (0,-4) {$\exec$};
      \path (t1) edge[dotted] (t2);
      \path (t2) edge[dotted] (t3);
      \path (t3) edge[dotted] (t4);
    \end{scope}
    \begin{scope}[shift={(6,0)}]
      \node[font=\normalsize] at (0,0) {$U$};
      \node (u1) at (0,-1) {$\exec$};
      \node (u2) at (0,-2) {$\exec$};
      \node (u3) at (0,-3) {$1$};
      \node (u4) at (0,-4) {$0$};
      \node (u5) at (0,-5) {$\exec$};
      \path (u1) edge[dotted] (u2);
      \path (u2) edge[dotted] (u3);
      \path (u3) edge[dotted] (u4);
      \path (u4) edge[dotted] (u5);
    \end{scope}
    \path (s1) edge (t1);
    \path (s2) edge (t2);
    \path (s3) edge (t4);
    \path (t1) edge (u1);
    \path (t2) edge (u2);
    \path (t3) edge (u4);
    \path (t4) edge (u5);
  \end{tikzpicture}
  \caption{Composition of morphisms $S\to T\to U$ in $\dotsquare$}
  \label{fi:dotsquare-comp}
\end{figure}
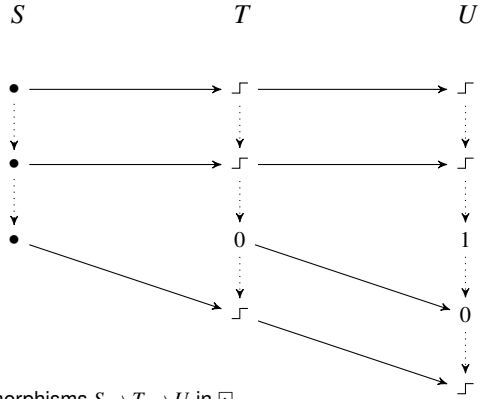

In an inclusion $( f, \epsilon): S\to T$, the events in
$f( S)= \epsilon^{ -1}( \exec)$ are executing, whereas the events in
$T\setminus f( S)$ are either not started or terminated.  In the
composition $( g\circ f, \eta)$, $\eta$ is defined such that events in
$U\setminus g( T)$ retain their status from the inclusion of $T$ in
$U$, events properly in $T$ preserve their status from the inclusion
of $S$ in $T$, and events coming from $S$ are executing.  See Figure
\ref{fi:dotsquare-comp} for an example.

For each $n\ge 1$ denote by $[ n]$ the linear poset
\begin{equation*}
  [ n]=\{ 1\intord 2\intord\dotsm\intord n\},
\end{equation*}
together with $[ 0]= \emptyset$.

\begin{definition}
  \label{de:precubecat}
  The \emph{precube category} is the full subcategory
  $\square\subseteq \dotsquare$ on objects $[n]$ for all $n\ge 0$.
\end{definition}

\begin{proposition}
  \label{pr:wide-cube}
  The category $\square$ is skeletal, and the inclusion
  $\square\subseteq \dotsquare$ is an equivalence of categories and
  admits a unique left inverse.
\end{proposition}

\begin{proof}
  It is clear that $\square$ contains no non-trivial isomorphisms,
  hence is skeletal.  For every $S\in\dotsquare$ of cardinality
  $|S|=n$ there is a unique isomorphism $\iota_S:S\to[ n]$ in
  $\dotsquare$.  Hence there is a unique functor
  $\rho:\dotsquare\to\square$, which is a left inverse of the
  inclusion $\square\subseteq \dotsquare$; it is given by
  $\rho(S)=[\,|S|\,]=[n]$ on objects and for any $(f,\epsilon):S\to T$
  by
  $\rho(f,\epsilon)=(\iota_T\circ f\circ \iota_S^{-1}, \epsilon\circ
  \iota_T^{-1})$ on morphisms.
\end{proof}

\begin{remark}
  The construction above mimics the situation for the base category of
  \emph{presimplicial sets}.  Let $\tPos$ be the full subcategory of
  $\Pos$ spanned by the linear posets and $\Delta\subseteq \tPos$ the
  full subcategory on objects $[ n]$ for $n\ge 0$.  Except for the
  maps being injective, $\Delta$ is the \emph{augmented simplex
    category}, see \cite[VII.5]{maclane-cwm}, and presheaves on
  $\Delta$, \ie functors from the opposite category $\Delta^\op$ into
  $\Set$, are presimplicial sets.  The category $\Delta$ is skeletal,
  and the inclusion $\Delta\subseteq \tPos$ is an equivalence of
  categories and admits a unique left inverse.  Consequently, the
  presheaf categories $\Set^{ \Delta^\op}$ and $\Set^{ \tPos^\op}$ are
  uniquely naturally isomorphic, and one may be used as drop-in
  replacement of the other.  See \cite{nlab:simplex_category} for more
  discussion on this subject.
\end{remark}

\begin{remark}
  A similar base category is introduced in
  \cite{DBLP:conf/types/BezemCH13} for cubical homotopy type theory,
  see also \cite{DBLP:journals/jar/BezemCH19,
    DBLP:journals/apal/Awodey18}.  Let $\mcal B$ be the category with
  objects linear posets and morphisms in $\mcal B( S, T)$ those
  functions $f: S\to T\sqcup\{ 0, 1\}$ (disjoint union) for which the
  restrictions $f\rest{ f^{ -1}(T)}$ to elements which do not map to
  $0$ or $1$ are poset isomorphisms.  Then $\mcal B= \dotsquare^\op$,
  as any $f\in \mcal B( S, T)$ is uniquely determined by
  $\epsilon: S\to\{ 0, \exec, 1\}$ given by
  \begin{equation*}
    \epsilon( x)=
    \begin{cases}
      \exec &\text{if } f( x)\in T,\\
      f(x) &\text{if } f(x)\in\{ 0, 1\}.
    \end{cases}
  \end{equation*}
  The category defined in \cite{DBLP:conf/types/BezemCH13} uses
  unordered sets and also permits morphisms $f$ for which
  $f\rest{ f^{ -1}(T)}$ is merely injective.  These two extensions are
  independent of each other; removing the order amounts to introducing
  \emph{symmetries}, and removing surjectivity equips precubical sets
  with degeneracies (thus passing to \emph{cubical} sets).  See
  \cite{GrandisM03-Site} for the presheaf categories of cubical and
  symmetric cubical sets.
\end{remark}

We proceed to show that the (reduced) precube category $\square$ is
isomorphic to the standard base category for precubical sets
\cite{book/Grandis09, GrandisM03-Site}.  For any $n\ge 1$, $i\in[ n]$,
and $\nu\in\{0,1\}$, define a $\square$-map
$d^\nu_{i,n}=( d_{i,n}, \epsilon^\nu_i):[ n- 1]\to[ n]$ by
\begin{equation*}
  d_{i,n}(k)=
  \begin{cases}
    k & \text{for $1\leq k<i$,}\\
    k+1 & \text{for $i\leq k\leq n-1$}
  \end{cases}
  \qquad\text{ and }\qquad
  \epsilon^\nu_i(k)=
  \begin{cases}
    \nu & \text{for $k=i$,}\\
    \exec & \text{for $k\neq i$.}
  \end{cases}
\end{equation*}

\begin{lemma}\label{le:PCIdElem}
  Let $(f,\epsilon)\in\square([m],[n])$, $m<n$, and let
  $s=\max\{i\in [n]\mid i\not\in f([m])\}$. Then
  $(f,\epsilon)=d^{\epsilon(s)}_{s,n}\circ (g,\zeta)$, where
  $(g,\zeta)\in\square([m],[n-1])$ is given by
  \begin{equation*}
    g(i)=
    \begin{cases}
      f(i) & \text{for $f(i)<s$,} \\
      f(i-1) & \text{for $f(i)>s$}
    \end{cases}
    \qquad\text{ and }\qquad
    \zeta(j)=
    \begin{cases}
      \epsilon(j) & \text{for $j<s$,}\\
      \epsilon(j+1) & \text{for $j\geq s$.}
    \end{cases}
  \end{equation*}
\end{lemma}

\begin{proof}
  Elementary calculations.
\end{proof}

\begin{lemma}\label{le:PCIdFull}
  Let $(f,\epsilon)\in \square([n-s], [n])$ and denote
  $[n]\setminus f([n-s])=\{a_1<\dots<a_{s}\}$. Then
  \begin{equation*}
    (f,\epsilon)=d^{\epsilon(a_{s})}_{a_s,n}\circ
    d^{\epsilon(a_{s-1})}_{a_{s-1},n-1}\circ \dotsm \circ
    d^{\epsilon(a_{2})}_{a_2,n-s+2}\circ
    d^{\epsilon(a_{1})}_{a_1,n-s+1}.
  \end{equation*}
\end{lemma}

\begin{proof}
  From Lemma \ref{le:PCIdElem} by induction.
\end{proof}

\begin{proposition}\label{prp:CatEq}
  The category $\square$ is generated by morphisms $d^\nu_{i,n}$ and
  the co-precubical identities
  $d^\nu_{j,n}\circ d^\mu_{i,n-1}= d^\mu_{i,n}\circ d^\nu_{j-1,n-1}$
  for $1\leq i<j<n$ and $\nu,\mu\in\{0,1\}$. Every $\square$-map
  $[n-s]\to [n]$ can be written uniquely as a composition
  \begin{equation*}
    d^{\nu_1,\dotsc,\nu_s}_{a_1,\dotsc,a_s}=d^{\nu_s}_{a_s,n}\circ
    d^{\nu_{s-1}}_{a_{s-1},n-1}\circ \dotsm \circ
    d^{\nu_2}_{a_2,n-s+2}\circ d^{\nu_{1}}_{a_1,n-s+1}
  \end{equation*}
  where $\nu_i\in\{0,1\}$ and $1\leq a_1<\dots<a_s\leq n$.
\end{proposition}

\begin{proof}
  Lemma \ref{le:PCIdFull} implies that every morphism
  $(f,\epsilon)\in \square([m], [n])$ can be presented as a
  composition of elementary morphisms in which the sequence of lower
  indices is strictly decreasing. Such presentations are unique since
  $f$ can be recovered from the sequences $a_{s},\dotsc,a_1$ and
  $\epsilon(a_s),\dots,\epsilon(a_1)$. It remains to show that the
  co-precubical relations hold, which is elementary.
\end{proof}

\subsection{Precubical Sets}
\label{se:precubical}

Precubical sets are usually defined as presheaves over $\square$, \ie
functors $\square^\op \to \Set$ \cite{book/Grandis09}.  Using
Proposition \ref{pr:wide-cube}, we may instead use presheaves over
$\dotsquare$:

\begin{proposition}
  \label{pr:precubidot}
  The presheaf categories $\Set^{ \square^\op}$ and
  $\Set^{ \dotsquare^\op}$ are uniquely naturally isomorphic.
\end{proposition}

\begin{proof}
  Each functor $F: \square^\op\to \Set$ extends uniquely to a functor
  $\dot F= \rho\circ F: \dotsquare^\op\to \Set$ by composition with the
  functor $\rho$ from the proof of Proposition \ref{pr:wide-cube}.
  The functor $\dot F$, in turn, restricts to $F$ on
  $\square^\op$.
\end{proof}

\begin{definition}
  \label{de:precubi}
  The category of \emph{precubical sets} is the presheaf category
  $\Set^{ \square^\op}$ or, equivalently, $\Set^{ \dotsquare^\op}$.
  That is, a precubical set is a functor $\square^\op\to\Set$ or
  $\dotsquare^\op\to\Set$, and a precubical map is a natural
  transformation of precubical sets.
\end{definition}

We write $X_n$ for $X([n])$,
$\delta^{\nu_1,\dotsc,\nu_s}_{\{a_1,\dotsc,a_s\}}$ for
$X(d^{\nu_1,\dotsc,\nu_s}_{a_1,\dotsc,a_s})$, and
$\delta^{\nu}_{\{a_1,\dotsc,a_s\}}$ for
$X(d^{\nu,\dotsc,\nu}_{a_1,\dotsc,a_s})$.  The map
$X(d^{\nu}_{i,n}):X_n\to X_{n-1}$ is denoted by $\delta^\nu_i$.  For
any $x\in X_n$, $n$ is called the \emph{dimension} of $x$ and
indicated by $\dim x= n$.  The maps $\delta^\nu_i$ are called
\emph{elementary face maps} and the maps
$\delta^{\nu_1,\dotsc,\nu_s}_{\{a_1,\dotsc,a_s\}}$, \emph{face maps}.

\begin{definition}
  \label{de:pobj-S}
  The \emph{standard $S$-cube} on a linear poset
  $( S, \mathord{\intord})$ is the precubical set $\square^S$, where
  \begin{itemize}
  \item $\square^S_k$ is the set of functions $x:S\to \{0,\exec,1\}$
    taking value $\exec$ on exactly $k$ elements;
  \item $\delta^\nu_i$ converts the $i$-th occurence of $\exec$ into
    $\nu\in\{ 0, 1\}$, \ie if
    $x^{-1}(\exec)=\{p_1\intord\dotsm\intord p_k\}$, then
    \begin{equation*}
      \delta^\nu_i x( p)=
      \begin{cases}
        \nu & \text{for $p=p_i$}, \\
        x( p) & \text{otherwise}.
      \end{cases}
    \end{equation*}
  \end{itemize}
\end{definition}

Every function $x:S\to \{0,\exec,1\}$ in $\square^S_k$
determines a unique poset map $f_x:[ k]\to S$ by the isomorphism
$f_x:[ k]\to f_x([ k])= x^{ -1}( \exec)$.
Denote the unique top-dimensional cell of $\square^S$ (the unique
element of $\square^S( S)$) by $\yoneda_S$. Then $\yoneda_S( x)=\exec$
for all $x\in S$. The order on $S$ is necessary to define face maps:
it determines which of the $\exec$ values should be converted into
$\nu$.

For $n\ge 0$, the \emph{standard $n$-cube} is
$\square^n:=\square^{[n]}$, and its unique $n$-cell is denoted by
$\yoneda_n$.

Regarded as a presheaf, $\square^S$ is the functor represented by $S$,
\ie $\square^S( T)=\dotsquare( T, S)$.  The cell $\yoneda_S$
corresponds to the identity morphism on $S$.  The following is an
immediate consequence of the Yoneda lemma.

\begin{lemma}
  \label{le:ineda}
  Let $X$ be a precubical set and $x\in X_n$. Then there exists a unique
  precubical map $\ineda_x: \square^n\to X$ such that
  $\ineda_x(\yoneda_n)=x$. \qed
\end{lemma}


\subsection{Labelings and Events}
\label{se:labelevent}

\begin{definition}
  \label{de:labelob}
  Let $A$ be a finite set.  The \emph{labeling object} on $A$ is the
  precubical set $\bang A$ with $\bang A_n= A^n$ and $\delta_i^\nu$
  defined by
  \begin{equation*}
    \delta^\nu_i(( a_1,\dotsc, a_n))=( a_1,\dotsc, a_{i-1},
    a_{i+1},\dotsc, a_n).
  \end{equation*}
  The \emph{event object} on $A$ is the precubical subset
  $\bbang A\subseteq \bang A$ given by
  \begin{equation*}
    \bbang A_n=\{( a_1,\dotsc, a_n)\mid a_i\ne a_j\text{ whenever }
  i\ne j\}.
  \end{equation*}
\end{definition}

Regarded as a presheaf, $\bang A( S)=\Set( S, A)$, hence $\bang A$ is
representable in $\Set$ via the forgetful functor
$\dotsquare\to \Set$.  In particular, $\bang A_n$ is exactly the set
of isomorphism classes of linear posets over $A$ with $n$ elements.
Similarly, $\bbang A( S)= \Inj( S, A)$, where $\Inj$ is the category
of sets and \emph{injective} maps, so that $\bbang A$ is representable
in $\Inj$ via the forgetful functor $\dotsquare\to \Inj$.  Also note
that $\bang A$ is infinite, whereas $\bbang A$ is finite: if $A$ has
$m$ elements, then $\bbang A_m$ consists of all permutations of these
elements and $\bbang A_n= \emptyset$ for $n> m$.

Every function $f\in \Set( A, B)$ induces a precubical map
$\bang f: \bang A\to \bang B$, and every injective function
$g\in \Inj( A, B)$ induces a precubical map
$\bbang g: \bbang A\to \bbang B$, turning them into functors
$\bang: \Set\to \pcSet$ and $\bbang: \Inj\to \Inj^{ \square^\op}$.
These are left adjoint to the functors $\pcSet\to \Set$ and
$\Inj^{ \square^\op}\to \Inj$ mapping $X$ to $X_1$, hence $\bang A$
and $\bbang A$ are \emph{free} in the following sense.

\begin{lemma}
  \label{le:bang-free}
  Let $X$ be a precubical set and $A$ a finite set.
  \begin{enumerate}
  \item Any function $\lambda_1: X_1\to A$ for which
    $\lambda_1( \delta_1^0 x)= \lambda_1( \delta_1^1 x)$ and
    $\lambda_1( \delta_2^0 x)= \lambda_1( \delta_2^1 x)$ for all
    $x\in X_2$ extends uniquely to a precubical map
    $\lambda: X\to \bang A$.
  \item Any function $\ev_1: X_1\to A$ for which
    $\ev_1( \delta_1^0 x)= \ev_1( \delta_1^1 x)$,
    $\ev_1( \delta_2^0 x)= \ev_1( \delta_2^1 x)$, and
    $\ev_1( \delta_1^0 x)\ne \ev_1( \delta_2^0 x)$ for all $x\in X_2$
    extends uniquely to a precubical map $\ev: X\to \bbang A$.
  \end{enumerate}
\end{lemma}

\begin{proof}
  For the first claim, define functions $f_i:[ 1]\to[ n]$, for all
  $n\ge 0$ and $i\in[ n]$, by $f_i( 1)= i$.  Then define $\lambda$ by
  $\lambda( x)=( \lambda_1(( f_1, \epsilon_1)^\op( x)),\dotsc,
  \lambda_1(( f_n, \epsilon_n)^\op( x)))$ for $x\in X_n$.  Because of
  $\lambda_1\circ \delta_i^0= \lambda_1\circ \delta_i^1$, the choices
  of $\epsilon_i$ do not matter.  It is clear that $\lambda$ is the
  unique extension of $\lambda_1$.

  For the second claim, we already know that $\ev_1$ extends uniquely
  to $\ev: X\to \bang A$. We show that the image of $\ev$ lies in
  $\bbang A$.  With a slight abuse of notation, write
  $\ev( x)=( \ev_1( x),\dotsc, \ev_n( x))$ for $x\in X_n$, and suppose
  there exists an $x\in X_n$ with $\ev_i(x)= \ev_j(x)$ for $i<j$.  Let
  $(f, \epsilon)\in \square([2],[n])$ be the morphism that satisfies
  $f( 1)= i$ and $f( 2)= j$ ($\epsilon$ is again irrelevant), and let
  $y=( f, \epsilon)^\op( x)\in X_2$. Then
  $\ev_1( \delta_1^0 y)= \ev_i( x)= \ev_j( x)= \ev_1( \delta_2^0 y)$,
  in contradiction to the second property of $\ev_1$.
\end{proof}

Let henceforth $\Sigma$ be a fixed finite set.

\begin{definition}
  Let $X$ be a precubical set.  A \emph{labeling} of\/ $X$ is a
  precubical map $\lambda: X\to \bang \Sigma$.  An \emph{event
    identification} on $X$ is a map $\ev: X\to \bbang \Sigma$.
\end{definition}

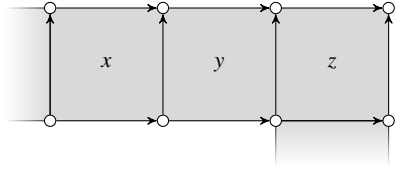
\begin{figure}
  \centering
  \begin{tikzpicture}[x=1.5cm, y=1.5cm]
    \begin{scope}
      \path[fill=black!15] (0,0) to (3,0) to (3,1) to (0,1);
      \filldraw[-, fill=black!15, path fading=west] (-.4,0) -- (0,0)
      -- (0,1) -- (-.4,1);
      \filldraw[-, fill=black!15, path fading=south] (2,-.4) -- (2,0)
      -- (3,0) -- (3,-.4);
      \foreach \x in {0, 1, 2, 3} \foreach \y in {0, 1} \node[state] (\x\y)
      at (\x,\y) {};
      \foreach \x in {0, 1, 2, 3} \path (\x0) edge (\x1);
      \foreach \y in {0, 1} \path (0\y) edge (1\y);
      \foreach \y in {0, 1} \path (1\y) edge (2\y);
      \foreach \y in {0, 1} \path (2\y) edge (3\y);
      \node at (.5,.5) {$\vphantom{y}x$};
      \node at (1.5,.5) {$y$};
      \node at (2.5,.5) {$\vphantom{y}z$};
    \end{scope}
  \end{tikzpicture}
  \bigskip
  \caption{Example of a precubical set that is not event consistent
    (left and bottom right edges identified).}
  \label{fi:pcsselflev}
\end{figure}

Every event identification on $X$ is also a labeling on $X$, but the
converse does not hold; in fact, all precubical sets admit labelings,
but not necessarily event identifications:

\begin{example}
  \label{ex:pcsselflev}
  Figure \ref{fi:pcsselflev}\, shows a precubical set with three
  $2$-cells $x, y, z$ and $\delta_1^0 x= \delta_2^0 z$.  Any event
  identification $\ev: X\to \bbang \Sigma$ must fulfill
  \begin{equation*}
    \ev( \delta_2^0 z)= \ev( \delta_1^0 x)= \ev( \delta_1^1 x)= \ev(
    \delta_1^0 y)= \ev( \delta_1^1 y)= \ev( \delta_1^0 z),
  \end{equation*}
  a contradiction.
\end{example}

\begin{definition}
  \label{de:evcons}
  A precubical set $X$ is \emph{event consistent} if it admits an event
  identification $\ev: X\to \bbang \Sigma$.
\end{definition}

\begin{lemma}
  \label{le:evcons}
  A precubical set $X$ is event consistent iff there exists an
  equivalence relation $\sim$ on $X_1$ such that for all $x\in X_2$,
  $\delta_1^0 x\sim \delta_1^1 x$, $\delta_2^0 x\sim \delta_2^1 x$,
  and $\delta_1^0 x\not\sim \delta_2^0 x$.
\end{lemma}

\begin{proof}
  First suppose that $X$ is event consistent. Let
  $\ev: X\to \bbang \Sigma$ and define the equivalence relation
  $\mathord{\sim}$ on $X_1$ by $x\sim y$ iff $\ev( x)= \ev( y)$.  
  From the definition of $\bang \Sigma$ we have
  $\ev( \delta_i^0 x)= \delta_i^0 \ev( x)= \delta_i^1 \ev( x)= \ev(
  \delta_i^1 x)$ and therefore $\delta_i^0 x\sim \delta_i^1 x$ holds
  for all $x\in X_2$ and for $i\in\{ 1, 2\}$.  From the definition of
  $\bang\bang \Sigma$ we have
  $\ev( \delta_1^0 x)= \delta_1^0 \ev( x)\ne \delta_2^0 \ev( x)= \ev(
  \delta_2^0 x)$ and therefore $\delta_1^0 x\not\sim \delta_2^0 x$
  holds for all $x\in X_2$.
  
  For the other direction, suppose there is a relation
  $\mathord{\sim}$ that satisfies the properties in the lemma and let
  $\Sigma= X_1/ \mathord{\sim}$ be the set of equivalence classes.
  The quotient map $X_1\to \Sigma$ then extends uniquely to
  $\ev: X\to \bbang \Sigma$ by Lemma \ref{le:bang-free}, which yields
  the event identification needed.
\end{proof}

Any event consistent precubical set admits a \emph{smallest}
equivalence relation owing to Lemma \ref{le:evcons}, denoted $\eveq$.
It is given as the transitive closure of
$\{(\delta_i^0 x, \delta_i^1 x)\mid x\in X_2, i\in\{ 1, 2\}\}$, we
call its equivalence classes the \emph{universal events} of $X$.

We will generally only concern ourselves with event consistent
precubical sets in the rest of this work, but come back to the more
general case at the end of Section \ref{se:langlang}.

\begin{example}
  \label{ex:scube-evcont}
  The standard $S$-cube $\square^S$ from Definition \ref{de:pobj-S}
  is event consistent for any linear poset $S$.  Its universal event
  equivalence is given by $x\eveq y$ iff
  $x^{ -1}( \exec)= y^{ -1}( \exec)$, induced by the event
  identification $\ev: \square^S\to \bbang S$ with
  $\ev( x)= x^{ -1}( \exec)$.  If $X$ is a precubical subset of
  $\square^S$,
  then $X$ is also event consistent.  Such precubical subsets of
  standard cubes are called \emph{sculptures} in
  \cite{DBLP:journals/corr/FahrenbergJTZ18}, where it is shown that
  they correspond to Chu spaces over $\{ 0, \exec, 1\}$
  \cite{pratt95chu, Pratt00Sculptures}.
\end{example}

The term ``universal events'' is justified by the following
factorization property, which follows immediately from the
definitions.

\begin{proposition}
  \label{prp:EventsAreUniversalLabeling}
  Every labeling $\lambda:X\to \bang \Sigma$ factors uniquely through
  $E_X$, \ie there is a unique \mbox{factorization}
  $\lambda_1= \lambda^\ev_1\circ \ev_1:X_1\to E_X\to \Sigma$ that
  extends to a factorization
  $\lambda= \lambda^\ev\circ \ev: X\to \bbang E_X\to \bang \Sigma$.
  \qed
\end{proposition}

Hence also any event identification factors uniquely through the
universal events.

\begin{proposition}
  \label{prp:LabelsOfCells}
  Let $\lambda: X\to \bang \Sigma$ be a labeling.  If $x\in X_n$ and
  $(f,\epsilon)\in \square([m], [n])$, then
  $\lambda((f,\epsilon)^\op( x))= \lambda(x)\circ f$. In particular,
  \begin{equation*}
    \lambda( \delta^\nu_i x)=( \lambda_1(x),\dotsc, \lambda_{i-1}(x),
    \lambda_{i+1}(x),\dotsc, \lambda_n(x)).
  \end{equation*}
\end{proposition}

\begin{proof}
  Straightforward from the definitions.
\end{proof}

\begin{lemma}
  \label{le:SemiNSI}
  Let $X$ be an event consistent precubical set, $n\ge 0$, $x\in X_n$,
  $A,B\subseteq [n]$, and $\nu\in\{ 0, 1\}$. Then
  $\delta^\nu_A x=\delta^\nu_B x$ implies $A= B$.
\end{lemma}

\begin{proof}
  Applying Proposition \ref{prp:LabelsOfCells} to
  $\ev: X\to \bbang E_X$ yields
  $\ev(\delta^\nu_A x)=(\ev_i(x))_{i\in [n]\setminus A}$ and
  $\ev(\delta^\nu_B x)=(\ev_i(x))_{i\in [n]\setminus B}$. Since events
  $\ev_i(x)$ are pairwise distinct, we obtain $A=B$.
\end{proof}

\begin{remark}
  For $n= 2$, the above lemma reduces to the definition of event
  consistency.  The lemma will show its importance once we consider
  \emph{tracks}, \ie sequences of cells connected at faces, in Section
  \ref{se:tracks}.  There is a related property of being
  \emph{non-selfintersecting} which has been used for the same
  purpose, see for example \cite{Fajstrup05-cubcomp}: a precubical set
  $X$ is non-selfintersecting if $\delta_A^\nu x= \delta_B^\mu x$
  implies $A= B$ and $\nu= \mu$ for all $x\in X$.  Example
  \ref{ex:pcsselflev} shows that precubical sets may be
  non-selfintersecting, but not event consistent; similarly, event
  consistency does not imply the non-selfintersecting property, see
  \cite{DBLP:journals/corr/FahrenbergJTZ18}.  Finally, also Section 4
  of \cite{Glabbeek96-pg} contains some precursors to our notion of
  event consistency.
\end{remark}

If $f:X\to Y$ is a precubical map, then $x\eveq y$ implies
$f(x)\eveq f(y)$ for $x,y\in X_1$, and  $f_1:X_1\to Y_1$ induces a
map $E_f:E_X\to E_Y$. This defines a functor $E:\pcSet\to
\Set$.

\subsection{Higher-Dimensional Automata}

Higher-dimensional automata are labeled precubical sets with initial
and accepting cells.

\begin{definition}
  A \emph{labeled precubical set} is an event consistent precubical
  set $X$ together with a labeling $\lambda: X\to \bang \Sigma$.  Maps
  $f: X\to Y$ of labeled precubical sets preserve labelings:
  $\lambda_X= \lambda_Y\circ f$.
\end{definition}

That is to say, the category of labeled precubical sets is the full
subcategory of the slice category $\pcSet/ \bang \Sigma$ on event
consistent objects.

\begin{definition}
  Let $( X, \lambda)$ be a labeled precubical set.  The \emph{label}
  of a cell $x\in X_n$ is the linear pomset
  $\ell( x)=( \ev( x), \mathord{\intord}, \lambda_{ \ev(x)})$ with
  \begin{equation*}
    \ev(x)=(\ev_1(x)\intord\ev_2(x)\intord\dotsm \intord\ev_n(x))
  \end{equation*}
  and $\lambda_{\ev(x)}(\ev_i(x))=\lambda_i(x)$.
\end{definition}

Note that $\ell( x)$ is basically the tuple $\lambda( x)$, but using
Proposition \ref{prp:EventsAreUniversalLabeling} we now regard it as a
pomset.  Event consistency is essential for this to make sense.  Of
the two following elementary lemmas, the first follows from
Proposition \ref{prp:LabelsOfCells} and Lemma \ref{le:evcons}; the
second one is trivial.

\begin{lemma}
  \label{le:EventsAreSubsets}
  If $x$ is a face of $y$, then there is a pomset inclusion
  $\ell(x)\subseteq \ell(y)$.  \qed
\end{lemma}

\begin{lemma}
  \label{le:CellParmsetFun}
  Let $f:X\to Y$ be a map of labeled precubical sets and $x\in
  X_n$. Then $\ell(f(x))=E_f(\ell(x))\cong \ell(x)$. \qed
\end{lemma}

\begin{definition}
  \label{de:hda}
  A \emph{higher-dimensional automaton} (HDA) is a tuple
  $( X, I, F, \lambda)$, where $( X, \lambda)$ is a labeled precubical
  set, $I\subseteq X$ is a set of \emph{initial cells} and
  $F\subseteq X$ a set of \emph{accepting cells}.  If $X$ and $Y$ are
  HDAs, then a precubical map $f:X\to Y$ is an \emph{HDA map} if it
  preserves labels and initial and accepting cells:
  $\lambda_X= \lambda_Y\circ f$, $f(I_X)\subseteq I_Y$, and
  $f(F_X)\subseteq F_Y$.
\end{definition}




\section{Pomsets with Interfaces}
\label{se:ipoms}

We now return to posets and pomsets and introduce interfaces for them,
building on our work in \cite{DBLP:conf/RelMiCS/FahrenbergJST20} but
enriched with event orders.  Recall that $\Sigma$ is a fixed finite
set.

\subsection{Ipomsets}


\begin{definition}
  \label{de:ipomset}
  An \emph{ipomset} is a tuple
  $(P, \mathord<_P, \mathord{\intord}_P, \lambda_P, S_P, T_P)$, where
  $P$ is a finite set;
  \begin{itemize}
  \item $<_P$ is a strict partial order on $P$ called
    \emph{precendence order};
  \item $\intord_P$ is a strict partial order on $P$ called
    \emph{event order};
  \item $\lambda_P: P\to \Sigma$ is a function called \emph{labeling};
  \item $S_P$ is a subset of the $<$-minimal elements of $P$ called
    \emph{source set};
  \item $T_P$ is a subset of the $<$-maximal elements of $P$ called
    \emph{target set}.
  \end{itemize}
  We require that the relation
  $\mathord{ <_P}\cup \mathord{ \intord_P}$ is total: 
  if $x\neq y\in P$,
  then $x$ and $y$ are comparable by $\mathord{ <_P}$  or by ${ \intord_P}$.
\end{definition}

Note that $\mathord{<_P}\cup \mathord{\intord_P}$ need not be a
partial order, see Example \ref{ex:ordernotorder} below.  The linear
pomset
$( S_P, \mathord{\intord_P}\cap( S_P\times S_P),
\lambda_P\rest{S_P})$, where $\lambda_P\rest{S_P}$ denotes the domain
restriction of $\lambda_P: P\to \Sigma$ to $S_P$, is called
\emph{source interface} of $P$; we often simply write
$( S_P, \intord_P, \lambda_P)$.  Similarly,
$( T_P, \intord_P, \lambda_P)$ is the \emph{target interface} of $P$.
If $S= T= \emptyset$, then $P$ is a pomset in the classical sense
\cite{Pratt86pomsets, DBLP:journals/tcs/Gischer88} (ignoring the event
order).  If $S= T= \emptyset$ and $<$ is linear, then $P$ corresponds
to a string.

\begin{remark}
  In \cite{DBLP:conf/RelMiCS/FahrenbergJST20} we defined ipomsets
  without an event order.  Instead we picked out sources and targets
  using injections $s:[ n]\to P$ and $t:[ m]\to P$.  This implicitly
  defines (linear) event orders on the subsets
  $S_P= s([ n])\subseteq P$ and $T_P= t([ m])\subseteq P$, so the only
  essential difference between
  \cite{DBLP:conf/RelMiCS/FahrenbergJST20} and our present setting is
  that the event order is extended to the whole of $P$.  When
  $S_P\cap T_P= \emptyset$, such an extension is always possible; and
  we will see later that the ordered structures properly corresponding
  to HDAs are our present event-ordered ipomsets, see Definitions
  \ref{de:label-track} and~\ref{de:pobj}.
\end{remark}

\begin{remark}
  The ordered structure $( P, \mathord{ <_P}, \mathord{ \intord_P})$
  underlying an ipomset is a \emph{biposet} in the sense of
  \cite{DBLP:journals/jalc/EsikN04}, but because of the requirement
  that $\mathord{ <_P}\cup \mathord{ \intord_P}$ be total, not all
  biposets may be used.  \cite{DBLP:journals/jalc/EsikN04} is
  concerned with $n$-posets, \ie finite sets with $n$ partial orders,
  and then introduces a notion of higher-dimensional automata as
  recognizers of such structures.  Except for the name, the
  higher-dimensional automata of \cite{DBLP:journals/jalc/EsikN04}
  have nothing to do with our HDAs.
\end{remark}

\begin{example}
  \label{ex:ordernotorder}
  By definition, the maximal antichains of the precedence order are
  linearly ordered by the event order, but the event order may contain
  further arrows.  As an example, consider the ipomset
  \begin{equation*}
    P= \ipomset{ & a \ar[r] & b \ar@{.>}[dl] & \\ & c \ar@{.>}[u]}
  \end{equation*}
  with precedence order $a< b$ and event order $b\intord c\intord a$.
  Both maximal antichains $a\incomp c$ and $b\incomp c$ are linearly
  $\intord$-ordered, but by transitivity, also $b\intord a$.
\end{example}

Let $Q\subseteq P$ be a subset of the ipomset
$( P, <, \intord, \lambda, S, T)$. Then the restriction
\begin{equation*}
  P\rest{ Q}:=(Q, \mathord<\cap (Q\times Q), \mathord{\intord}\cap
  (Q\times Q), \lambda\rest{Q}, S\cap Q, T\cap Q)
\end{equation*}
is also an ipomset.

\begin{definition}
  Ipomsets $P$ and $Q$ are \emph{isomorphic} if there exists a
  bijection $f: P\to Q$ (an \emph{ipomset isomorphism}) that
  \begin{itemize}
  \item respects precedence: for all $x, y\in P$, $x<_P y$ iff
    $f(x)<_Q f(y)$;
  \item respects essential event ordering: for all $x, y\in P$
    with $x\incomp_P y$, $x\intord_P y$ iff $f(x)\intord_Q f(y)$;
  \item respects labels and interfaces:
    $\lambda_Q\circ f=\lambda_P$, $f(S_P)=S_Q$, and
    $f(T_P)=T_Q$.
  \end{itemize}
\end{definition}

By definition, $f$ is only required to respect the part of the event
ordering which orders events in antichains.  In Section
\ref{se:subsu} we will introduce a notion of morphism between ipomsets for
which the above form the isomorphisms.

Isomorphisms between ipomsets are unique:

\begin{lemma}
  \label{le:CanonicalIso}
  There is at most one isomorphism between any two
  ipomsets.
\end{lemma}

\begin{proof}
  Using poset filtrations, we can combine the two orders on an ipomset
  into a linear order and then use the fact that isomorphisms between
  linearly ordered sets are unique:

  Let $P$ be an ipomset and $P_0$ its set of $<_P$-minimal elements.
  Let $P_1$ be the set of $<_P$-minimal elements of the sub-ipomset
  $P\setminus P_0$, $P_2$ the set of $<_p$-minimal elements of
  $P\setminus P_0\setminus P_1$, and so on.  The finite disjoint union
  $P= P_0\sqcup P_1\sqcup P_2\sqcup\dotsc$ is called \emph{filtration}
  of $P$.  (More precisely, one can set $P_{ > -1}= P$ and then
  inductively for $i\ge 0$, until exhaustion, let $P_i$ be the
  $<_P$-minimal elements of $P_{ > i- 1}$ and
  $P_{ > i}= P_{ > i- 1}\setminus P_i$.)

  Now all $P_i$ are $<_P$-antichains and hence linearly ordered by
  $\intord_P$.  Let $\prec_P$ be the relation on $P$ defined by
  $x\prec_P y$ if $x\in P_i$ and $y\in P_j$ for $i< j$, or
  $x, y\in P_i$ for some common $i$ and $x\intord_P y$.  Then
  $\prec_P$ is a linear order on $P$.  Further, if $f: P\to Q$ is an
  ipomset isomorphism, then
  $f:( P, \mathord{ \prec_P})\to( Q, \mathord{ \prec_Q})$ is an
  isomorphism of linear orders; hence $f$ is unique.
\end{proof}

\subsection{Gluing and Parallel Compositions}

The \emph{gluing composition} $P* Q$ of two ipomsets is defined if the
target interface of $P$ is isomorphic to the source interface of $Q$,
in which case it identifies the targets of $P$ with their
corresponding sources in $Q$ and makes all non-interface elements in
$P$ precede all non-interface elements in $Q$.  Below, $(\cdot)^+$ is
used for transitive closure.

\begin{definition}
  Let $P$ and $Q$ be ipomsets such that
  $( T_P, \mathord{\intord_P}, \lambda_P)$ is isomorphic to
  $( S_Q, \mathord{\intord_Q}, \lambda_Q)$.  The \emph{gluing
    composition} of $P$ and $Q$ is
  $P* Q=( P\sqcup( Q\setminus S_Q), \mathord<, \mathord{\intord},
  \lambda, S_P, T_Q)$ with $<$, $\intord$, and $\lambda$ defined as
  follows:
  \begin{equation*}
    \begin{aligned}
      \mathord< &= \mathord{<_P}\cup \mathord{<_Q}\cup( P\setminus
      T_P)\times( Q\setminus S_Q), \\
      \mathord{\intord} &= ( \mathord{\intord_P}\cup
      \mathord{\intord_Q})^+,
    \end{aligned}
    \qquad\quad
    \lambda( x)=
    \begin{cases}
      \lambda_P( x) &\text{if } x\in P, \\
      \lambda_Q( x) &\text{if } x\in Q.
    \end{cases}
  \end{equation*}
\end{definition}

It is clear that $P* Q$, if defined, is indeed again an ipomset: the
only non-trivial property to check is irreflexivity of $\intord$,
which follows from the fact that the restrictions of $\intord$ to $P$
and $Q$ are precisely $\intord_P$ and $\intord_Q$, respectively.
It is also clear that gluing composition
respects isomorphisms:

\begin{lemma}
  \label{le:glue-iso}
  If $P\cong P'$ and $Q\cong Q'$, then $P* Q$ is defined iff
  $P'* Q'$ is, and in that case, $P* Q\cong P'* Q'$.  \qed
\end{lemma}

Gluing composition of \emph{pomsets}, \ie ipomsets with empty
interfaces $S_P= T_P= \emptyset$, is the same as the standard
\emph{serial composition} \cite{DBLP:journals/fuin/Grabowski81,
  DBLP:journals/tcs/Gischer88}.  Gluing composition of \emph{strings}
is concatenation.
We also introduce a parallel composition of ipomsets that generalizes
the eponymous operation for pomsets.

\begin{definition}
  The \emph{parallel composition} of ipomsets $P$ and $Q$ is
  $P\para Q=( P\sqcup Q, <, \intord, \lambda,$ \linebreak 
  $S_P\sqcup S_Q, T_P\sqcup
  T_Q)$ with $<$, $\intord$, and $\lambda$ defined as
  follows:
  \begin{equation*}
    \begin{aligned}
      \mathord< &= \mathord{<_P}\cup \mathord{<_Q}, \\
      \mathord{\intord} &= \mathord{\intord_P}\cup
      \mathord{\intord_Q}\cup P\times Q,
    \end{aligned}
    \qquad\quad
    \lambda( x)=
    \begin{cases}
      \lambda_P( x) &\text{if } x\in P, \\
      \lambda_Q( x) &\text{if } x\in Q.
    \end{cases}
  \end{equation*}
\end{definition}

It is easy to see that $P\para Q$ is again an ipomset.  Parallel
composition of ipomsets is not commutative because of the event order.
It is again clear that parallel composition respects isomorphisms:

\begin{lemma}
  If $P\cong P'$ and $Q\cong Q'$, then $P\para Q\cong P'\para Q'$.  \qed
\end{lemma}

\subsection{Interval Orders}
\label{se:interval}

An \emph{interval order} is a poset $P$ in which $x< z$ and $y< w$
imply $x< w$ or $y< z$ for all $x, y, z, w\in P$.

\begin{lemma}[\cite{journals/mpsy/Fishburn70, book/Fishburn85,
    DBLP:journals/tcs/JanickiK93}]
  \label{le:intord}
  The following are equivalent for any poset $P$:
  \begin{enumerate}
  \item $P$ is an interval order;
  \item $P$ does not contain an induced subposet
    $\twotwo= \pomset{\bullet \ar[r] & \bullet \\ \bullet \ar[r] & \bullet}$;
  \item $P$ has an \emph{interval representation}: a pair of functions
    $s, t: P\to Q$ into a linear poset $(Q, <_Q)$ such that for all
    $x, y\in P$, $s( x)<_Q t( x)$, and $x<_P y$ iff $t( x)<_Q s( y)$;
  \item \label{en:intord.maxanti} the order $\prec$ on maximal
    antichains of $P$ defined by $X\prec Y$ if $X\ne Y$ and
    $y\not<_P x$ for all $x\in X$, $y\in Y$ is linear.
  \end{enumerate}
\end{lemma}

\begin{definition}
  \label{def_evilio}
  An \emph{interval ipomset} is an ipomset $P$ for which the
  underlying precedence poset $( P, \mathord<_P)$ is an interval
  order.
\end{definition}

\begin{figure}
  \centering
  \begin{tikzpicture}[label distance=-.2cm,shorten <=-3pt, shorten >=-3pt]
    \begin{scope}[->, >=latex', xscale=.4, yscale=.75]
      \begin{scope}
        \node (1) at (1.5,0) {$\vphantom{b}a$};
        \node [label=right:{\tiny 1}] (2) at (4.5,0) {$b$};
        \node (3) at (1.5,-1) {$\vphantom{b}c$};
        \node (4) at (4.5,-1) {$d$};
        \node [label=right:{\tiny 2}] (5) at (4.5,-2) {$\vphantom{b}e$};
        \foreach \i/\j in {1/2,3/2,3/4,3/5} \path (\i) edge (\j);
        \node at (6.5,-1) {$*$};
      \end{scope}
      \begin{scope}[xshift=7.5cm]
        \node [label=left:{\tiny 1}]  (6) at (.5,0) {$b$};
        \node (7) at (3.5,0) {$\vphantom{b}\smash[b]g$};
        \node [label=left:{\tiny 2}] (8) at (.5,-2) {$\vphantom{b}e$};
        \node (9) at (3.5,-2) {$\smash[b]f$};
        \foreach \i/\j in {6/7,8/7,8/9} \path (\i) edge (\j);
        \node at  (5.7,-1) {$=$};
      \end{scope}
      \begin{scope}[xshift=16cm]
        \node (1) at (-.5,0) {$\vphantom{b}a$};
        \node (2) at (4.5,0) {$b$};
        \node (3) at (-.5,-1) {$\vphantom{b}c$};
        \node (4) at (2.5,-1) {$d$};
        \node (5) at (4.5,-2) {$\vphantom{b}e$};
      \end{scope}
      \begin{scope}[xshift=20cm]
        \node (6) at (.5,0) {\phantom{$b$}};
        \node (7) at (5.5,0) {$\vphantom{b}\smash[b]{g}$};
        \node (8) at (.5,-2) {\phantom{$\vphantom{b}e$}};
        \node (9) at (5.5,-2) {$\smash[b]{f}$};
      \end{scope}
      \foreach \i/\j in {1/2,3/2,3/4,3/5} \path (\i) edge (\j);
      \foreach \i/\j in {6/7,8/7,8/9} \path (\i) edge (\j);
      \foreach \i/\j in {1/9,4/7,4/9} \path (\i) edge (\j);
    \end{scope}
    \begin{scope}[-, xscale=.4, yshift=-3cm, shorten <=-4.5pt,
      shorten >=-4.5pt, yscale=.9]]
      \begin{scope}
        \node (1l) at (0,0) {{\tiny $|$}};
        \node (1r) at (3,0) {{\tiny $|$}};
        \node (2l) at (4,0) {{\tiny $|$}};
        \node (2r) at (5,0) {{\tiny $|$}};
        \node (3l) at (0,-1) {{\tiny $|$}};
        \node (3r) at (1,-1) {{\tiny $|$}};
        \node (4l) at (2,-1) {{\tiny $|$}};
        \node (4r) at (4.6,-1) {{\tiny $|$}};
        \node (5l) at (2,-2) {{\tiny $|$}};
        \node (5r) at (5,-2) {{\tiny $|$}};
        \path (1l) edge node[above] {\small $I(a)$}  (1r);
        \path (2l) edge node[above] {\small $I(b)$}  (2r);
        \path (3l) edge node[above] {\small $I(c)$}  (3r);
        \path (4l) edge node[above] {\small $I(d)$}  (4r);
        \path (5l) edge node[above] {\small $I(e)$}  (5r);
        \node at (6,-1) {$*$};
      \end{scope}
      \begin{scope}[xshift=7cm]
        \node (6l) at (0,0) {{\tiny $|$}};
        \node (6r) at (3,0) {{\tiny $|$}};
        \node (7l) at (4,0) {{\tiny $|$}};
        \node (7r) at (5,0) {{\tiny $|$}};
        \node (8l) at (0,-2) {{\tiny $|$}};
        \node (8r) at (1,-2) {{\tiny $|$}};
        \node (9l) at (2,-2) {{\tiny $|$}};
        \node (9r) at (5,-2) {{\tiny $|$}};
        \path (6l) edge node[above] {\small $I(b)$}  (6r);
        \path (7l) edge node[above] {\small $I(g)$}  (7r);
        \path (8l) edge node[above] {\small $I(e)$}  (8r);
        \path (9l) edge node[above] {\small $I(f)$}  (9r);
        \node at  (6.5,-1) {$=$};
      \end{scope}
      \begin{scope}[xshift=15.2cm]
        \node (1l') at (0,0) {{\tiny $|$}};
        \node (1r') at (3,0) {{\tiny $|$}};
        \node (2/6l') at (4,0) {{\tiny $|$}};
        \node (3l') at (0,-1) {{\tiny $|$}};
        \node (3r') at (1,-1) {{\tiny $|$}};
        \node (4l') at (2,-1) {{\tiny $|$}};
        \node (4r') at (4.6,-1) {{\tiny $|$}};
        \node (5/8l') at (2,-2) {{\tiny $|$}};
      \end{scope}
      \begin{scope}[xshift=21cm]
        \node (2/6r') at (3,0) {{\tiny $|$}};
        \node (7l') at (4,0) {{\tiny $|$}};
        \node (7r') at (5,0) {{\tiny $|$}};
        \node (5/8r') at (1,-2) {{\tiny $|$}};
        \node (9l') at (2,-2) {{\tiny $|$}};
        \node (9r') at (5,-2) {{\tiny $|$}};
      \end{scope}
      \path (1l') edge node[above] {\small $I(a)$}  (1r');
      \path (2/6l') edge node[above] {\small $I(b)$}  (2/6r');
      \path (3l') edge node[above] {\small $I(c)$}  (3r');
      \path (4l') edge node[above] {\small $I(d)$}  (4r');
      \path (5/8l') edge node[above] {\small $I(e)$}  (5/8r');
      \path (7l') edge node[above] {\small $I(g)$}  (7r');
      \path (9l') edge node[above] {\small $I(f)$}  (9r');
    \end{scope}
  \end{tikzpicture}
  \bigskip\medskip
  \caption{Two interval ipomsets and their gluing: above as ipomsets,
    below using interval representations (event order not shown)}
  \label{fi:intcomp}
\end{figure}
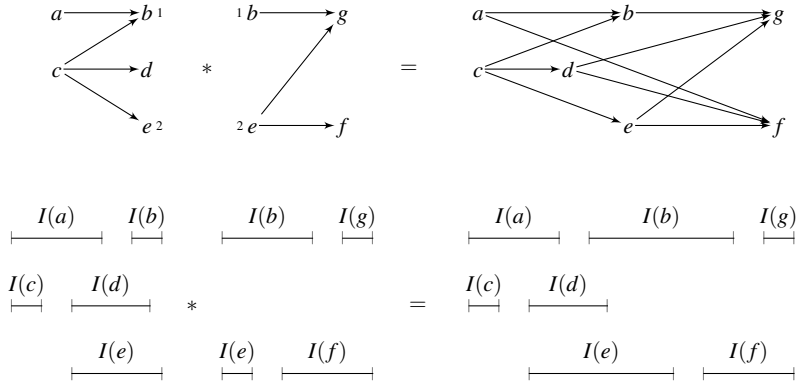

Restrictions of interval ipomsets are again interval.  The following
is shown in~\cite{DBLP:conf/RelMiCS/FahrenbergJST20} using interval
representations, see Figure \ref{fi:intcomp} for an example.

\begin{lemma}
  \label{le:intcomp}
  If $P$ and $Q$ are interval ipomsets and $P* Q$ exists, then $P* Q$
  is an interval ipomset. \qed
\end{lemma}

We develop a decomposition property for interval ipomsets which will
be useful later.

\begin{definition}
  An ipomset $P$ is \emph{discrete} if $<_P$ is empty (thus, $\intord_P$
  is a linear order).  In addition, $P$ is a \emph{starter} if
  $T_P= P$, and $P$ is a \emph{terminator} if $S_P= P$.  A starter $P$
  is \emph{elementary} if $P\setminus S_P$ is a singleton, and a
  terminator $P$ is elementary if $P\setminus T_P$ is a singleton.
\end{definition}

Starters may be used to start events and terminators to terminate
them.  In compositions, they can switch off parts of starting or
terminating interfaces, see Figure \ref{fi:discrete}.


Every starter is a gluing of elementary starters, and every terminator
a gluing of elementary terminators (both not necessarily unique).
Identity ipomsets are both starters and terminators, and any discrete
ipomset can be written as a gluing of a starter followed by a
terminator.

We introduce special notation for discrete ipomsets: for subsets
$S, T\subseteq U$ of a linear pomset $( U, \intord, \lambda)$ we write
\begin{equation*}
  \subid{S}{U}{T}=( U, \emptyset, \intord, \lambda, S, T).
\end{equation*}
The next lemma follows easily.

\begin{lemma}
  \label{le:SubIdComposition}
  Let $S$, $T$, $U$ be linear pomsets.  If $S\subseteq T\subseteq U$,
  then $\subid{S}{T}{T}*\subid{T}{U}{U}\cong \subid{S}{U}{U}$ and
  $\subid{U}{U}{T}*\subid{T}{T}{S}\cong \subid{U}{U}{S}$.  If
  $S, T\subseteq U$, then
  $\subid{S}{U}{U}*\subid{U}{U}{T}\cong \subid{S}{U}{T}$. \qed
\end{lemma}

\begin{proposition}
  \label{pr:DecompositionOfEvilio}
  For an ipomset $P$ the following are equivalent:
  \begin{enumerate}
  \item \label{en:DecompositionOfEvilio.intv} $P$ is an interval
    ipomset;
  \item \label{en:DecompositionOfEvilio.gludisc} $P$ is a finite
    gluing of discrete ipomsets;
  \item \label{en:DecompositionOfEvilio.glust} $P$ is a finite gluing
    of elementary starters and terminators.
  \end{enumerate}
\end{proposition}

\begin{proof}
  Equivalence of \eqref{en:DecompositionOfEvilio.gludisc} and
  \eqref{en:DecompositionOfEvilio.glust} is clear.  Given that
  discrete ipomsets are interval,
  \eqref{en:DecompositionOfEvilio.gludisc} implies
  \eqref{en:DecompositionOfEvilio.intv} by Lemma \ref{le:intcomp}.

  To show that \eqref{en:DecompositionOfEvilio.intv} implies
  \eqref{en:DecompositionOfEvilio.gludisc}, let $P$ be an interval
  ipomset and $P_1\prec\dotsm\prec P_m$ the sequence of maximal
  antichains in $P$ given by Lemma
  \ref{le:intord}\eqref{en:intord.maxanti}.  Each $P_i$ is linearly
  ordered by the restriction $\intord_i$ of $\intord_P$ to $P_i$.  Let
  $S_1= S_P$, $T_m= T_P$, and $T_i= S_{ i+ 1}= P_i\cap P_{ i+ 1}$ for
  $i\in[ m- 1]$, and define ipomsets
  $P_i=( P_i, \emptyset, \mathord{\intord_i}, \lambda\rest{P_i}, S_i,
  T_i)$.  Then the gluing $Q= P_1*\dotsm* P_m$ is defined; we show
  that $Q= P$.

  It is clear that the underlying sets of $Q$ and $P$ are equal and
  that the source and target interfaces agree.  Further,
  $\mathord{\intord_Q}=( \mathord{\intord_1}\cup\dotsm\cup
  \mathord{\intord_m})^+= \mathord{\intord_P}$.  To see that
  $\mathord{<_Q}= \mathord{<_P}$, we note that $x<_P y$ implies that
  $x\in P_i$ and $y\in P_j$ with $i< j$ and $x, y\notin P_i\cap P_j$,
  and vice versa.
\end{proof}

\begin{figure}
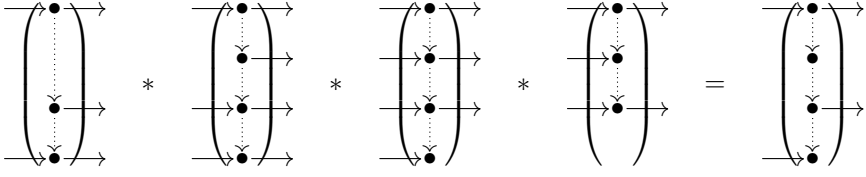

  \centering
  $\ipomset{ \ar[r] & \bullet \ar@{.>}[dd] \ar[r] & \\ &
    \phantom{\bullet} & \\ & \bullet
    \ar@{.>}[d] \ar[r] & \\ \ar[r] & \bullet \ar[r] &} \quad*\quad%
  \ipomset{ \ar[r] & \bullet \ar@{.>}[d] \ar[r] & \\ & \bullet
    \ar@{.>}[d] \ar[r] & \\ \ar[r] & \bullet \ar@{.>}[d] \ar[r] & \\
    \ar[r] & \bullet \ar[r] &} \quad*\quad%
  \ipomset{ \ar[r] & \bullet \ar@{.>}[d] \ar[r] & \\ \ar[r] & \bullet
    \ar@{.>}[d] \ar[r] & \\ \ar[r] & \bullet \ar@{.>}[d] \ar[r] & \\
    \ar[r] & \bullet &} \quad*\quad%
  \ipomset{ \ar[r] & \bullet \ar@{.>}[d] \ar[r] & \\ \ar[r] & \bullet
    \ar@{.>}[d] & \\ \ar[r] & \bullet \ar[r] & \\ & \phantom{\bullet} &} \quad=\quad%
  \ipomset{ \ar[r] & \bullet \ar@{.>}[d] \ar[r] & \\ & \bullet
    \ar@{.>}[d] & \\ & \bullet \ar@{.>}[d] \ar[r] & \\ \ar[r] &
    \bullet &}$ \bigskip
  \caption{Decomposition of discrete ipomset into elementary starters
    and terminators.}
  \label{fi:discrete}
\end{figure}

Figure \ref{fi:discrete} shows an example of a decomposition of a
discrete ipomset into elementary starters and terminators.  This
proposition also gives an alternate proof of Lemma \ref{le:intcomp}.

\subsection{Subsumption}
\label{se:subsu}

Pomsets may be \emph{smoothened}, or made less concurrent, by
strengthening precedence relations.  The corresponding relation
between pomsets has been introduced by Grabowski
\cite{DBLP:journals/fuin/Grabowski81} and is nowadays often called
\emph{subsumption} \cite{DBLP:journals/tcs/Gischer88}. We adapt it to
ipomsets.

\begin{definition}
  An ipomset $Q$ \emph{subsumes} an ipomset $P$ if there exists a bijection
  $f: P\to Q$, called a \emph{subsumption map}, such that
  \begin{itemize}
  \item for all $x, y\in P$, $f( x)<_Q f( y)$ implies $x<_P y$;
  \item for all $x, y\in P$ with $x\incomp_P y$, $x\intord_P y$
    implies $f( x)\intord_Q f( y)$; and
  \item $\lambda_Q\circ f=\lambda_P$, $f(S_P)=S_Q$, and $f(T_P)=T_Q$.
  \end{itemize}
\end{definition}

We write $P\subsu Q$ if $Q$ subsumes $P$.  That is, the points of $P$
and $Q$ are in bijection, but $P$ may be more precedence ordered than
$Q$, and $Q$ may be more event ordered than $P$.

If $P$ is discrete or $Q$ is linear, then any subsumption map
$f: P\to Q$ is an isomorphism and, in particular, unique.  We extend
gluing composition to subsumptions:

\begin{definition}
  \label{de:comp-subsu}
  Let $f: P\to P'$ and $g: Q\to Q'$ be subsumption maps and assume $P*
  Q$ and $P'* Q'$ to be defined.  Define $h= f* g: P* Q\to P'* Q'$ by
  \begin{equation*}
    h( x)=
    \begin{cases}
      f( x) &\text{if } x\in P, \\
      g( x) &\text{if } x\in Q,
    \end{cases}
  \end{equation*}
\end{definition}

\begin{lemma}
  The map $h$ from Definition~\ref{de:comp-subsu} is well-defined and
  a subsumption map.
\end{lemma}

\begin{proof}
  By Lemma \ref{le:glue-iso} we may assume that $T_P= S_Q$ and $T_{
    P'}= S_{ Q'}$, showing that $h$ is well-defined and a bijection.
  The other properties follow easily.
\end{proof}

\begin{lemma}
  \label{le:CompSubsu}
  If $P\subsu P'$ and $Q\subsu Q'$, then $P* Q$ is defined iff
  $P'* Q'$ is, and in that case, $P* Q\subsu P'* Q'$.
\end{lemma}

\begin{proof}
  Let $f:P\to P'$ and $g:Q\to Q'$ be the subsumption maps.  The first
  claim is clear as $f$ and $g$ respect interfaces and labels.  The
  second claim follows from Definition \ref{de:comp-subsu}.
\end{proof}

Using subsumption maps as 2-morphisms, ipomsets assemble as morphisms
into a bicategory.  Below, the \emph{identity} on a linear pomset
$( S, \intord, \lambda)$ is the discrete ipomset
$\id_S= \subid{S}{S}{S}=( S, \emptyset, \mathord{\intord}, \lambda, S,
S)$, with trivial precedence order and all points in both interfaces.

\begin{proposition}
  \label{pr:ipomscat}
  Ipomsets form a (large) bicategory $\iPoms$ with objects linear
  pomsets $( S, \mathord{\intord}, \lambda)$, ipomsets
  $( P, <, \intord, \lambda, S, T)$ as morphisms from
  $( S, \intord, \lambda)$ to $( T, \intord, \lambda)$ with gluing as
  composition and identities $\id_S$, subsumptions as 2-morphisms, and
  ipomset isomorphisms as 2-isomorphisms.
\end{proposition}

\begin{proof}
  It is clear that subsumption maps compose associatively and are
  invertible precisely when they are ipomset isomorphisms.  Gluing
  composition is associative up-to 2-isomorphism, and the ipomsets
  $\id_S$ are on-the-nose units for $*$.  The pentagon identity is
  trivially satisfied due to uniqueness of 2-isomorphisms.
\end{proof}

\begin{remark}
  The bicategory $\iPoms$ is large as its objects and morphisms form a
  proper class.  However, given that any ipomset is uniquely
  isomorphic to one on points $[ k]$ and with interfaces $[ n]$ and
  $[ m]$, for some $k, n, m\ge 0$, $\iPoms$ is equivalent to its
  skeleton, which is a small 2-category; hence $\iPoms$ is
  \emph{essentially small}.
\end{remark}


By Lemma \ref{le:intcomp}, interval ipomsets form a sub-bicategory
of $\iPoms$ which we will denote $\iiPoms$.

\section{Tracks and their labels}
\label{se:tracks}

We are now ready to introduce \emph{tracks} in precubical sets, which are our model of
computations, \ie
sequences of cells connected at faces.  We define labels of tracks as interval ipomsets and
show, conversely, how interval ipomsets give rise to tracks.

\subsection{Tracks}

Let $X$ be an event consistent precubical set.  For $x,y\in X$, we say
that $x$ is an \emph{elementary lower face} of $y$, denoted
$x\face y$, if $x=\delta^0_i y$ for some $i$; $x$ is an
\emph{elementary upper face} of $y$, denoted $y\ecaf x$, if
$x=\delta^1_i y$.  The reflexive transitive closures of the relations
$\face$ and $\ecaf$ are denoted $\face^*$ and $\ecaf^*$.

We say that $x$ is a \emph{lower}, resp.\ \emph{upper face} of $y$ if
$x\face^* y$, resp.\ $y\ecaf^* x$.  This is equivalent to the
condition that $x=\delta^{0,\dotsc,0}_A y$ for some (possibly empty)
$A$, resp.\ $x=\delta^{1,\dotsc,1}_A y$.  By Lemma \ref{le:SemiNSI},
$A$ is determined uniquely by $x$ and $y$.

\begin{definition}
  A \emph{track} in $X$ is a non-empty sequence
  $\rho=( x_1,\dotsc, x_m)$, $m\ge 1$, of elements of $X$ such that
  for all $i= 1,\dotsc, m- 1$, $x_i\face^* x_{ i+ 1}$ or
  $x_{ i+ 1}\ecaf^* x_i$.  A track $\rho$ as above is \emph{from $x_1$
    to $x_m$}, denoted $\rho: x_1\leadsto x_m$.
\end{definition}


We allow repeated cells $x_i= x_{ i+ 1}$ in tracks for notational
convenience.  A track $( x_1,\dotsc, x_m)$ is \emph{full} if it does
not contain such repeated cells and all face relations are elementary,
that is, $x_i\face x_{ i+ 1}$ or $x_i\ecaf x_{ i+ 1}$ for all
$i= 1,\dotsc, m- 1$.  Any track without repeated cells may be
\emph{filled} to a full track by inserting appropriate (not
necessarily unique) cells.  In \cite{DBLP:journals/tcs/Glabbeek06},
full tracks are called \emph{execution paths}.

\begin{example}
  Figure \ref{fi:track} in the introduction displays the track
  $( x_1, x_2, x_3, x_4, x_5, x_6)$.  As $x_3\face^3 x_4$, this track
  is not full; it may be filled by inserting appropriate faces of
  $x_4$ and $x_6$, for example
  \begin{equation*}
    ( x_1, x_2, x_3, \delta_1^0 \delta_1^0 x_4, \delta_1^0 x_4, x_4,
    \delta_1^1 x_4, x_5, x_6).
  \end{equation*}
\end{example}

\begin{definition}
  Let $\rho=( x_1,\dotsc, x_m)$, $\tau=( y_1,\dotsc, y_k)$ be tracks
  in $X$.  The \emph{concatenation} $\rho* \tau$ of $\rho$ and
  $\tau$ is defined if $x_m= y_1$, and in that case,
  $\rho* \tau=( x_1,\dotsc, x_m, y_2,\dotsc, y_k)$.
\end{definition}

The \emph{unit} tracks are $( x)$ for $x\in X$.  Tracks containing
exactly two cells are called \emph{basic}.  Concatenation is
associative, and every non-unit track is a unique concatenation of
basic tracks.
The following is clear.

\begin{lemma}
  \label{le:trackcat}
  Tracks in $X$ form a small category\/ $\Track{ X}$ with objects
  $x\in X$, tracks $\rho: x\leadsto y$ as morphisms, $*$ as
  composition, and identities $\id_x=( x)$.
  \qed
\end{lemma}

\subsection{Labels of Tracks}

Let $( X, \lambda)$ be a labeled precubical set.

\begin{definition}
  \label{de:label-track}
  The \emph{label} of a track $\rho$ in $X$ is the ipomset
  $\ell( \rho)$ defined recursively as follows:
  \begin{itemize}
  \item If $\rho=( x)$ is a unit track, then
    $\ell(\rho)= \id_{\ell(x)}$: the identity ipomset on $\ell(x)$.
  \item If $\rho=( x, y)$ with $x\face^* y$, then
    $\ell(\rho)= \subid{\ell(x)}{\ell(y)}{\ell(y)}$: a starter.
  \item If $\rho=( y, x)$ with $y\ecaf^* x$, then
    $\ell(\rho)= \subid{\ell(y)}{\ell(y)}{\ell(x)}$: a terminator.
  \item If $\rho= \tau* \rho'$ with $\tau$ a basic track, then
    $\ell(\rho)=\ell(\tau)*\ell(\rho')$.
  \end{itemize}
\end{definition}

By Proposition \ref{pr:DecompositionOfEvilio}, labels of tracks are
interval ipomsets, and
$\ell( \rho_1* \rho_2)\cong \ell( \rho_1)* \ell( \rho_2)$ for all
tracks $\rho_1$, $\rho_2$.  The following is therefore clear.

\begin{proposition}
  \label{pr:lambda-track}
  Labeling defines a functor $\ell: \Track{X}\to \iiPoms$.
\end{proposition}

Next we see that filling a track with extra cells does not change its
label.  Let $\sim$ be the equivalence on sets of tracks generated by
$( x_1,\dotsc, x_m)\sim( x_1,\dotsc, x_{i-1}, x_{i+1},\dotsc, x_m)$
for $x_{i-1}\face^* x_{i}\face^* x_{i+1}$ or
$x_{i-1}\ecaf^* x_{i}\ecaf^* x_{i+1}$.

\begin{lemma}
  If $\rho\sim \tau$, then $\ell(\rho)\cong \ell(\tau)$.
\end{lemma}

\begin{proof}
  If $x\face^* y\face^* z$, then
  \begin{equation*}
    \ell((x, y, z))=
    \ell((x, y))* \ell((y, z)) =
    \subid{\ell(x)}{\ell(y)}{\ell(y)}*\subid{\ell(y)}{\ell(z)}{\ell(z)}=
    \subid{\ell(x)}{\ell(z)}{\ell(z)}=
    \ell((x, z))
  \end{equation*}
  from Lemma \ref{le:SubIdComposition}. The computations for
  $x\ecaf^* y\ecaf^* z$ are similar, and the result then follows by
  induction.
\end{proof}

The next lemma shows that labels of tracks that consist of a cell and
two of its faces on either side are discrete ipomsets. Its proof is a
straightforward application of the definition.

\begin{lemma}
  \label{le:UpDownTrackLabel}
  If $x\face^* y \ecaf^* z$, then
  $\ell(( x, y, z))=\subid{\ell(x)}{\ell(y)}{\ell(z)}$. \qed
\end{lemma}

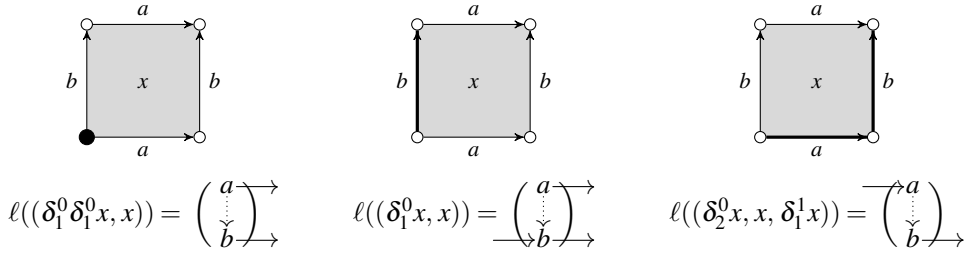
\begin{figure}
  \centering
  \setlength{\tabcolsep}{.5cm}
  \begin{tabular}{ccc}
    \begin{tikzpicture}[x=1.5cm, y=1.5cm]
      \path[fill=black!15] (0,0) to (1,0) to (1,1) to (0,1);
      \foreach \x in {0, 1} \foreach \y in {0, 1} \node[state] (\x\y)
      at (\x,\y) {};
      \path (00) edge node[below] {$a$} (10);
      \path (01) edge node[above] {$a$} (11);
      \path (00) edge node[left] {$b$} (01);
      \path (10) edge node[right] {$b$} (11);
      \node at (.5,.5) {$\vphantom{b}x$};
      \node[state, fill=black, minimum size=2mm] at (0,0) {};
    \end{tikzpicture}
    &
    \begin{tikzpicture}[x=1.5cm, y=1.5cm]
      \path[fill=black!15] (0,0) to (1,0) to (1,1) to (0,1);
      \foreach \x in {0, 1} \foreach \y in {0, 1} \node[state] (\x\y)
      at (\x,\y) {};
      \path (00) edge node[below] {$a$} (10);
      \path (01) edge node[above] {$a$} (11);
      \path (00) edge node[left] {$b$} (01);
      \path (10) edge node[right] {$b$} (11);
      \node at (.5,.5) {$\vphantom{b}x$};
      \path (00) edge[-, very thick] (01);
    \end{tikzpicture}
    &
    \begin{tikzpicture}[x=1.5cm, y=1.5cm]
      \path[fill=black!15] (0,0) to (1,0) to (1,1) to (0,1);
      \foreach \x in {0, 1} \foreach \y in {0, 1} \node[state] (\x\y)
      at (\x,\y) {};
      \path (00) edge node[below] {$a$} (10);
      \path (01) edge node[above] {$a$} (11);
      \path (00) edge node[left] {$b$} (01);
      \path (10) edge node[right] {$b$} (11);
      \node at (.5,.5) {$\vphantom{b}x$};
      \path (00) edge[-, very thick] (10);
      \path (10) edge[-, very thick] (11);
    \end{tikzpicture}
    \\
    $\ell(( \delta_1^0 \delta_1^0 x, x))= \!\!\ipomset{ & a
      \ar@{.>}[d] \ar[r] & \\ & b \ar[r] &}$ &
    $\ell(( \delta_1^0 x, x))= \!\!\ipomset{ & a \ar@{.>}[d]
      \ar[r] & \\ \ar[r] & b \ar[r] &}$ &
    $\ell(( \delta_2^0 x, x, \delta_1^1 x))= \!\!\ipomset{ \ar[r] & a \ar@{.>}[d]
      \\ & b \ar[r] &}$
  \end{tabular}
  \bigskip
  \caption{Some tracks and their labels in $X=\square^{\{a\intord
      b\}}$ (solid cells indicate the respective faces).}
  \label{fi:label-iface}
\end{figure}

Figure \ref{fi:label-iface} shows some examples of simple tracks and
their labels.

\subsection{Interval Ipomsets as Tracks}

We have seen how labels of tracks in HDAs can be computed as interval
ipomsets.  Now we show the inverse: how interval ipomsets may be
converted into HDAs consisting essentially of a single track.  To this
end, first introduce a relation $\prec$ on the set $\{0, \exec, 1\}$
by
\begin{equation*}
  \mathord{ \prec}= \{(0,0),(\exec,0),(1,0),(1,\exec),(1,1)\}.
\end{equation*}

The relation $\prec$ (which is \emph{not} a partial order, given that
it is neither reflexive nor irreflexive) corresponds to the meaning
that we associate to the elements of $\{ 0, \exec, 1\}$: $0$ meaning
the event has not yet started; $\exec$ for an executing event; and $1$
if the event has terminated.  The intuition is that when $x$ and $y$
are events so that $x< y$ in the precedence order, then either $y$ has
not yet started, in which case $x$ may be in any state, hence the
first three pairs $( 0, 0)$, $( \exec, 0)$, $( 1, 0)$; or $x$ has
terminated and $y$ may be in any state, hence the last three pairs
$( 1, 0)$, $( 1, \exec)$, $( 1, 1)$.  In particular, it is impossible
that both are active, so that $(\exec, \exec)\notin \mathord{ \prec}$.

\begin{remark}
  In Chu spaces for concurrency \cite{pratt95chu},
  $K_{3}=\{0 \precchu \exec \precchu 1\}$ is the structure that
  defines the possible execution forms that an event can take.  The
  intuition of the order $\precchu$ is that $0$ (the event has not yet
  started) can happen only before $\exec$ (the event is executing),
  which can happen only before $1$ (the event has terminated).  This
  order can be extended to sets of (execution values of) events, which
  in Chu terminology is called a state. For states, the order
  $\precchu$ expresses admissible sequences of executions of the
  system.  For two independent events $x$, $y$, all execution forms
  $( e_x, e_y)$ with $e_x, e_y\in\{ 0, \exec, 1\}$ would thus be
  admissible; but if there is a precendence order $x< y$, then the
  allowed tuples are precisely the ones in the relation $\prec$
  above.
\end{remark}

The following generalizes Definition \ref{de:pobj-S} of the standard
$S$-cube $\square^S$ on a linear poset $S$ to arbitrary ipomsets.

\begin{definition}
  \label{de:pobj}
  For an ipomset $P$, define the HDA
  $( \pobj{ P}, I_{ \pobj{P}}, F_{ \pobj{P}}, \lambda_{ \pobj{P}})$ as
  follows:
  \begin{itemize}
  \item $\pobj{ P}_k$ is the set of all relation-preserving functions
    $x:( P, \mathord{<_P})\to(\{ 0, \exec, 1\}, \mathord{\prec})$
    taking value $\exec$ on exactly $k$ elements;
  \item for $x\in \pobj{ P}_k$ and
    $x^{-1}(\exec)=\{p_1\intord_P\dotsm\intord_P p_k\}$,
    \begin{equation*}
      \delta_i^\nu( x)( p)=
      \begin{cases}
        \nu &\text{for } p= p_i, \\
        x( p) &\text{for } p\ne p_i;
      \end{cases}
    \end{equation*}
  \item $I_{ \pobj{ P}}=\{ i_{ \pobj{ P}}\}$ and
    $F_{ \pobj{ P}}=\{ f_{ \pobj{ P}}\}$ are given by
    \begin{equation*}
      i_{ \pobj{ P}}( p)=
      \begin{cases}
        \exec & \text{if } p\in S_P, \\
        0 & \text{if } p\not\in S_P,
      \end{cases}
      \qquad
      f_{ \pobj{ P}}( p)=
      \begin{cases}
        \exec & \text{if } p\in T_P, \\
        1 & \text{if } p\not\in T_P;
      \end{cases}
    \end{equation*}
  \item For $x\in \pobj{P}_k$ with
    $x^{ -1}( \exec)=\{p_1\intord_P\dotsm\intord_P p_k\}$,
    $\lambda_{\pobj{P}}(x)=(\lambda_P(p_1),\dotsc,\lambda_P(p_k))$.
  \end{itemize}
\end{definition}


Above, $x^{ -1}( \exec)$ is indeed an antichain in the precedence
order and hence linearly ordered by~$\intord_P$.

\begin{example}
  \label{ex:N-to-HDA}
  Let $P$ be the ipomset
  \begin{equation*}
    P= \!\!\ipomset[3]{& a \ar[r] & b & \\ & c \ar[r] \ar@{.>}[u] \ar[ur] & d
      \ar@{.>}[u] \ar@{.>}[ul] &}\!\!.
  \end{equation*}
  The cells of $\pobj{P}$ are as follows, in increasing order of
  dimension (and with the event order omitted):
  \begin{gather*}
    x_1= \pomset{0 \ar[r] & 0 \\ 0 \ar[r] \ar[ur] & 0} \qquad%
    x_2= \pomset{1 \ar[r] & 0 \\ 0 \ar[r] \ar[ur] & 0} \qquad%
    x_3= \pomset{0 \ar[r] & 0 \\ 1 \ar[r] \ar[ur] & 0} \qquad%
    x_4= \pomset{1 \ar[r] & 0 \\ 1 \ar[r] \ar[ur] & 0} \\%
    x_5= \pomset{0 \ar[r] & 0 \\ 1 \ar[r] \ar[ur] & 1} \qquad%
    x_6= \pomset{1 \ar[r] & 0 \\ 1 \ar[r] \ar[ur] & 1} \qquad%
    x_7= \pomset{1 \ar[r] & 1 \\ 1 \ar[r] \ar[ur] & 0} \qquad%
    x_8= \pomset{1 \ar[r] & 1 \\ 1 \ar[r] \ar[ur] & 1} \\%
    y_1= \pomset{\exec \ar[r] & 0 \\ 0 \ar[r] \ar[ur] & 0} \qquad%
    y_2= \pomset{0 \ar[r] & 0 \\ \exec \ar[r] \ar[ur] & 0} \qquad%
    y_3= \pomset{\exec \ar[r] & 0 \\ 1 \ar[r] \ar[ur] & 0} \qquad%
    y_4= \pomset{0 \ar[r] & 0 \\ 1 \ar[r] \ar[ur] & \exec} \\%
    y_5= \pomset{\exec \ar[r] & 0 \\ 1 \ar[r] \ar[ur] & 1} \qquad%
    y_6= \pomset{1 \ar[r] & 0 \\ \exec \ar[r] \ar[ur] & 0} \qquad%
    y_7= \pomset{1 \ar[r] & \exec \\ 1 \ar[r] \ar[ur] & 0} \qquad%
    y_8= \pomset{1 \ar[r] & 0 \\ 1 \ar[r] \ar[ur] & \exec} \\%
    y_9= \pomset{1 \ar[r] & \exec \\ 1 \ar[r] \ar[ur] & 1} \qquad%
    y_{10}= \pomset{1 \ar[r] & 1 \\ 1 \ar[r] \ar[ur] & \exec} \\%
    z_1= \pomset{\exec \ar[r] & 0 \\ \exec \ar[r] \ar[ur] & 0} \qquad%
    z_2= \pomset{\exec \ar[r] & 0 \\ 1 \ar[r] \ar[ur] & \exec} \qquad%
    z_3= \pomset{1 \ar[r] & \exec \\ 1 \ar[r] \ar[ur] & \exec}
  \end{gather*}

  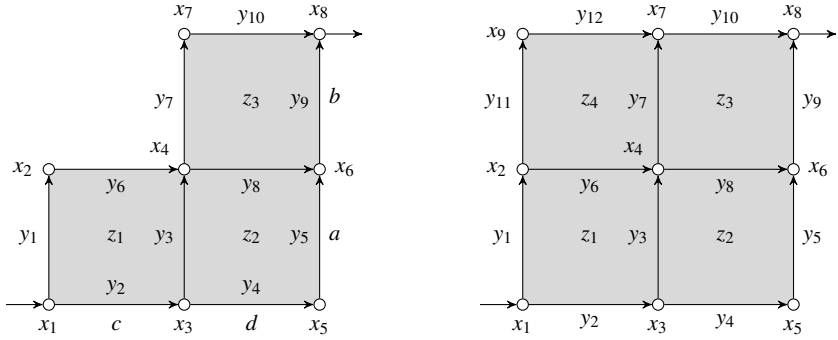
\begin{figure}
    \centering
    \begin{tikzpicture}[x=1.8cm, y=1.8cm]
      \begin{scope}
        \path[fill=black!15] (0,0) to (2,0) to (2,2) to (1,2) to (1,1)
        to (0,1);
        \node[state, initial left] (00) at (0,0) {};
        \node[state] (10) at (1,0) {};
        \node[state] (20) at (2,0) {};
        \node[state] (01) at (0,1) {};
        \node[state] (11) at (1,1) {};
        \node[state] (21) at (2,1) {};
        \node[state] (12) at (1,2) {};
        \node[state, accepting] (22) at (2,2) {};
        \node[below] at (00.south) {$x_1$};
        \node[left] at (01.west) {$x_2$};
        \node[below] at (10.south) {$x_3$};
        \node[above left] at (11.north west) {$x_4$};
        \node[below] at (20.south) {$x_5$};
        \node[right] at (21.east) {$x_6$};
        \node[above] at (12.north) {$x_7$};
        \node[above] at (22.north) {$x_8$};
        \node at (.5,.5) {$\vphantom{y}z_1$};
        \node at (1.5,.5) {$\vphantom{y}z_2$};
        \node at (1.5,1.5) {$\vphantom{y}z_3$};
        \path (00) edge node[below] {$\vphantom{d}c$} node[above] {$y_2$} (10);
        \path (10) edge node[below] {$d$} node[above] {$y_4$} (20);
        \path (01) edge node[below] {$y_6$} (11);
        \path (11) edge node[below] {$y_8$} (21);
        \path (12) edge node[above] {$y_{10}$} (22);
        \path (00) edge node[left] {$y_1$} (01);
        \path (10) edge node[left] {$y_3$} (11);
        \path (11) edge node[left] {$y_7$} (12);
        \path (20) edge node[right] {$\vphantom{y_5}a$} node[left] {$y_5$} (21);
        \path (21) edge node[right, pos=.55] {$b$} node[left] {$y_9$} (22);
      \end{scope}
      \begin{scope}[shift={(3.5,0)}]
        \path[fill=black!15] (0,0) to (2,0) to (2,2) to (0,2);
        \node[state, initial left] (00) at (0,0) {};
        \node[state] (10) at (1,0) {};
        \node[state] (20) at (2,0) {};
        \node[state] (01) at (0,1) {};
        \node[state] (11) at (1,1) {};
        \node[state] (21) at (2,1) {};
        \node[state] (02) at (0,2) {};
        \node[state] (12) at (1,2) {};
        \node[state, accepting] (22) at(2,2) {};
        \node[below] at (00.south) {$x_1$};
        \node[left] at (01.west) {$x_2$};
        \node[below] at (10.south) {$x_3$};
        \node[above left] at (11.north west) {$x_4$};
        \node[below] at (20.south) {$x_5$};
        \node[right] at (21.east) {$x_6$};
        \node[above] at (12.north) {$x_7$};
        \node[above] at (22.north) {$x_8$};
        \node[left] at (02.west) {$x_9$};
        \node at (.5,.5) {$\vphantom{y}z_1$};
        \node at (1.5,.5) {$\vphantom{y}z_2$};
        \node at (1.5,1.5) {$\vphantom{y}z_3$};
        \node at (.5,1.5) {$\vphantom{y}z_4$};
        \path (00) edge node[below] {$y_2$} (10);
        \path (10) edge node[below] {$y_4$} (20);
        \path (01) edge node[below] {$y_6$} (11);
        \path (11) edge node[below] {$y_8$} (21);
        \path (12) edge node[above] {$y_{10}$} (22);
        \path (00) edge node[left] {$y_1$} (01);
        \path (10) edge node[left] {$y_3$} (11);
        \path (11) edge node[left] {$y_7$} (12);
        \path (20) edge node[right] {$y_5$} (21);
        \path (21) edge node[right] {$y_9$} (22);
        \path (01) edge node[left] {$y_{11}$} (02);
        \path (02) edge node[above] {$y_{12}$} (12);
      \end{scope}
    \end{tikzpicture}
    \bigskip
    \caption{The HDAs $\pobj{P}$ (left) and $\pobj{Q}$ (right) from
      Examples \ref{ex:N-to-HDA} and \ref{ex:2+2-to-HDA}.}
    \label{fi:pobj-ex}
  \end{figure}

  The cells $y_1$, $y_3$ and $y_5$ are labeled by $a$, $y_7$ and $y_9$
  by $b$, $y_2$ and $y_6$ by $c$, $y_4$, $y_8$, $y_{10}$ by $d$.  The
  labels of $z_1$, $z_2$ and $z_3$ are $(c,a)$, $(d,a)$ and $(d,b)$,
  respectively.  The order of letters in these pairs is determined be
  the event order on $P$.  Geometrically these are arranged as shown
  in Figure \ref{fi:pobj-ex} (left).
\end{example}

We will later apply Definition \ref{de:pobj} to \emph{interval}
ipomsets to conclude in Proposition \ref{pr:XPSubsumption} that the
language of $\pobj{P}$ is generated by $P$.  Our definition applies to
general ipomsets, but as we will see, Proposition
\ref{pr:XPSubsumption} fails for ipomsets which are not interval.  It
is an interesting open problem to characterize those HDA which are
isomorphic to some $\pobj{P}$.

\begin{example}
  \label{ex:2+2-to-HDA}
  If we instead of the ipomset $P$ of Example \ref{ex:N-to-HDA} take
  $Q$ to be the $\twotwo$-ipomset
  \begin{equation*}
    Q= \!\!\ipomset[3]{& a \ar[r] & b & \\ & c \ar[r] \ar@{.>}[u]
      \ar@{.>}[ur] & d \ar@{.>}[u] \ar@{.>}[ul] &}\!\!,
  \end{equation*}
  then $\pobj{Q}$ contains $\pobj{P}$ and the following extra cells:
  \begin{equation*}
    x_9= \pomset{1 \ar[r] & 1 \\ 0 \ar[r] & 0} \qquad%
    y_{11}= \pomset{1 \ar[r] & \exec \\ 0 \ar[r] & 0} \qquad%
    y_{12}= \pomset{1 \ar[r] & 1 \\ \exec \ar[r] & 0} \qquad%
    z_4= \pomset{1 \ar[r] & \exec \\ \exec \ar[r] & 0}
  \end{equation*}
  Geometrically this amounts to adding the top-left square to
  $\pobj{P}$, see Figure \ref{fi:pobj-ex} (right).
\end{example}

\begin{figure}
  \centering
  \begin{tikzpicture}[x=1.5cm, y=1.5cm]
    \begin{scope}
      \node[font=\normalsize] (ab) at (.5,2.2) {$\ipomset{& a
          \ar@{.>}[d] & \\ & b &}$};
      \path[fill=black!15] (0,0) to (1,0) to (1,1) to (0,1);
      \node[state, initial left] (00) at (0,0) {};
      \node[state] (01) at (0,1) {};
      \node[state] (10) at (1,0) {};
      \node[state, accepting] (11) at (1,1) {};
      \path (00) edge node[below] {$a$} (10);
      \path (01) edge node[above] {$a$} (11);
      \path (00) edge node[left] {$b$} (01);
      \path (10) edge node[right] {$b$} (11);
      \node[coordinate] (ableft) at (-.5,.5) {};
      \node[coordinate] (abright) at (1.5,.5) {};
    \end{scope}
    \begin{scope}[shift={(3,0)}]
      \node[font=\normalsize] (a-b) at (.5,2.2) {$\ipomset{& a
          \ar[r] & b &}$};
      \node[state, initial left] (00) at (0,0) {};
      \node[state] (10) at (1,0) {};
      \node[state, accepting] (11) at (1,1) {};
      \path (00) edge node[below] {$a$} (10);
      \path (10) edge node[right] {$b$} (11);
      \node[coordinate] (a-bleft) at (.3,.5) {};
    \end{scope}
    \begin{scope}[shift={(-2.6,0)}]
      \node[font=\normalsize] (b-a) at (.5,2.2) {$\ipomset{& b
          \ar[r] & a &}$};
      \node[state, initial left] (00) at (0,0) {};
      \node[state] (01) at (0,1) {};
      \node[state, accepting] (11) at (1,1) {};
      \path (01) edge node[above] {$a$} (11);
      \path (00) edge node[left] {$b$} (01);
      \node[coordinate] (b-aright) at (.7,.5) {};
    \end{scope}
    \path (a-b.west) edge (ab.east);
    \path (b-a.east) edge (ab.west);
    \path (a-bleft) edge (abright);
    \path (b-aright) edge (ableft);
  \end{tikzpicture}
  \bigskip
  \caption{Subsumptions (top) give rise to HDA inclusions (bottom)}
  \label{fi:subsu-tracks}
\end{figure}
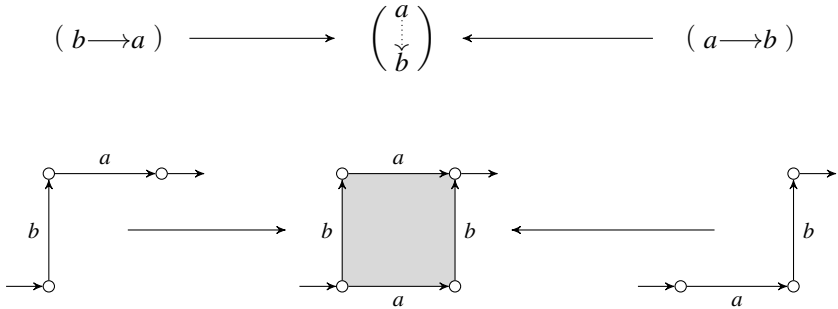

The following can be shown by easy calculations; Figure
\ref{fi:subsu-tracks} shows some simple examples.

\begin{lemma}
  \label{le:SubsumptionToPObj}
  If $f: P\to Q$ is a subsumption map, then the function
  $\pobj{f}: \pobj{P}\to \pobj{Q}$ given by $\pobj{f}(x)(p)=x(f(p))$
  is an injective HDA map. \qed
\end{lemma}

For the next lemma, recall the notions of event consistency and
universal events from Section~\ref{se:labelevent}.

\begin{lemma}
  \label{le:events-pobj}
  Let $P$ be an ipomset.  Then $\pobj{P}$ is event consistent, and
  $E_{ \pobj{ P}}\cong P$ as labeled sets.  For every $x\in \pobj{P}$,
  $\ev(x)=x^{-1}(\exec)$ as linear posets.
\end{lemma}

\begin{proof}
  To show that $\pobj{P}$ is event consistent, let $f: P\to Q$ be a
  subsumption map into a discrete ipomset $Q$; $Q$ may be obtained
  from any linearization of $\intord_P$.  Then the precubical set
  underlying~$\pobj{Q}$ is a standard cube, and by Lemma
  \ref{le:SubsumptionToPObj}, $\pobj{f}: \pobj{P}\to \pobj{Q}$ is an
  embedding.  By Example \ref{ex:scube-evcont}, $\pobj{P}$ is event
  consistent.
  
  There is a $P$-labeling of $\pobj{P}$, \ie a precubical map
  $\pobj{P}\to \bang P$ that sends $x\in\pobj{P}$ to
  $x^{-1}(\exec)=(p_1\intord_P\dotsm\intord_P p_n)$.  This induces a
  function $\pi:E_{\pobj{P}}\to P$ (Proposition
  \ref{prp:EventsAreUniversalLabeling}).  For every $p\in P$, define
  $x_p\in \pobj{P}_1$ by
  \begin{equation*}
    x_p( q)=
    \begin{cases}
      1 &\text{if } q<_P p, \\
      \exec &\text{if } q= p, \\
      0 &\text{otherwise}.
    \end{cases}
  \end{equation*}
  Since $\pi(x_p)=p$, $\pi$ is surjective.

  It remains to show that $\pi$ is injective.  Let $x\in \pobj{P}_1$
  satisfy $\pi(x)=p$; we will show that $x\eveq x_p$.  Note that
  $x( q)= \exec$ iff $q= p$.  We proceed by induction on the number of
  elements in the set $D_x=\{q\in P\mid x(q)\neq x_p(q)\}$, see Figure
  \ref{fi:events-pobj} for an illustration.

  \begin{figure}
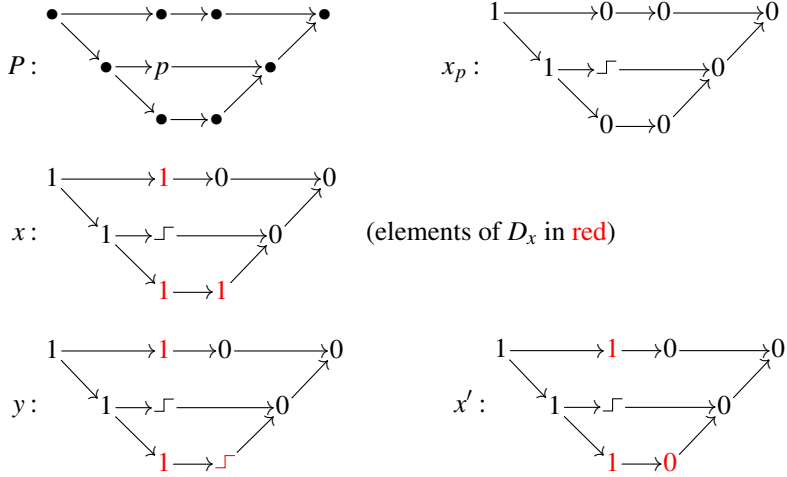

    \begin{align*}
      P &: \pomsetwop{2.7ex}{1.3em}{1pt}{%
        \bullet \ar[rr] \ar[dr] && \bullet \ar[r] & \bullet \ar[rr] &&
        \bullet \\%
        & \bullet \ar[r] \ar[rd] & p \ar[rr] && \bullet \ar[ur] \\%
        && \bullet \ar[r] & \bullet \ar[ur]}%
      \qquad\qquad x_p: \pomsetwop{2.7ex}{1.3em}{1pt}{%
        1 \ar[rr] \ar[dr] && 0 \ar[r] & 0 \ar[rr] && 0 \\%
        & 1 \ar[r] \ar[rd] & \exec \ar[rr] && 0 \ar[ur] \\%
        && 0 \ar[r] & 0 \ar[ur]} \\[1.5ex]
      x &: \pomsetwop{2.7ex}{1.3em}{1pt}{%
        1 \ar[rr] \ar[dr] && \red1 \ar[r] & 0 \ar[rr] && 0 \\%
        & 1 \ar[r] \ar[rd] & \exec \ar[rr] && 0 \ar[ur] \\%
        && \red1 \ar[r] & \red1 \ar[ur]}%
      \quad \text{(elements of $D_x$ in \red{red})} \\[2ex]
      y &: \pomsetwop{2.7ex}{1.3em}{1pt}{%
        1 \ar[rr] \ar[dr] && \red1 \ar[r] & 0 \ar[rr] && 0 \\%
        & 1 \ar[r] \ar[rd] & \exec \ar[rr] && 0 \ar[ur] \\%
        && \red1 \ar[r] & \red{\exec} \ar[ur]}%
      \qquad\qquad x': \pomsetwop{2.7ex}{1.3em}{1pt}{%
        1 \ar[rr] \ar[dr] && \red1 \ar[r] & 0 \ar[rr] && 0 \\%
        & 1 \ar[r] \ar[rd] & \exec \ar[rr] && 0 \ar[ur] \\%
        && \red1 \ar[r] & \red0 \ar[ur]}
    \end{align*}
    \bigskip
    \caption{Pomsets and cells in the proof of Lemma
      \ref{le:events-pobj}.}
    \label{fi:events-pobj}
  \end{figure}
  
  All elements of $D_x$ are parallel with $p$, since monotonicity of
  $x$ and $x_p$ implies that $x(q)=x_p(q)=1$ if $q<_P p$ and
  $x(q)=x_p(q)=0$ if $p<_P q$.  Thus, $x_p(q)=0$ and $x(q)=1$ for all
  $q\in D_x$.

  Let $q\in D_x$ be a $<_P$-maximal element.  Let
  $y:P\to \{0,\exec,1\}$ be given by $y(q)= \exec$ and $y(r)=x(r)$ for
  $r\neq q$.  We show that $y$ is monotone and hence a $2$-cell in
  $\pobj{P}$.

  We have $y(r)=x(r)=1$, hence $y(r)\prec y(q)$, for all $r<_P q$, since $x$
  preserves $<_P$.  For $q<_P r$, on the other hand, maximality of $q$
  in $D_x$ implies that $r\notin D_x$, so that $y(r)=x(r)=x_p(r)=0$,
  hence $y(q)\prec y(r)$.  Given that $y(r)= x(r)$ for $r\ne q$, we
  have shown that $y$ is monotone.

  Now $y( p)= y( q)= \exec$, so $y$ is a $2$-cell in $\pobj{P}$.  Let
  $i=1$ if $q\intord p$ and $i=2$ if $p\intord q$, then
  $\delta^1_i(y)=x$.  Let $x'=\delta^0_i(y)$, then also
  $\pi(x')=p$, and $D_{x'}=D_x\setminus\{q\}$.  The inductive
  hypothesis asserts that
  $x=\delta^1_i(y)\eveq \delta^0_i(y)= x'\eveq x_p$.
\end{proof}

Next we see that gluings of ipomsets correspond to pushouts of their
HDA objects.  Recall the Yoneda inclusions $\ineda_x$ from Lemma
\ref{le:ineda}.

\begin{lemma}
  \label{le:pobjJoin}
  Let $Q$ and $R$ be composable ipomsets with $T_Q\cong S\cong S_R$
  and $P= Q* R$.  There is a pushout
  \begin{equation*}
    \xymatrix{%
      \square^S \ar[r]^{\ineda_{f_{\pobj{Q}}}}
      \ar[d]_{\ineda_{i_{\pobj{R}}}} \ar@{}[dr]|>{\Big\ulcorner} &
      \pobj{Q} \ar[d]^{\jneda^0_{Q\subseteq P}} \\
      \pobj{R} \ar[r]_{\jneda^1_{R\subseteq P}} & \pobj{P}
    }
  \end{equation*}
  where 
  \begin{equation*}
    \jneda^0_{Q\subseteq P}(x)(p)=
    \begin{cases}
      x(p) & \text{for $p\in Q$}, \\
      0 & \text{otherwise},
    \end{cases}
    \qquad
    \jneda^1_{R\subseteq P}(x)(p)=
    \begin{cases}
      x(p) & \text{for $p\in R$}, \\
      1 & \text{otherwise}.
    \end{cases}
  \end{equation*}
\end{lemma}

\begin{proof}
  It is clear that all the maps in the diagram are injective.  Fix
  $x\in \pobj{P}$.
  \begin{itemize}
  \item If there exists $q\in Q\setminus T_Q$ with
    $x(q)\in \{0,\exec\}$, then obviously
    $x\not\in \jneda^1_{R\subseteq P} (\pobj{R})$. But for all
    $r\in R\setminus S_R$ we have $q<r$ and then $x(r)=0$.  It is easy
    to verify that the restriction $x\rest{Q}\in \pobj{Q}$ and then
    $x=\jneda^0_{Q\subseteq P}(x\rest{Q})\in \jneda^0_{Q\subseteq P}
    (\pobj{Q})$.
  \item Similarly, if $x(r)\in\{\exec,1\}$ for some
    $r\in R\setminus S_R$, then $x\in \jneda^1_{R\subseteq P} (\pobj{R})$.
  \item We have
    \begin{equation*}
      \jneda^0_{Q\subseteq P}\circ \ineda_{f_{\pobj{Q}}}(y)(p) = \jneda^1_{R\subseteq
        P}\circ \ineda_{i_{\pobj{R}}}(y)(p) =
      \begin{cases}
        1 & \text{for $p\in Q\setminus T_Q$}, \\
        \exec & \text{for $p\in S\cong T_Q\cong S_R$}, \\
        0 & \text{for $p\in R\setminus S_R$}.
      \end{cases}
    \end{equation*}
    Thus, the diagram commutes. Denote
    $j=\jneda^0_{Q\subseteq P}\circ \ineda_{f_{\pobj{Q}}}=\jneda^1_{R\subseteq
      P}\circ \ineda_{i_{\pobj{R}}}$.  The condition $x(q)=1$ for all
    $q\in Q\setminus T_Q$ and $x(r)=0$ for all $r\in R\setminus S_R$
    is equivalent to both
    $x\in \jneda^0_{Q\subseteq P} (\pobj{Q}) \cup \jneda^1_{R\subseteq P}
    (\pobj{R})$ and $x\in j(\pobj{P})$.
  \end{itemize}
  As a consequence,
  $\pobj{P}=\jneda^0_{Q\subseteq P} (\pobj{Q}) \cup \jneda^1_{R\subseteq
    P}(\pobj{R})$ and
  $j(\square^S)=\jneda^0_{Q\subseteq P} (\pobj{Q}) \cap \jneda^1_{R\subseteq
    P}(\pobj{R})$.
\end{proof}

\begin{lemma}
  \label{le:track-to-map}
  Let $X$ be a labeled precubical set,
  $\rho: x\leadsto y\in \Track{ X}$, and $P= \ell( \rho)$.  There
  is a map of labeled precubical sets $g:\pobj{P}\to X$ such that
  $g( i_{ \pobj{ P}})=x$ and $g( f_{ \pobj{ P}})= y$.
\end{lemma}

\begin{proof}
  Induction on the number of cells in $\rho$.
  \begin{itemize}
  \item If $\rho=(x)$, then $P= \id_{ \ell(x)}$ and
    $\pobj{P}=\square^{\ev(x)}$.  The Yoneda map
    $\ineda_x:\square^{\ev(x)}\to X$ satisfies the required condition.
  \item If $\rho=(x, y)$ with $x\face^* y$, then
    $P=\subid{\ell(x)}{\ell(y)}{\ell(y)}$.  Again, we may take $g=\ineda_y$.
  \item The case $\rho=(y, x)$ with $y\ecaf^* x$ is similar.
  \item In case $\rho=\sigma*\tau$, where both $\sigma:x\leadsto z$
    and $\tau:z\leadsto y$ are shorter than $\rho$, let
    $Q=\ell(\sigma)$, $R=\ell(\tau)$.  By the inductive hypothesis,
    there are labeled precubical maps $g_Q:\pobj{Q}\to X$ and
    $g_R:\pobj{R}\to X$ such that $g_Q(i_{\pobj{Q}})=x$,
    $g_R(f_{\pobj{R}})=y$, and
    $g_Q(f_{\pobj{Q}})=g_R(i_{\pobj{R}})=z$.  The last equality,
    together with Lemma \ref{le:pobjJoin}, guarantees that $g_Q$ and
    $g_R$ glue to a map $g: \pobj{P}\to X$.  It is clear that
    $g(i_{\pobj{P}})=g_Q(i_{\pobj{Q}})=x$ and
    $g(f_{\pobj{P}})=g_R(f_{\pobj{R}})=y$. \qedhere
  \end{itemize}
\end{proof}

\begin{proposition}
  \label{pr:internal-track}
  For any interval ipomset $P$ there exists a track
  $\rho:i_{\pobj{P}}\leadsto f_{\pobj{P}}$ such that
  $\ell(\rho)\cong P$.
\end{proposition}

\begin{proof}
  If $P=\subid{S}{U}{T}$ is discrete, then
  $\ell(( i_{\pobj{P}}, \yoneda_U, f_{\pobj{P}}))\cong P$.  If $P$ is
  not discrete, then there is a presentation $P=Q*R$ (Proposition
  \ref{pr:DecompositionOfEvilio}). If
  $\sigma:i_{\pobj{Q}}\leadsto f_{\pobj{Q}}$ is a track with
  $\ell(\sigma)\cong Q$ and $\tau:i_{\pobj{R}}\leadsto f_{\pobj{R}}$ a
  track with $\ell(\tau)\cong R$, then
  \begin{equation*}
    \ell(\jneda^0_{Q\subseteq
      P}(\sigma)*\jneda^1_{R\subseteq P}(\tau))\cong Q*R= P
  \end{equation*}
  by Lemma \ref{le:pobjJoin}.
\end{proof}

\begin{example}
  We follow up on Example \ref{ex:N-to-HDA}.  For
  \begin{equation*}
    P= \!\!\ipomset[3]{& a \ar[r] & b & \\ & c \ar[r] \ar@{.>}[u] \ar[ur] & d
      \ar@{.>}[u] \ar@{.>}[ul] &}\!\!,
  \end{equation*}
  we have $i_{ \pobj{P}}= x_1$ and $f_{ \pobj{P}}= x_8$ (see also
  Figure \ref{fi:pobj-ex}), and the track $\rho$ of the proposition is
  given by $\rho=( x_1, z_1, y_3, z_2, y_8, z_3, x_8)$.  If we add
  interfaces to $P$, for example
  \begin{equation*}
    Q= \ipomset[1.5]{ \ar[r] & a \ar[rr] && b & \\ & c \ar[rr]
      \ar@{.>}[u] \ar[urr] && d \ar[r]
      \ar@{.>}[u] \ar@{.>}[ull] &},
  \end{equation*}
  then $i_{ \pobj{Q}}= y_1$, $f_{ \pobj{Q}}= y_{10}$, and $\rho=( y_1,
  z_1, y_3, z_2, y_8, z_3, y_{10})$.
\end{example}

\section{The Geometric View}
\label{se:geo}

Precubical sets may be realized as directed topological spaces, and
then directed paths through these spaces give an intuitive model of
computations.  In this section we first recap the geometric
realization and then introduce labels of directed paths in HDAs.  We
will see that for every directed path there exists a track with the
same label, and vice versa, so that HDA languages defined using tracks
and using directed paths are the same.

\subsection{Geometric Realization}

Recall that the \emph{concatenation} $\alpha* \beta$ of two paths
$\alpha, \beta: I=[ 0, 1]\to \X$ in a topological space $\X$ is
defined, if $\alpha( 1)= \beta( 0)$, as
\begin{equation*}
  \alpha* \beta( t)=
  \begin{cases}
    \alpha( 2 t) &\text{for } t\le \frac{ 1}{ 2}, \\
    \beta( 2 t- 1) &\text{for } t\ge \frac{ 1}{ 2}.
  \end{cases}
\end{equation*}

A \emph{directed topological space}, or \emph{d-space} \cite{book/Grandis09}
is a pair $( \X, \po P \X)$ consisting of a topological space $\X$ and
a set $\po P \X\subseteq \X^I$ of paths in $\X$ such that $\po P \X$
\begin{itemize}
\item contains all constant paths;
\item is closed under concatenation: if $\alpha, \beta\in \po P \X$
  and $\alpha( 1)= \beta( 0)$, then $\alpha* \beta\in \po P \X$;
\item is closed under reparametrization and subpath: for any
  $\alpha\in \po P \X$ and $h: I\to I$ continuous and (weakly)
  increasing, also $\alpha\circ h\in \po P \X$.
\end{itemize}
The elements of $\po P\X$ are called directed paths or \emph{d-paths}.

Prominent examples of d-spaces are the directed interval
$\po I=[ 0, 1]$ with the natural ordering on the real numbers and the
directed $n$-cubes $\po I^n$ for $n\ge 0$.  Similarly, there are
directed Euclidean spaces $\po \Real^n$ for all $n\ge 0$.  In each of
these, the d-paths a precisely the paths which are (weakly) increasing
in each coordinate, that is, $\alpha: I\to \po \Real^n$ is a d-path
iff $t_1\le t_2$ implies $\alpha( t_1)\le \alpha( t_2)$ in the usual
ordering $( x_1,\dotsc, x_n)\le( y_1,\dotsc, y_n)$ iff $x_i\le y_i$
for all $i$.

Morphisms $f:( \X, \po P\X)\to( \Y, \po P\Y)$ of d-spaces are
\emph{d-maps}; they are those continuous functions that also preserve
directedness, \ie $f\circ \alpha\in \po P \Y$ for all
$\alpha\in \po P\X$.  For any d-space $(\X, \po P\X)$ we have
$\po P\X= \X^{ \po I}$ as function spaces.

The so-defined category $\dTop$ of d-spaces is complete and cocomplete
\cite{book/Grandis09}.  In particular, quotients of d-spaces are
well-defined.  If $\X$ is a d-space and $\sim$ an equivalence on $\X$\!,
then d-paths in the quotient space $\X/ \mathord\sim$ are of the form
\begin{equation*}
  \alpha=(\pi(\beta_1)*\dotsm *\pi(\beta_m))\circ h
\end{equation*}
where all $\beta_i$ are d-paths in $\X$ such that
$\beta_i(1)\sim \beta_{i+1}(0)$ and $h: \po I\to \po I$ is a
surjective d-map.

Surjective d-maps $h: \po I\to \po I$ as above are called
\emph{reparametrizations} and will play a central role below.

\begin{definition}
  \label{de:georel}
  The \emph{geometric realization} of a precubical set $X$ is the
  d-space
  \begin{equation*}
    \georel{ X}= \bigsqcup_{ n\ge 0} X_n\times \smash[t]{\po I^n} /
    \mathord\sim,
  \end{equation*}
  where the equivalence relation $\sim$ is generated by
  \begin{equation*}
    ( \delta_i^\nu x,( t_1,\dotsc, t_{ n- 1}))\sim( x,( t_1,\dotsc, t_{
      i- 1}, \nu, t_{ i+ 1},\dotsc, t_{ n- 1})).
  \end{equation*}
  The geometric realization of a precubical map $f:X\to Y$ is the
  d-map $\georel{f}:\georel{X}\to \georel{Y}$ given by
  $\georel{f}([ x,( t_1,\dotsc, t_n)])=[ f( x),( t_1,\dotsc,
  t_n)]$.
\end{definition}

Above, $[ x,( t_1,\dotsc, t_n)]$ is used to denote equivalence classes
of $\sim$.  Geometric realization is a functor from $\pcSet$ to $\dTop$.

\begin{example}
  The geometric realization of the $n$-cube $\square^n$ is the
  directed cube $\po I^n$.  The purpose of the equivalence relation
  $\sim$ in the definition is to embed faces as subspaces, for
  example, the elementary face $\delta_1^0 \yoneda_n$ of the top cell
  of $\square^n$ is the subset
  $\{( 0, t_2,\dotsc, t_n)\mid 0\le t_i\le 1\}\subseteq I^n$.
\end{example}

The \emph{interior image} $\intimg{ x}\subseteq| X|$ of a cell
$x\in X_n$ in a precubical set $X$ is defined as
\begin{equation*}
  \intimg{ x}= \{[ x,( t_1,\dotsc, t_n)]\mid 0< t_i< 1 \text{ for all
  } i\in[ n]\}.
\end{equation*}
The set
$\intimg{ x}$ is open for $x\notin X_0$; for $x\in X_0$,
$\intimg x=\{ x\}$.

\begin{definition}
  The \emph{carrier} $\carr(p)$ of a point $p\in \georel{ X}$ is the
  unique cell $x\in X$ such that $p\in \intimg x$.
\end{definition}

For later use we record the following lemma, whose proof easily
follows from the definition; see also \cite{Fahrenberg05-thesis}:

\begin{lemma}
  \label{le:fcarrp}
  For a precubical map $f: X\to Y$ and $p\in \georel{X}$, $\carr(
  \georel{f}( p))= f( \carr( p))$. \qed
\end{lemma}

We conclude with a description of d-paths on $\georel{X}$.  Recall the
Yoneda inclusions $\ineda_x: \square^n\to X$ from
Lemma \ref{le:ineda}. These induce d-maps
$\georel{ \ineda_x}: \po I^n\to \georel{ X}$.

\begin{lemma}
  \label{le:DipathPresentation}
  Every d-path $\alpha\in \po P\georel{X}$ has a presentation
  \begin{equation}
    \label{eq:DipathPresentation} \tag{$\ast$}
    \alpha= \big( ( \georel{\ineda_{x_1}}\circ \beta_1)*(
    \georel{\ineda_{x_2}}\circ \beta_2)*\dotsm*(
    \georel{\ineda_{x_m}}\circ \beta_m) \big)\circ h,
  \end{equation}
  where $x_1,\dotsc,x_m\in X$, $\beta_i\in \po P (\po I^{\dim x_i})$,
  $[x_i,\beta_i(1)]=[x_{i+1},\beta_{i+1}(0)]\in \georel{X}$, and
  $h:\po I\to \po I$ is a reparametrization. Moreover, we can assume
  that $\carr(\beta_i(\tfrac{1}{2}))=\yoneda_{\dim(x_i)}$.
\end{lemma}

\begin{figure}
  \centering
  \begin{tikzpicture}[x=1.5cm, y=1.5cm]
    \begin{scope}
      \path[fill=black!15] (0,0) to (3,0) to (3,2) to (2,2) to (2,1)
      to (0,1);
      \node[state] (00) at (0,0) {};
      \node[state] (10) at (1,0) {};
      \node[state] (20) at (2,0) {};
      \node[state] (30) at (3,0) {};
      \node[state] (01) at (0,1) {};
      \node[state] (11) at (1,1) {};
      \node[state] (21) at (2,1) {};
      \node[state] (31) at (3,1) {};
      \node[state] (22) at (2,2) {};
      \node[state] (32) at (3,2) {};
      \path (00) edge (10);
      \path (10) edge (20);
      \path (20) edge (30);
      \path (01) edge (11);
      \path (11) edge (21);
      \path (21) edge (31);
      \path (22) edge (32);
      \path (00) edge (01);
      \path (10) edge (11);
      \path (20) edge (21);
      \path (30) edge (31);
      \path (21) edge (22);
      \path (31) edge (32);
      \node at (.5,.5) {$x_1$};
      \node at (1.5,.5) {$x_2$};
      \node at (2.5,.5) {$x_3$};
      \node at (2.5,1.5) {$x_4$};
      \draw[-, very thick, orange] plot[smooth] coordinates {(.3,.2) (1.6,.3)
        (2.3,.7) (2.8,1.7)};
    \end{scope}
    \begin{scope}[shift={(5,0)}]
      \path[fill=black!15] (0,0) to (3,0) to (3,2) to (2,2) to (2,1)
      to (0,1);
      \node[state] (00) at (0,0) {};
      \node[state] (10) at (1,0) {};
      \node[state] (20) at (2,0) {};
      \node[state] (30) at (3,0) {};
      \node[state] (01) at (0,1) {};
      \node[state] (11) at (1,1) {};
      \node[state] (21) at (2,1) {};
      \node[state] (31) at (3,1) {};
      \node[state] (22) at (2,2) {};
      \node[state] (32) at (3,2) {};
      \path (00) edge (10);
      \path (10) edge (20);
      \path (20) edge (30);
      \path (01) edge (11);
      \path (11) edge (21);
      \path (21) edge (31);
      \path (22) edge (32);
      \path (00) edge (01);
      \path (10) edge (11);
      \path (20) edge (21);
      \path (30) edge (31);
      \path (21) edge (22);
      \path (31) edge (32);
      \node at (.5,.5) {$x_1$};
      \node at (1.5,.5) {$x_2$};
      \node at (2.5,.5) {$x_3$};
      \node at (2.5,1.5) {$x_4$};
      \draw[-, very thick, orange] plot[smooth] coordinates {(.3,.2) (1.6,.3)
        (2,.5)};
      \draw[-, very thick, orange] (2,.5) -- (2,1) -- (2.45,1);
      \draw[-, very thick, orange] plot[smooth] coordinates {(2.45,1)
        (2.8,1.7)};
    \end{scope}
  \end{tikzpicture}
  \bigskip\bigskip\medskip
  \caption{Left: d-path $\alpha$ with presentation
    $(( \georel{ \ineda_{ x_1}}\circ \beta_1)*( \georel{ \ineda_{
        x_2}}\circ \beta_2)*( \georel{ \ineda_{ x_3}}\circ \beta_3)*(
    \georel{ \ineda_{ x_4}}\circ \beta_4))\circ h$ (Lemma
    \ref{le:DipathPresentation}); right: counterexample in the proof
    of Lemma \ref{le:DipathPresentation}, with presentation
    $(( \georel{ \ineda_{ x_1}}\circ \beta_1)*( \georel{ \ineda_{
        x_2}}\circ \beta_2)*$ $( \georel{ \ineda_{ \delta_1^0 x_3}}\circ
    \gamma_1)*$ $( \georel{ \ineda_{ \delta_1^1 x_3}}\circ \gamma_2)*(
    \georel{ \ineda_{ x_4}}\circ \beta_4))\circ h''$.}
  \label{fi:dpath-reali}
\end{figure}
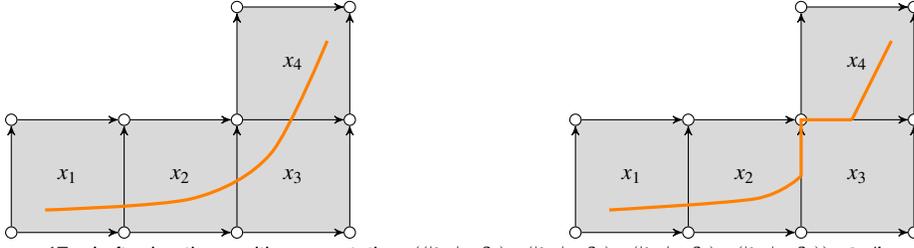

Figure \ref{fi:dpath-reali} shows an example: on the left, a d-path
and a presentation; on the right, the counterexample used below in the
proof.

\begin{proof}
  Apart from the last statement, this follows immediately from the
  description of d-paths on quotient d-spaces and the definition of
  the geometric realization.
	
  Let $\mathfrak{S}$ be the set of sequences $(d_0,d_1,\dotsc)$ of
  natural numbers which are eventually vanishing, that is, there
  exists $n\ge 0$ such that $d_j=0$ for all $j> n$.  Equip
  $\mathfrak{S}$ with the reverse lexicographic order, \ie
  $(d_j)<(d'_j)$ if there exists $n$ such that $d_n<d'_n$ and
  $d_j=d'_j$ for $j> n$.  For every presentation
  \eqref{eq:DipathPresentation} of $\alpha$ we associate the sequence
  $(d_j)\in \mathfrak{S}$ such that $d_j$ is the number of indices $i$
  such that $\dim(x_i)=j$.  Choose a presentation
  \eqref{eq:DipathPresentation} with a minimal associated sequence
  $(d_j)\in \mathfrak{S}$.  Denote $n_i=\dim x_i$.

  Assume that for some $i$,
  $\beta_i(t)\not \in \oil \yoneda_{n_i}\oir$ for all $t$.  But then
  $\beta_i\in \po P\georel{\square^{ n_i}_{ n_i- 1}}$, the set of
  d-paths in the $( n_i- 1)$-restriction of $\square^{ n_i}$, and
  hence it has a presentation
  \begin{equation*}
    \beta_i= \big( ( \georel{\ineda_{y_1}}\circ \gamma_1)*\dotsm*(
    \georel{\ineda_{y_l}}\circ \gamma_l) \big) \circ h'.
  \end{equation*}
  Obviously $\dim(y_k)<n_i$ for all $k$. Collecting these two
  presentations, we have
  \begin{align*}
    \alpha &= \big( ( \georel{ \ineda_{ x_1}}\circ \beta_1)*\dotsm*(
    \georel{ \ineda_{ x_{ i-1}}}\circ \beta_{ i-1})*( \georel{
      \ineda_{ x_i}}\circ(( \georel{\ineda_{y_1}}\circ
    \gamma_1)*\dotsm \\
    &\hspace*{9em}
    \dotsm*( \georel{\ineda_{y_l}}\circ \gamma_l)) \circ h')*(
    \georel{\ineda_{x_{i+1}}}\circ \beta_{i+1})*\dotsm*(
    \georel{\ineda_{x_{m}}}\circ \beta_{m}) \big) \circ h \\
    &= \big( ( \georel{\ineda_{x_1}}\circ \beta_1)*\dotsm*(
    \georel{\ineda_{x_{i-1}}}\circ \beta_{i-1})*(
    \georel{\ineda_{\ineda_{x_i}(y_1)}}\circ \gamma_1)*\dotsm \\
    &\hspace*{9em}
    \dotsm*( \georel{\ineda_{\ineda_{x_i}(y_l)}}\circ \gamma_l)*(
    \georel{\ineda_{x_{i+1}}}\circ \beta_{i+1})*\dotsm*(
    \georel{\ineda_{x_{m}}}\circ \beta_{m}) \big) \circ h''
  \end{align*}
  for some reparametrization $h''$ obtained from $h$ and $h'$.  Let
  $(d'_j)$ be the associated sequence of this presentation.  Then
  $d'_j=d_j$ for $j>n_i$ and $d'_{n_i}=d_{n_i}-1$, since $x_i$ no
  longer appears in the presentation and cells $\ineda_{ x_i}(y_k)$
  have smaller dimensions: a contradiction to the minimality of $( d_j)$.

  As a consequence, for every $i$ there exists $t_i$ with
  $\carr(\beta_i(t_i))=\yoneda_{n_i}$. By reparametrizing $\beta_i$
  and adjusting $h''$ we can ensure that $t_i=\tfrac{1}{2}$.
\end{proof}

\subsection{Labels of d-paths}

For the rest of this section, $( X, \lambda)$ is a labeled precubical
set (which is, by definition, event consistent).  We will associate to
every d-path $\alpha$ in $\georel{X}$ its \emph{label} $\ell(\alpha)$
as an interval ipomset.  In order to do so, we first need to find the
(universal) events in $X$ that are active during the execution
$\alpha$.


We say that an event $e\in E_X$ is \emph{active} at point
$p=[ x,( t_1,\dotsc, t_n)]\in \georel{ X}$, for $x\in X_n$, if there
is $i\in[n]$ such that $\ev_i(x)=e$ and $0< t_i< 1$. Otherwise, $e$ is
\emph{inactive} at $p$. It is easy to verify that this does not depend
on the choice of a presentation of $p$. Let $U^X_e\subseteq \georel{X}$
be the set of points in which $e$ is active.  The following is clear.

\begin{lemma}
  \label{le:UeX}
  $U^X_e=\bigcup \{\intimg{x}\mid e\in \ev(x) \}=\{p\in \georel{X}\mid
  e\in \ev(\carr(p))\}$. \qed
\end{lemma}

Note that all events are inactive at vertices, exactly one event is
active along an edge, and so on: if $\dim(\carr(p))=n$, then exactly
the $n$ events in $\ev(\carr(p))$ are active at $p$.  We will write
$U_e$ for $U_e^X$ when $X$ is clear.

Now fix a d-path $\alpha\in \po P\georel{X}$.  For every event
$e\in E_X$, let
\begin{equation*}
  J_e^\alpha=\alpha^{-1}(U_e)\subseteq[ 0, 1]
\end{equation*}
be the set of points in time when $e$ is active.  $J_e^\alpha$ is an
open subset of $[0,1]$, since $U_e$ is open.  Moreover, by Lemma
\ref{le:DipathPresentation} it has a finite number of connected
components.  Thus, there is a unique presentation
\begin{equation}
  \label{eq:int-arr}
  J_e^\alpha= I_{e,1}^\alpha\cup \dots\cup I^\alpha_{e,n_e^\alpha}
\end{equation}
as a union of connected components ordered increasingly.  Each of
these components is open in $[0,1]$, though not necessarily in
$\Real$: possibly $I^\alpha_{e,1}=[0,t\oir$ or
$T^\alpha_{e,n^\alpha_e}=\oil t,1]$, or even $I^\alpha_{e,1}=[0,1]$
for $n^\alpha_e=1$.  The collection of presentations
\eqref{eq:int-arr} is called the \emph{interval arrangement} of
$\alpha$.

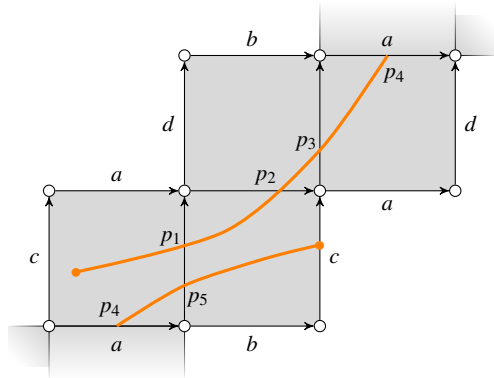
\begin{figure}
  \centering
  \begin{tikzpicture}[x=1.8cm,y=1.8cm]
    \path[fill=black!15] (0,0) to (2,0) to (2,1) to (3,1) to (3,2) to
    (1,2) to (1,1) to (0,1);
    \fill[-, fill=black!15, path fading=east, fading
    transform={rotate=45}] (3,2) -- (3.3,2) -- (3.3,2.3) -- (3,2.3) --
    (3,2);
    \filldraw[-, fill=black!15, path fading=north] (2,2.4) -- (2,2) --
    (3,2) -- (3,2.4);
    \draw[-, path fading=east] (3,2) -- (3.3,2);
    \fill[-, fill=black!15, path fading=west, fading
    transform={rotate=45}] (0,0) -- (-.3,0) -- (-.3,-.3) -- (0,-.3) --
    (0,0);
    \filldraw[-, fill=black!15, path fading=south] (0,-.4) -- (0,0) --
    (1,0) -- (1,-.4);
    \draw[-, path fading=west] (0,0) -- (-.3,0);
    \node[state] (00) at (0,0) {};
    \node[state] (10) at (1,0) {};
    \node[state] (20) at (2,0) {};
    \node[state] (01) at (0,1) {};
    \node[state] (11) at (1,1) {};
    \node[state] (21) at (2,1) {};
    \node[state] (31) at (3,1) {};
    \node[state] (12) at (1,2) {};
    \node[state] (22) at(2,2) {};
    \node[state] (32) at (3,2) {};
    \path (00) edge node[below] {$\vphantom{b}a$} (10);
    \path (10) edge node[below] {$b$} (20);
    \path (01) edge node[above] {$a$} (11);
    \path (11) edge (21);
    \path (21) edge node[below] {$a$} (31);
    \path (12) edge node[above] {$b$} (22);
    \path (22) edge node[above]{$a$} (32);
    \path (00) edge node[left] {$c$} (01);
    \path (10) edge (11);
    \path (11) edge node[left] {$d$} (12);
    \path (20) edge node[right] {$c$} (21);
    \path (21) edge (22);
    \path (31) edge node[right] {$d$} (32);
    \draw[-, very thick, orange] plot[smooth] coordinates {(.2,.4) (1.3,.7)
      (2,1.3) (2.5,2)};
    \draw[-, very thick, orange] plot[smooth] coordinates {(.5,0) (1,.3)
      (1.6,.5) (2,.6)};
    \node[orange] at (.2,.39) {$\bullet$};
    \node[orange] at (2,.59) {$\bullet$};
    \node at (.9,.65) {$p_1$};
    \node at (1.6,1.1) {$p_2$};
    \node at (1.9,1.35) {$p_3$};
    \node at (2.55,1.85) {$p_4$};
    \node at (.45,.12) {$p_4$};
    \node at (1.1,.2) {$p_5$};
  \end{tikzpicture}
  \bigskip
  \caption{Directed path which wraps around a two-dimensional loop
    (bottom left and top right edges identified).}
  \label{fi:dpath-loop}
\end{figure}

\begin{example}
  Figure \ref{fi:dpath-loop} shows a d-path $\alpha$ through a labeled
  precubical set with a two-dimensional loop: $\alpha$ starts inside
  the bottom-left square with events $a$ and $c$, continues until the
  upper face of the top-right square, which is identified with the
  lower face of the bottom-left square, and finishes in the right
  $c$-labeled edge.  Assuming that $\alpha$ is parametrized so that
  $\alpha( \frac{i}{6})= p_i$ for $i\in[ 5]$ (the intersection points
  $p_1,\dotsc, p_5$ of $\alpha$ with the edges are indicated in the
  figure), its interval arrangement is
  \begin{equation*}
    J^\alpha_a= \bigl[ 0, \tfrac{1}{6}\bigoir\cup \bigoil \tfrac{1}{2},
    \tfrac{5}{6} \bigr], \qquad%
    J^\alpha_b= \bigoil \tfrac{1}{6}, \tfrac{1}{2} \bigoir\cup \bigoil
    \tfrac{5}{6}, 1\bigoir, \qquad%
    J^\alpha_c= \bigl[ 0, \tfrac{1}{3}\bigoir\cup \bigoil \tfrac{2}{3},
    1\bigr], \qquad%
    J^\alpha_d= \bigoil \tfrac{1}{3}, \tfrac{2}{3} \bigoir.
  \end{equation*}
\end{example}

Now, for every $x\in X$, define a relation $\intord_x$ on $\ev(x)$ by
$e\intord_x e'$ if $e=\ev_i(x)$ and $e'=\ev_j(x)$ for $i<j$.  From
Lemma \ref{le:EventsAreSubsets} we immediately get

\begin{lemma}
  If $e, e'\in \delta_i^\nu x$ for some $i$ and $\nu$, then
  $e\intord_{ \delta_i^\nu x} e'$ iff $e\intord_x e'$. \qed
\end{lemma}

As a consequence, on every connected component
$C\subseteq U_e\cap U_{e'}$ there is a well-defined relation
$\intord_C$ between $e$ and $e'$ (although it may differ between
different components).  We write $\mathord{\intord_p}=
\mathord{\intord_C}$ for any point $p\in C$.

\begin{definition}
  \label{de:label-dpath}
  The \emph{label} of $\alpha$ is the ipomset
  $\ell( \alpha)=( P, <_P, \intord_P, \lambda_P, S_P, T_P)$ given
  as follows:
  \begin{itemize}
  \item $P=\{(e,i)\mid e\in E_X, 1\leq i\leq n^\alpha_e\}$;
  \item $(e,i)<_P (e',i')$ if $I^\alpha_{e,i}<I^\alpha_{e',i'}$;
  \item $\intord_P$ is the transitive closure of the relations
    $\smash[t]{( e, i)\intord_{ \alpha( t)}( e', i')}$ for
    $\smash[b]{t\in I^\alpha_{e,i}\cap I^\alpha_{e',i'}}$ (this does
    not depend on the choice of $t$ since
    $I^\alpha_{e,i}\cap I^\alpha_{e',i'}$ is connected);
  \item $\lambda_P((e,i))=\lambda^{\ev}(e)$,
    $S_P=\{(e,i)\in P\mid 0\in I^\alpha_{e,i}\}$, and
    $T_P=\{(e,i)\in P\mid 1\in I^\alpha_{e,i}\}$.
  \end{itemize}
\end{definition}

Hence all elements of $S_P$ are of the form $(e,1)$ and all elements
of $T_P$ are of the form $(e,n^\alpha_e)$.  Further,
$S_P\cong \carr(\alpha(0))$ and $T_P\cong \carr(\alpha(1))$ as linear
posets.

\begin{proposition}
  \label{pr:lamalpha}
  The label $\ell( \alpha)$ is an interval ipomset.
\end{proposition}

\begin{proof}
  By definition, $<_P$ is an interval order, $S_P$ contains only
  $<_P$-minimal elements, and $T_P$ contains only $<_P$-maximal
  elements.  Assume that $(e,i)$ and $(e',i')$ are $<_P$-incomparable,
  then $I^\alpha_{e,i}\cap I^\alpha_{e',i'}\ne \emptyset$.  Let
  $t\in I^\alpha_{e,i}\cap I^\alpha_{e',i'}$, then
  $( e, i)\intord_{ \alpha( t)}( e', i')$ or
  $( e', i')\intord_{ \alpha( t)}( e, i)$, hence $(e,i)$ and $(e',i')$
  are $\intord_P$-comparable.

  It remains to show that $\intord_P$ is irreflexive. So let
  \begin{equation*}
    (e_1,i_1)\intord_{ \alpha(t_1)} \dotsm \intord_{ \alpha( t_{ r-
        1})} (e_r,i_r)\intord_{ \alpha( t_r)} (e_1,i_1)
  \end{equation*}
  be a shortest loop of elementary relations and denote
  $H_k=I^\alpha_{e_k,i_k}= \oil a_k, b_k\oir$ (or $[ a_k, b_k\oir$,
  $\oil a_k, b_k]$, or $[ a_k, b_k]$, in case $a_k= 0$ or $b_k= 1$;
  this will not matter for our argument below).

  We have $H_k\cap H_{k+1}\neq\emptyset$ for $k\in [r-1]$, and also
  $H_r\cap H_1\neq \emptyset$.  On the other hand,
  $H_k\cap H_l=\emptyset$ for $k<l-1$ and $(k,l)\neq (1,r)$; otherwise
  we can construct a shorter loop.  Further,
  $H_1\cap\dots\cap H_r=\emptyset$; otherwise, these elements would be
  linearly ordered by $\intord_{ \alpha( t)}$ for some
  $t\in \bigcap H_k$.

  We show that for every $k$, $H_{ k+ 1}$ is either to the right or to
  the left of $H_k$.  Let $k\in[ r- 2]$ and assume
  $H_{ k+ 1}\subseteq H_k$.  Then
  $H_k\cap H_{ k+ 2}\supseteq H_{ k+ 1}\cap H_{ k+ 2}\ne \emptyset$,
  forcing $k= 1$ and $r= k+ 2= 3$; but now also
  $H_1\cap H_2\cap H_3\ne \emptyset$, a contradiction.  A similar
  contradiction is obtained when assuming $H_k\subseteq H_{ k+ 1}$,
  and also for $H_r\subseteq H_{ r- 1}$, $H_{ r- 1}\subseteq H_r$,
  $H_1\subseteq H_r$, and $H_r\subseteq H_1$.
  
  \begin{figure}
    \centering
    \begin{tikzpicture}[-, y=.5cm]
      \path (0,0) edge node[below] {$H_1$} (3,0);
      \path (2.5,-1) edge node[below] {$H_2$} (5.5,-1);
      \path (5,-2) edge node[below] {$H_3$} (8,-2);
      \node[left] at (0,0) {$\vphantom{b}a_1$};
      \node[right] at (3,0) {$b_1$};
      \node[left] at (2.5,-1) {$\vphantom{b}a_2$};
      \node[right] at (5.5,-1) {$b_2$};
      \node[left] at (5,-2) {$\vphantom{b}a_3$};
      \node[right] at (8,-2) {$b_3$};
      \node at (7.3,-2.7) {$\ddots$};
    \end{tikzpicture}
    \caption{Progression of intervals in the proof of Proposition
      \ref{pr:lamalpha}.}
    \label{fi:interval-cycle}
  \end{figure}
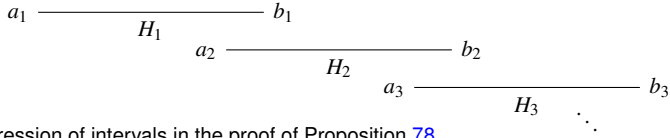

  Now assume that $H_2$ is to the \emph{right} of $H_1$ (the argument
  for the other case is similar), then $a_1< a_2< b_1< b_2$, see
  Figure \ref{fi:interval-cycle} for an illustration.

  We proceed by induction.  Let $k\in[ r- 2]$ and assume $H_{ k+ 1}$
  is to the right of $H_k$, then $a_k< a_{ k+ 1}< b_k< b_{ k+ 1}$.  We
  show that also $H_{ k+ 2}$ is to the right of $H_{ k+ 1}$.  Assume
  otherwise, then $a_{ k+2}< a_{ k+ 1}< b_{ k+ 2}$, hence
  $a_{ k+ 2}< b_k$ and $a_k< b_{ k+ 2}$, which implies
  $H_k\cap H_{ k+ 2}\ne \emptyset$, again forcing $k= 1$ and
  $r= k+ 2= 3$ and then a contradiction.

  Hence if $H_2$ is to the right of $H_1$, then the sequence of
  intervals $H_1,\dotsc, H_r$ proceeds to the right; but the same
  argument as above then also shows that $H_1$ is to the right of
  $H_r$ which is impossible.  Similarly, if $H_2$ is to the left of
  $H_1$, then the sequence proceeds to the left, and $H_1$ then has
  the impossible task of being to the left of $H_r$.  Overwhelmed by
  contradictions, we are forced to accept that $\intord_P$ is
  irreflexive.
\end{proof}


\subsection{Properties of d-path labels}

The main goal of this section is to prove that for every d-path
$\alpha$ in $\georel{X}$ there is a track $\rho$ in $X$ with the same
labeling and vice versa.  First, we show several properties of labels
of d-paths.

\begin{lemma}\label{le:PathLabelRepar}
  Let $\alpha\in \po P\georel{X}$ and $h: \po I\to \po I$ a
  (surjective) reparametrization.  Then
  $\ell(\alpha)\cong \ell(\alpha\circ h)$.
\end{lemma}

\begin{proof}
  If
  $J^\alpha_e=I_{e,1}^\alpha\cup \dots\cup I^\alpha_{e,n_e^\alpha}$,
  then
  \begin{equation*}
    J^{\alpha\circ h}_e=h^{-1}(J^\alpha_e)= h^{-1}(
    I_{e,1}^\alpha)\cup \dots\cup h^{-1}(I^\alpha_{e,n_e^\alpha})
  \end{equation*}
  is a presentation as a union of connected components, so that
  $n_e^{\alpha\circ h}=n_e^\alpha$ and
  $I^{\alpha\circ h}_{e,i}=h^{-1}(I^\alpha_{e,i})$.  The result
  follows from the definition of d-path label.
\end{proof}

\begin{lemma}
  \label{le:FunctorialUe}
  Let $f:X\to Y$ be a map of labeled precubical sets and $e\in E_Y$. Then
  \begin{equation*}
    \georel{f}^{-1}(U^Y_e)= \bigsqcup_{\smash[b]{e'\in E_f^{-1}(e)}} U^X_{e'}
  \end{equation*}
  as a disjoint union.
\end{lemma}

\begin{proof}
  For $p\in \georel{X}$ we have
  \begin{alignat*}{2}
    p\in \georel{f}^{-1}(U^Y_e) &\IFF \georel{f}(p)\in U^Y_e \\
    &\IFF e\in \ev(\carr(\georel{f}(p))) &\qquad&\eqref{le:UeX} \\
    &\IFF e\in \ev(f(\carr(p))) &&\eqref{le:fcarrp} \\
    &\IFF e\in E_f(\ev(\carr(p))) &&\eqref{le:CellParmsetFun} \\
    &\IFF \exists e'\in E_f^{-1}(e): e'\in\ev(\carr(p)) \\
    &\IFF \exists e'\in E_f^{-1}(e): p\in U^X_{e'}.  &&\eqref{le:UeX}
  \end{alignat*}
  If $p\in U^X_{e'}\cap U^X_{e''}$ for $e'\neq e''\in E_X$, then
  $e', e''\in \ev(\carr(p))$. By Lemma \ref{le:CellParmsetFun},
  $E_f(e'),E_f(e'')\in \ev(\carr( \georel{f}(p)))$, so $E_f(e')\neq
  E_f(e'')$. Consequently, $e'$ and $e''$ cannot both belong to
  $E_f^{-1}(e)$.
\end{proof}

\begin{lemma}
  \label{le:PathLabelFun}
  For any map of labeled precubical sets $f:X\to Y$ and
  $\alpha\in \po P\georel{X}$,
  $\ell(\alpha)\cong \ell(\georel{f}\circ \alpha)$.
\end{lemma}

\begin{proof}
  By Lemma \ref{le:FunctorialUe} there is a bijection between
  connected components of $J^{\georel{f}\circ \alpha}_{e}$ and
  $\bigcup_{e'\in E_f^{-1}(e)} J^\alpha_{e'}$ for every $e\in
  E_Y$. These induce a bijection between the ipomsets
  $\ell(\georel{f}\circ \alpha)$ and $\ell(\alpha)$.  It is easy to
  check that this is an ipomset isomorphism.
\end{proof}

\begin{proposition}\label{le:PathLabelComp}
  Let $\alpha,\beta\in \po P\georel{X}$ be such that
  $\alpha( 1)= \beta( 0)$. Then
  $\ell(\alpha*\beta)=\ell(\alpha)*\ell(\beta)$.
\end{proposition}

\begin{proof}
  Let $p=\alpha(1)=\beta(0)$ and $l,r:[0,1]\to [0,1]$,
  $l(t)=\tfrac{t}{2}$, $r(t)=\tfrac{t+1}{2}$.  Then, for each $e\in
  E_X$,
  \begin{equation*}
    J^{\alpha*\beta}_e = \tfrac{1}{2}J^\alpha_e \cup
    (\tfrac{1}{2}(J^\beta_e+\tfrac{1}{2})=l(J^\alpha_e)\cup
    r(J^\beta_e).
  \end{equation*}
  If $e\not\in \carr(p)$, then $1\not\in J^\alpha_e$ and
  $0\not\in J^\beta_e$. Thus $l(J^\alpha_e)$ and
  $r(J^\beta_e)$ are disjoint,
  $n_e^{\alpha*\beta}=n_e^\alpha+n_e^\beta$ and
  \begin{equation*}
    I^{\alpha*\beta}_{e,i}=
    \begin{cases}
      l(I^\alpha_{e,i}) & \text{for $1\leq i\leq n^\alpha_e$,}\\
      r(I^\beta_{e,i-n^\alpha_e}) & \text{for
        $n^\alpha_e<i\leq n^{\alpha*\beta}_e$.}
    \end{cases}
  \end{equation*}
  If $e\in \carr(p)$, then $1\in J^\alpha_e$ and $0\in
  J^\beta_e$.
  Therefore $l(J^\alpha_e)$ and $r(J^\beta_e)$ are glued along
  $\smash{\frac{1}{2}}$ and consequently
  $n_e^{\alpha*\beta}=n_e^\alpha+n_e^\beta-1$ and
  \begin{equation*}
    I^{\alpha*\beta}_{e,i}=
    \begin{cases}
      l(I^\alpha_{e,i}) & \text{for $1\leq i< n^\alpha_e$},\\
      l(I^\alpha_{e,n^\alpha_e})\cup r(I^\beta_{e,1}) & \text{for $i=n^\alpha_e$},\\
      r(I^\beta_{e,i-n^\alpha_e}) & \text{for
        $n^\alpha_e<i< n^{\alpha*\beta}_e$}.
    \end{cases}
  \end{equation*}

  It follows that the maps
  $i^\alpha_e:\ell(\alpha)\ni (e,i)\mapsto (e,i)\in
  \ell(\alpha*\beta)$ and
  \[
    i^\beta_e:\ell(\beta)\ni (e,i) \mapsto
    \begin{cases}
      (e,i+n^\alpha_e) & \text{if $e\not\in \carr(p)$}\\
      (e,i+n^\alpha_e-1) & \text{if $e\in \carr(p)$}
    \end{cases}
  \]
  glue to the bijection
  $i:\ell(\alpha)*\ell(\beta)\to \ell(\alpha*\beta)$. It is elementary
  to check that $i$ is an ipomset isomorphism.
\end{proof}

We record the following easy fact for use in the next proof.

\begin{lemma}
  \label{le:PathLabelCube}
  Let $S$ be a linear pomset and $\alpha\in\po P\georel{\square^S}$ a path
  such that $\carr(\alpha(t))=\yoneda_S$ for some $t\in[0,1]$.  Let
  $x= \carr( \alpha( 0))$ and $y= \carr( \alpha( 1))$, then
  $\ell(\alpha)=\subid{\ell(x)}{S}{\ell(y)}$. \qed
\end{lemma}

\begin{proposition}
  \label{prp:LabelCoherence}
  For every d-path $\alpha\in\po P\georel{X}$ there exists a track
  $\rho:\carr(\alpha(0))\leadsto \carr(\alpha(1))$ in $X$ such that
  $\ell(\alpha)\cong \ell(\rho)$.
\end{proposition}

\begin{proof}
  By Lemma \ref{le:DipathPresentation} there exists a presentation
  \begin{equation*}
    \alpha= \big( ( \georel{\ineda_{x_1}}\circ \beta_1)*(
    \georel{\ineda_{x_2}}\circ \beta_2)*\dots*( \georel{\ineda_{x_m}}\circ
    \beta_m) \big) \circ h
  \end{equation*}
  such that $x_i\in X_{n_i}$,
  $\beta_i\in \po P\georel{\square^{\ell(x_i)}}$ and
  $\ineda_{x_i}:\square^{\ell(x_i)}\to X$ is the unique HDA map sending
  $\yoneda_{\ell(x_i)}$ into $x_i$ (we replace $\square^{n_i}$ with
  $\square^{\ell(x_i)}$ to obtain compatible labelings and make
  $\ineda_{x_i}$ HDA maps). We have
  \begin{alignat*}{2}
    \ell(\alpha) &= \ell((( \georel{\ineda_{x_1}}\circ
    \beta_1)*\dots*(
    \georel{\ineda_{x_m}}\circ \beta_m))\circ h) \\
    &= \ell(( \georel{\ineda_{x_1}}\circ \beta_1)*\dots*(
    \georel{\ineda_{x_m}}\circ \beta_m)) &\qquad
    &\eqref{le:PathLabelRepar} \\
    &= \ell( \georel{\ineda_{x_1}}\circ \beta_1)*\dots* \ell(
    \georel{\ineda_{x_m}}\circ \beta_m) &&\eqref{le:PathLabelComp} \\
    &= \ell(\beta_1)*\dots * \ell(\beta_m)) &&\eqref{le:PathLabelFun}
    \\
    &=
    \subid{\ell(\carr(\beta_1(0)))}{\ell(x_1)}{\ell(\carr(\beta_1(1)))}*\dotsm*
    \subid{\ell(\carr(\beta_m(0)))}{\ell(x_m)}{\ell(\carr(\beta_m(1)))}
    &&\eqref{le:PathLabelCube} \\
    &= \ell(\carr(\beta_1(0)), x_1,
    \carr(\beta_1(1)))*\dotsm*\ell(\carr(\beta_m(0)), x_m,
    \carr(\beta_m(1))) &&\eqref{le:UpDownTrackLabel} \\
    &= \ell(\carr(\alpha(0)), x_1, \carr(\beta_1(1)),\dotsc,
    \carr(\beta_m( 0)), x_m, \carr(\alpha(1))),
  \end{alignat*}
  hence we can set
  $\rho=( \carr(\alpha(0)), x_1, \carr(\beta_1(1)),\dotsc,
  \carr(\beta_m( 0)), x_m, \carr(\alpha(1)))$.
\end{proof}

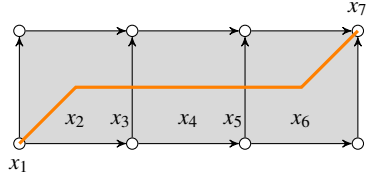
\begin{figure}
  \centering
  \begin{tikzpicture}[x=1.5cm, y=1.5cm]
    \path[fill=black!15] (0,0) to (3,0) to (3,1) to (0,1);
    \foreach \x in {0, 1, 2, 3} \foreach \y in {0, 1} \node[state] (\x\y)
    at (\x,\y) {};
    \foreach \x in {0, 1, 2, 3} \path (\x0) edge (\x1);
    \foreach \y in {0, 1} \path (0\y) edge (1\y);
    \foreach \y in {0, 1} \path (1\y) edge (2\y);
    \foreach \y in {0, 1} \path (2\y) edge (3\y);
    \draw[-, very thick, orange] (0,0) -- (.5,.5) -- (2.5,.5) --
    (3,1);
    \node[below] at (00.south) {$x_1$};
    \foreach \i/\x in {2/.5, 3/.9, 4/1.5, 5/1.9, 6/2.5} \node at
    (\x,.2) {$\vphantom{p}x_\i$};
    \node[above] at (31.north) {$x_7$};
  \end{tikzpicture}
  \medskip
  \caption{Track $\rho=( x_1,\dotsc, x_7)$ together with d-path
    $\alpha$ through center points of $\rho$.}
  \label{fi:cptrack}
\end{figure}

For the converse result, we construct a d-path through the center
points of a given track, see also \cite{Fajstrup05-cubcomp} and Figure
\ref{fi:cptrack} for an example.

\begin{proposition}
  \label{prp:TracksHavePaths}
  For every track $\rho: x\leadsto y$ in $X$ there is a d-path
  $\alpha$ with $\carr( \alpha( 0))= x$, $\carr( \alpha( 1))= y$, and
  $\ell(\alpha)=\ell(\rho)$.
\end{proposition}

\begin{proof}
  If $\rho=( x)$ is a unit track, we can let
  $\beta\in \po P( \po I^{ \dim x})$ be the constant d-path
  $\beta( t)=( \frac{1}{2},\dotsc, \frac{1}{2})$ and
  $\alpha= \georel{ \ineda_x}\circ \beta$.  Otherwise, write
  $\rho=( x_1,\dotsc, x_m)$ with $m\ge 2$ and let $n_i= \dim x_i$ for
  $i\in[ m]$.  We construct $\alpha$ as a concatenation of d-paths
  $\alpha_1*\dotsm* \alpha_{ m- 1}$.  Let $i\in[ m- 1]$.
  \begin{itemize}
  \item If $x_i\face^* x_{ i+ 1}$, then
    $x_i= \delta_A^{ 0,\dotsc, 0} x_{ i+ 1}$ for a unique set
    $A\subseteq[ n_{ i+ 1}]$.  Let
    $\beta_i\in \po P( \po I^{ n_{ i+ 1}})$ be the d-path
    \begin{equation*}
      \beta_i( t)=( t_1,\dotsc, t_{ n_{ i+ 1}}), \qquad t_j=
      \begin{cases}
        \frac{1}{2} t &\text{if } j\in A,\\
        \frac{1}{2} &\text{if } j\notin A
      \end{cases}
    \end{equation*}
    and $\alpha_i= \georel{ \ineda_{ x_{ i+ 1}}}\circ \beta_i$.  Then
    $\carr( \alpha_i(0))= x_i$ and $\carr( \alpha_i(t))= x_{ i+ 1}$
    for $0< t\le 1$.
  \item If $x_i\ecaf^* x_{ i+ 1}$, then
    $x_{ i+ 1}= \delta_A^{ 1,\dotsc, 1} x_i$ for a unique set
    $A\subseteq[ n_i]$.  Let $\beta_i\in \po P( \po I^{ n_i})$ be the
    d-path
    \begin{equation*}
      \beta_i( t)=( t_1,\dotsc, t_{ n_i}), \qquad t_j=
      \begin{cases}
        \frac{1}{2}+ \frac{1}{2} t &\text{if } j\in A,\\
        \frac{1}{2} &\text{if } j\notin A
      \end{cases}
    \end{equation*}
    and $\alpha_i= \georel{ \ineda_{ x_i}}\circ \beta_i$.  Then
    $\carr( \alpha_i(t))= x_i$ for $0\le t< 1$ and
    $\carr( \alpha_i( 1))= x_{ i+ 1}$.
  \end{itemize}
  By construction, $\alpha_i(1)= \alpha_{ i+ 1}(0)$ for all
  $i\in[ m- 1]$, so the concatenation
  $\alpha= \alpha_1*\dotsm* \alpha_{ m- 1}$ exists.  Further, this is a
  representation as in Lemma \ref{le:DipathPresentation}, hence
  $\ell(\alpha)=\ell(\rho)$ by Proposition \ref{prp:LabelCoherence}.
\end{proof}

\section{Languages of Higher-Dimensional Automata}
\label{se:lang}

We define languages of HDAs and discuss some of their properties.

\subsection{Languages}
\label{se:langlang}

Using the work in Sections \ref{se:tracks} and \ref{se:geo}, we can
define languages of HDAs in two different ways.  The first one is a
straight application of van Glabbeek's track-based approach from
\cite{DBLP:journals/tcs/Glabbeek06}, and the second one uses d-paths
through geometric realizations in the spirit of
\cite{DBLP:journals/tcs/FajstrupRG06}.

\begin{definition}
  \label{de:lang}
  A track $\rho: x\leadsto y$ in an HDA $( X, I, F, \lambda)$ is
  \emph{accepting} if $x\in I$ and $y\in F$.  The
  \emph{track language} of $X$ is
  $L_t( X)=\{ \lambda( \rho)\in \iiPoms\mid \rho \text{ accepting
    track in } X\}$.

  A d-path $\alpha\in\po P\georel X$ is \emph{accepting} if
  $\carr(\alpha(0))\in I$ and $\carr(\alpha(1))\in F$.  The \emph{path
    language} of $X$ is
  $L_p( X)=\{ \lambda( \alpha)\in \iiPoms\mid \alpha \text{ accepting
    d-path in } \georel X\}$.
\end{definition}

\begin{therm}
  \label{th:track=path}
  For every HDA $X$, $L_t(X)=L_p(X)$.
\end{therm}

\begin{proof}
  Immediate from Propositions \ref{prp:LabelCoherence} and
  \ref{prp:TracksHavePaths}.
\end{proof}
	
From now on we write $L(X)= L_t(X) =L_p(X)$ and call this set simply
the \emph{language} of $X$.  It follows immediately from Proposition
\ref{pr:lambda-track} that languages of HDAs are sets of interval
ipomsets:

\begin{proposition}
  \label{pr:LXinterval}
  For any HDA $X$, $L( X)\subseteq \iiPoms$. \qed
\end{proposition}

The following property allows us to reason about languages using maps
from objects $\pobj{P}$.

\begin{proposition}
  \label{pr:Representability}
  For any HDA $X$ and any interval ipomset $P$, $P\in L( X)$ iff there
  is an HDA map $\pobj{ P}\to X$.
\end{proposition}

\begin{proof}
  For the forward direction, assume $P\in L(X)$, then there exists a
  track $\rho:x\leadsto y$ with $x\in I_X$, $y\in F_X$, and
  $\lambda(\rho)=P$.  The conclusion follows from Lemma
  \ref{le:track-to-map}.

  For the reverse direction, let $g: \pobj{ P}\to X$. Then, by
  Proposition \ref{pr:internal-track}, there exists a track
  $\rho: i_{\pobj{P}}\leadsto f_{\pobj{P}}$ such that
  $\lambda(\rho)=P$, $g(i_{\pobj{P}})\in I_X$, and
  $g(f_{\pobj{P}})\in F_X$.  Now Proposition \ref{pr:lambda-track}
  implies that $\lambda(f(\rho))=P$.
\end{proof}

\begin{remark}
  Thanks to Proposition \ref{pr:Representability}, the language of an
  HDA $X$ may alternatively be defined as the set of interval ipomsets
  $P$ that admit an HDA map $\pobj{P}\to X$.  This definition remains
  valid even if we do \emph{not} assume event consistency, hence it
  may be used to introduce languages also of HDA which are not event
  consistent.  We will expand on this in future work.
\end{remark}

We finish this section with some properties of languages of HDAs
generated by interval ipomsets.  The following is immediate from
Proposition \ref{pr:internal-track}.

\begin{lemma}
  $P\in L(\pobj{P})$ for every interval ipomset $P$. \qed
\end{lemma}

\begin{proposition}
  \label{pr:XPSubsumption}
  For all interval ipomsets $P$ and $Q$, $Q\in L(\pobj{P})$ iff
  $Q\subsu P$.
\end{proposition}

\begin{proof}
  The backwards direction is immediate from Lemma
  \ref{le:SubsumptionToPObj} and Proposition
  \ref{pr:Representability}: a subsumption map $f: Q\to P$ gives rise
  to $\pobj{f}: \pobj{Q}\to \pobj{P}$, thus $Q\in L( \pobj{P})$.  For
  the forward direction, let $\rho: i_{\pobj{P}}\leadsto f_{\pobj{P}}$
  be an accepting track in $\pobj{P}$.  We show that
  $\lambda(\rho)\subsu P$ by induction on the length of $\rho$.

  If $\rho=(x)$, then $i_{\pobj{P}}=x=f_{\pobj{P}}$, which implies
  that $P=\subid{P}{P}{P}=\ell(\rho)$.  Otherwise, there is a
  presentation $\rho=(x,y)*\tau$.  Note that $\ell(x)\cong S_P$.
  There are two cases to consider:
  \begin{itemize}
  \item $x\face^* y$.  Then $y(p)=\exec$ for $p\in \ev(y)$ and
    $y(p)=0$ otherwise.  Let $Q$ be an interval ipomset with the same
    elements as $P$, $\mathord{<_Q}= \mathord{<_P}$,
    $\mathord{\intord_Q}= \mathord{\intord_P}$, $\lambda_Q=\lambda_P$,
    $T_Q=T_P$; the only difference is that $S_Q=\ev(y)$.  Then
    $\pobj{P}$ and $\pobj{Q}$ are naturally isomorphic as labeled
    precubical sets and $\tau$ can be regarded as an accepting track
    in $\pobj{Q}$.  Moreover, $P=\subid{\ell(x)}{\ell(y)}{\ell(y)}*Q$.
    By induction, $\ell(\tau)\subsu Q$; using Lemma
    \ref{le:CompSubsu}, $\ell(\rho)=\ell(x, y)*\ell(\tau)\subsu P$.
  \item $x\ecaf^* y$.  Then
    \[
      y(p)=
      \begin{cases}
        0 & \text{for $p\in P\setminus \ev(x)$,}\\
        \exec & \text{for $p\in \ev(y)$,}\\
        1 & \text{for $p\in \ev(x)\setminus \ev(y)$.}
      \end{cases}
    \]
    Let $Q$ be the restriction of $P$ to
    $P\setminus(\ev(x)\setminus \ev(y))$, then the precubical map
    $\jneda^1_{Q\subseteq P}:\pobj{Q}\to \pobj{P}$ is an injection
    onto $\{x\mid x(p)=1 \text{ for } p\in \ev(x)\setminus\ev(y)\}$.
    Furthermore, $\tau$ is a track from
    $\jneda^1_{Q\subseteq P}( i_{\pobj{Q}})$ to
    $\jneda^1_{Q\subseteq P}( f_{\pobj{Q}})$ lying in
    $\jneda^1_{Q\subseteq P}(\pobj{Q})$.  Thus, $\tau$ lifts uniquely
    to an accepting track $\tau'$ on $\pobj{Q}$.  By induction
    hypothesis, $\ell(\tau)\subsu Q$, and then with Lemma
    \ref{le:CompSubsu},
    $\ell(\rho)=\ell(x, y) * \ell(\tau)\subsu
    \subid{\ell(x)}{\ell(x)}{\ell(y)}*Q= P$. \qedhere
  \end{itemize}
\end{proof}

\begin{remark}
  Example \ref{ex:2+2-to-HDA} shows that the above proposition fails
  if $P$ is not an interval ipomset: for $P= \twotwo$,
  $P\notin L( \pobj{P})$.  In general, Proposition \ref{pr:LXinterval}
  implies that if $P\notin \iiPoms$, then $P\notin L( \pobj{P})$.  We
  will get back to this issue in Example \ref{ex:para} below.
\end{remark}

\subsection{Languages are Subsumption-Closed}

Because of Proposition \ref{pr:LXinterval} we henceforth restrict
ourselves to \emph{interval} ipomsets.

\begin{definition}
  The \emph{weak closure} of a set $\mcal S\subseteq \iiPoms$ is
  $\mcal S\down=\{ Q\in \iiPoms\mid \exists P\in \mcal S: Q\subsu
  P\}$.
\end{definition}

That is, $\mcal S\down$ is the smallest subsumption-closed superset of
$\mcal S$.  The set $\mcal S$ is called \emph{weak} if $\mcal S\down=
\mcal S$.

\pagebreak 

\begin{therm}
  \label{th:LXsubsucl}
  For every HDA $X$, $L( X)\subseteq \iiPoms$ is weak.
\end{therm}

\begin{proof}
  This follows from subsumption closedness of $L( \pobj{P})$,
  Proposition \ref{pr:XPSubsumption}: Choose interval ipomsets
  $Q\subsu P$ with $P\in L(X)$. Proposition \ref{pr:Representability}
  gives a map $f: \pobj{P}\to X$ and Proposition
  \ref{pr:XPSubsumption} gives a map $g: \pobj{Q}\to \pobj{P}$. The
  composition $f\circ g$ with Proposition \ref{pr:Representability}
  again gives the conclusion.
\end{proof}

As a partial converse, we will see in Theorem
\ref{th:hda-finitesetiipoms} below that any \emph{finite}
subsumption-closed set of interval ipomsets can be generated by an
HDA.

\begin{lemma}
  \label{le:HDAmap-L}
  If $f: X\to Y$ is an HDA map, then $L(X)\subseteq L(Y)$.
\end{lemma}

\begin{proof}
  Let $P\in L(X)$, then Proposition \ref{pr:Representability} gives a
  map $\pobj{P}\to X$.  Composition with $f$ yields a map $\pobj{P}\to
  Y$, hence $P\in L(Y)$.
\end{proof}

For HDAs generated by interval pomsets, Proposition
\ref{pr:XPSubsumption} implies the following.

\begin{lemma}
  \label{le:langXP}
  $L(\pobj{P})=\{P\}\down$. \qed
\end{lemma}

\subsection{Languages are Closed under Union}

We now show that languages of HDAs are closed under union (that is,
they form filters).  To this end, we introduce \emph{coproducts} of
HDAs.  First, the coproduct of precubical sets $X$ and $Y$ is
$Z= X\sqcup Y$ given by
\begin{equation*}
  Z_n= X_n\sqcup Y_n, \qquad
  \delta_i^\nu( z)=
  \begin{cases}
    ( \delta_X)_i^\nu( z) &\text{if } z\in X, \\
    ( \delta_Y)_i^\nu( z) &\text{if } z\in Y.
  \end{cases}
\end{equation*}

\begin{definition}
  The \emph{coproduct} of HDAs $( X, I_X, F_X, \lambda_X)$ and
  $( Y, I_Y, F_Y, \lambda_Y)$ is the HDA
  $X\sqcup Y=( X\sqcup Y, I_X\cup I_Y, F_X\cup F_Y, \lambda)$ with
  $\lambda( z)= \lambda_X( z)$ if $z\in X$ and
  $\lambda( z)= \lambda_Y( z)$ if $z\in Y$.
\end{definition}

It can easily be shown that these are in fact the categorical
coproducts in the categories of precubical sets and HDAs,
respectively.  Next we note that subsumption closure of sets of
interval ipomsets distributes over union
\cite{DBLP:journals/fuin/Grabowski81}:

\begin{lemma}
  \label{le:subsu-union}
  For any subsets $\mcal S_1, \mcal S_2\subseteq \iiPoms$,
  $( \mcal S_1\cup \mcal S_2)\down= \mcal S_1\down\cup \mcal
  S_2\down$. \qed
\end{lemma}

\begin{therm}
  \label{th:hdlang-union}
  For HDAs $X$ and $Y$, $L( X\sqcup Y)= L( X)\cup L(Y)$.
\end{therm}

\begin{proof}
  By construction of $X\sqcup Y$, any accepting track in $X\sqcup Y$
  is an accepting track in $X$ or in $Y$, and vice versa.  The result
  follows with Lemma \ref{le:subsu-union}.
\end{proof}

\begin{therm}
  \label{th:hda-finitesetiipoms}
  Let $\mcal S\subseteq \iiPoms$ be weak and finite.  There is an HDA
  $X$ with $L( X)= \mcal S$.
\end{therm}

\begin{proof}
  Write $\mcal S=\{ P_1,\dotsc, P_n\}\down$ and let
  $X= \pobj{ P_1}\sqcup\dotsm\sqcup \pobj{ P_n}$.  By Lemma
  \ref{le:langXP}, $L( \pobj{P_i})=\{ P_i\}\down$ for all
  $i= 1,\dotsc, n$, so using Theorem \ref{th:hdlang-union},
  $L( X)= \{ P_1\}\down\cup\dotsm\cup\{ P_n\}\down=( P_1\cup\dotsm
  P_n)\down$ by Lemma~\ref{le:subsu-union}.
\end{proof}

\subsection{Languages are Closed under Parallel Composition}

We show below that parallel compositions of HDA languages are
languages of \emph{tensor products} of HDAs.  First, the tensor
product of precubical sets $X$ and $Y$ is $Z= X\otimes Y$ given
by
\begin{equation*}
  Z_n= \bigsqcup_{ k+ l= n} X_k\times Y_l, \qquad
  \delta_i^\nu(( x, y))=
  \begin{cases}
    ( \delta_X)_i^\nu( x) &\text{if } i\le \dim x, \\
    ( \delta_Y)_{ i- \dim x}^\nu( y) &\text{if } i> \dim x.
  \end{cases}
\end{equation*}

We will below use the important fact that geometric realizations of
tensor products are products of geometric realizations
\cite{book/Grandis09}:

\begin{lemma}
  For precubical sets $X$ and $Y$,
  $\georel{ X\otimes Y}= \georel{X}\times \georel{Y}$. \qed
\end{lemma}

\begin{definition}
  The \emph{tensor product} of HDAs $( X, I_X, F_X, \lambda_X)$ and
  $( Y, I_Y, F_Y, \lambda_Y)$ is
  $X\otimes Y=( X\otimes Y, I, F, \lambda)$ with
  $I=\{( x, y)\mid x\in I_X, y\in I_Y\}$,
  $F=\{( x, y)\mid x\in F_X, y\in F_Y\}$, and $\lambda(( x, y))=
  \lambda_X( x)* \lambda_Y( y)$.
\end{definition}

Above, $\lambda(( x, y))= \lambda_X( x)* \lambda_Y( y)$ denotes the
concatenation of $\lambda_X( x)$ and $\lambda_Y( y)$ as sequences in
$\bang \Sigma$.  More formally, one can easily show that
$\bang \Sigma\otimes \bang \Sigma= \bang \Sigma$, so that $\lambda$ is
the tensor product of the maps $\lambda_X$ and $\lambda_Y$.

\begin{remark}
  If $X$ and $Y$ are one-dimensional HDAs, \ie $X_2= Y_2= \emptyset$,
  then $Z= X\otimes Y$ is two-dimensional, with $Z_2= X_1\times Y_1$
  and $Z_1= X_1\times Y_0\sqcup X_0\times Y_1$.  The labels of
  $2$-cells $( x, y)\in Z_2$ are
  $\lambda(( x, y))=( \lambda_X( x), \lambda_Y( y))$, and the labels
  of $1$-cells $( x, y)\in Z_1$ are $\lambda( x, y)= \lambda_X( x)$
  for $x\in X_1$ and $\lambda( x, y)= \lambda_Y( y)$ for $y\in Y_1$.
  Hence $X\otimes Y$ can be seen as the \emph{synchronized product}
  \cite[Sec.~2.2.3]{WinskelN95-Models} of the finite automata $X$ and
  $Y$.
\end{remark}

\begin{lemma}
  \label{le:pobj-para}
  For ipomsets $P$ and $Q$,
  $\pobj{ P\para Q}\cong \pobj{P}\otimes \pobj{Q}$.
\end{lemma}

\begin{proof}
  Let $X= \pobj{ P\para Q}$.  As the underlying set of $P\para Q$ is
  the disjoint union $P\sqcup Q$ and
  $\mathord{<_{ P\para Q}}= \mathord{<_P}\sqcup \mathord{<_Q}$, any
  poset map
  $x:( P\para Q, \mathord{<_{ P\para Q}})\to\{ 0, \exec, 1\}$ has a
  unique decomposition $x= x_P\sqcup x_Q$ into poset maps
  $x_P:( P, \mathord{<_P})\to\{ 0, \exec, 1\}$ and
  $x_Q:( Q, \mathord{<_Q})\to\{ 0, \exec, 1\}$; and any two such maps
  give rise to a poset map $x$.  Hence
  $X_n\cong \bigsqcup_{ k+ l= n} \pobj{P}_k\times \pobj{Q}_l$ as sets.
  It is easy to see that the face maps agree on both sides, and the
  same holds for the labeling.  For the initial cell we have
  $i_{ \pobj{ P\para Q}}( p)= \exec$ iff $p\in S_{ P\para Q}$ (and $0$
  otherwise), iff $p\in S_P$ or $p\in S_Q$, hence
  $i_{ \pobj{ P\para Q}}$ maps to $i_{ \pobj{P}}\sqcup i_{ \pobj{Q}}$
  under the isomorphism; similarly for the accepting cell.
\end{proof}

\begin{definition}
  The \emph{parallel composition} of subsumption-closed subsets
  $\mcal S_1, \mcal S_2\subseteq \iiPoms$ is \linebreak 
  $\mcal S_1\para \mcal S_2=\{ P\| Q\mid P\in \mcal S_1, Q\in \mcal
  S_2\}\down \cap \iiPoms$.
\end{definition}

We need to take the intersection with $\iiPoms$ above because parallel
compositions of interval ipomsets may not be interval.

\begin{figure}
  \centering
  \begin{tikzpicture}
    \begin{scope}
      \node[state, initial left] (00) at (0,0) {};
      \node[state] (10) at (1,0) {};
      \node[state, accepting] (20) at (2,0) {};
      \path (00) edge node[below] {$\vphantom{b}a$} (10);
      \path (10) edge node[below] {$b$} (20);
    \end{scope}
    \begin{scope}[shift={(0,-1.5)}]
      \node[state, initial left] (00) at (0,0) {};
      \node[state] (10) at (1,0) {};
      \node[state, accepting] (20) at (2,0) {};
      \path (00) edge node[below] {$\vphantom{d}c$} (10);
      \path (10) edge node[below] {$d$} (20);
    \end{scope}
    \begin{scope}[shift={(5,-1.75)}]
      \path[fill=black!15] (0,0) to (2,0) to (2,2) to (0,2);
      \node[state, initial left] (00) at (0,0) {};
      \node[state] (10) at (1,0) {};
      \node[state] (20) at (2,0) {};
      \node[state] (01) at (0,1) {};
      \node[state] (11) at (1,1) {};
      \node[state] (21) at (2,1) {};
      \node[state] (02) at (0,2) {};
      \node[state] (12) at (1,2) {};
      \node[state, accepting] (22) at (2,2) {};
      \path (00) edge node[below] {$\vphantom{d}c$} (10);
      \path (10) edge node[below] {$d$} (20);
      \path (00) edge node[left] {$a$} (01);
      \path (01) edge node[left] {$b$} (02);
      \foreach \s/\t in {01/11, 11/21, 02/12, 12/22, 10/11, 11/12,
        20/21, 21/22} \path (\s) edge (\t);
    \end{scope}
  \end{tikzpicture}
  \medskip
  \caption{HDAs $X= \pobj{(a< b)}$ and $Y= \pobj{(c< d)}$ (left) and
    their tensor product $X\otimes Y$.}
  \label{fig:tensor-hda}
\end{figure}
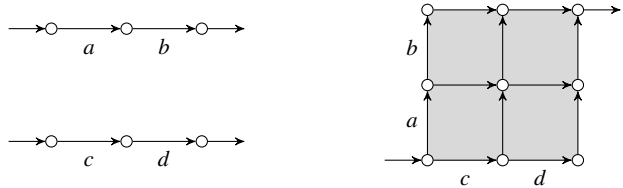

\begin{example}
  \label{ex:para}
  Let $P$ and $Q$ be the ipomsets $P=(a\longrightarrow b)$,
  $Q=(c\longrightarrow d)$.
  Figure \ref{fig:tensor-hda} shows the
  one-dimensional HDAs $X= \pobj{P}$ and $Y= \pobj{Q}$ as well as
  their tensor product $X\otimes Y= \pobj{ P\para Q}$ (\cf Example
  \ref{ex:2+2-to-HDA} and Figure \ref{fi:pobj-ex}).  Now $L( X)=\{
  P\}\down$ and $L( Y)=\{ Q\}\down$, but as $P\| Q$ is not an interval
  ipomset, $L( X\otimes Y)\ne\{ P\para Q\}\down$.  Instead,
  \begin{equation*}
    L( X\otimes Y)=\{ P\para Q\}\down\cap \iiPoms = \bigg\{ \pomset{ a
      \ar[r] & b \\ c \ar[r] \ar[ur] & d}, \pomset{ a \ar[r] \ar[dr] &
      b \\ c \ar[r] & d} \bigg\} \biggdown.
  \end{equation*}
\end{example}


\begin{therm}
  \label{th:hdapara}
  For HDAs $X$ and $Y$, $L( X\otimes Y)= L( X)\para L(Y)$.
\end{therm}

\begin{proof}
  To show $L( X)\para L(Y)\subseteq L( X\otimes Y)$, let
  $R\in L( X)\para L( Y)$, then there are $P\in L(X)$ and $Q\in L(Y)$
  such that $R\subsu P\para Q$.  Let $f:\pobj{P}\to X$ and
  $g:\pobj{Q}\to Y$ be the maps given by Proposition
  \ref{pr:Representability}.  There is a composition
  \begin{equation*}
    \pobj{R} \overset{\eqref{le:SubsumptionToPObj}}{\longrightarrow}
    \pobj{P\para Q} \overset{\eqref{le:pobj-para}}{\longrightarrow}
    \pobj{P}\otimes \pobj{Q} \xrightarrow{f\otimes g} X\otimes Y,
  \end{equation*}
  thus $R\in L(X\otimes Y)$.

  For showing $L( X\otimes Y)\subseteq L( X)\para L(Y)$ we have to do
  more work.  Let $R\in L( X\otimes Y)$, then there is a d-path
  $\gamma\in \po P\georel{ X\otimes Y}$ with $\lambda( \gamma)= R$,
  $\carr( \gamma( 0))\in I_{ X\otimes Y}$, and
  $\carr( \gamma( 1))\in F_{ X\otimes Y}$.  Now
  $\georel{ X\otimes Y}= \georel{X}\times \georel{Y}$, so let $\alpha$
  and $\beta$ be the projections of $\gamma$ to $\georel{X}$ and
  $\georel{Y}$, respectively.  Let $P= \lambda( \alpha)$ and
  $Q= \lambda( \beta)$.  We have $\carr( \alpha( 0))\in I_X$,
  $\carr( \alpha( 1))\in F_X$, $\carr( \beta( 0))\in I_Y$, and
  $\carr( \beta( 1))\in F_Y$, so that $P\in L( X)$ and $Q\in L( Y)$.

  We show that $R\subsu P\para Q$.  We have
  $E_{ X\otimes Y}= E_X\sqcup E_Y$, so for every
  $e\in E_{ X\otimes Y}$, $J_e^\alpha= \emptyset$ or
  $J_e^\beta= \emptyset$.  Further,
  $J_e^\gamma= J_e^\alpha\cup J_e^\beta$ for every
  $e\in E_{ X\otimes Y}$, so that the presentation
  $J_{\smash{e}}^\gamma=I_{\smash{e,1}}^\gamma\cup \dots\cup
  I^\gamma_{\smash{e,n_e^\gamma}}$ is the same as the one for
  $J_e^\alpha$ or $J_e^\beta$.  Hence the underlying sets
  $R= P\sqcup Q$, and $x<_P y$ or $x<_Q y$ imply $x<_R y$.

  Regarding the event orders, we work directly with the elementary
  relations $\intord_{ \gamma( t)}$.  Assume
  $x\intord_{ \gamma(t)} y$, then $t\in I^\gamma_x\cap I^\gamma_y$.
  Now, writing $\mcal J^\alpha= \{ J^\alpha_e\mid e\in E_X\}$ and
  $\mcal J^\beta= \{ J^\beta_e\mid e\in E_Y\}$,
  \begin{itemize}
  \item if $I^\gamma_x, I^\gamma_y\in \mcal J^\alpha$, then
    $x\intord_{ \alpha(t)} y$;
  \item if $I^\gamma_x, I^\gamma_y\in \mcal J^\beta$, then
    $x\intord_{ \beta(t)} y$;
  \item if $I^\gamma_x\in \mcal J^\alpha$ and
    $I^\gamma_y\in \mcal J^\beta$, then $x\in P$ and $y\in Q$, hence
    $x\intord_{ P\para Q} y$; and
  \item the case $I^\gamma_x\in \mcal J^\beta$ and
    $I^\gamma_y\in \mcal J^\alpha$ cannot occur: this would imply
    $x\in Q$ and $y\in P$ and hence $y\intord_{ \gamma(t)} x$ instead
    of $x\intord_{ \gamma(t)} y$.
  \end{itemize}
  We have shown that $x\intord_{ \gamma(t)} y$ implies
  $x\intord_{ P\para Q} y$, so this also holds for the transitive
  closure~$\intord_R$.
\end{proof}



\subsection{Language Equivalence is Implied by Bisimulation}

As a final sanity check of our notion of language, we show that
bisimilarity of HDAs implies their language equivalence.  Fahrenberg
\cite{DBLP:conf/fossacs/Fahrenberg05} has introduced a notion of
\emph{hd-bisimilarity} for HDAs which in our setting can be stated as
follows.  An \emph{hd-bisimulation} between HDAs $X$ and $Y$ is a
graded set $R= \bigcup R_n$ with $R_n\subseteq X_n\times Y_n$ such
that
\begin{enumerate}
\item \label{en:bisim.face} $R$ is closed under face maps: for all $(
  x, y)\in R_n$, $i\in[ n]$ and $\nu\in\{ 0, 1\}$, $( \delta_i^\nu x,
  \delta_i^\nu y)\in R_{ n- 1}$;
\item \label{en:bisim.l} $R$ respects labels: for all $( x, y)\in R$,
  $\lambda_X( x)= \lambda_Y( y)$;
\item \label{en:bisim.if} the restrictions $R\cap I_X\times I_Y$ and
  $R\cap F_X\times F_Y$ are bijections;
\item \label{en:bisim.e1} for all $( x, y)\in R$ and any $x'\in X$ and
  $k\in[ \dim x']$ such that $x= \delta_k^0 x'$, there exists
  $y'\in Y$ such that $y= \delta_k^0 y'$ and $( x', y')\in R$;
\item \label{en:bisim.e2} for all $( x, y)\in R$ and any $y'\in Y$ and
  $k\in[ \dim y']$ such that $y= \delta_k^0 y'$, there exists
  $x'\in X$ such that $x= \delta_k^0 x'$ and $( x', y')\in R$.
\end{enumerate}
Hence initial and accepting cells are related bijectively
(\ref{en:bisim.if}), and (\ref{en:bisim.e1}) whenever a cell in $X$
can be extended, then a related extension is available in $Y$, and
vice versa (\ref{en:bisim.e2}).  Finally, $X$ and $Y$ are
\emph{hd-bisimilar} if there exists an hd-bisimulation
$R\subseteq X\times Y$: this is an equivalence relation.

As in \cite{DBLP:conf/fossacs/Fahrenberg05}, we can express
hd-bisimilarity using \emph{open
  maps}~\cite{DBLP:journals/iandc/JoyalNW96}.  We say that an HDA map
$f: X\to Y$ is \emph{open} if $f$ is bijective
on initial and accepting cells and the following \emph{zig-zag
  property} holds for every $x\in X$: if $y'\in Y$ and
$k\in[ \dim y']$ are such that $f( x)= \delta_k^0 y'$, then there
exists $x'\in X$ with $x= \delta_k^0 x'$ and $y'= f( x')$.  The
following is shown in \cite{DBLP:conf/fossacs/Fahrenberg05}.

\begin{lemma}
  HDAs $X$ and $Y$ are hd-bisimilar iff there exists an HDA $Z$ and a
  span of open maps $X\from Z\to Y$. \qed
\end{lemma}



\begin{therm}
  \label{th:bisim-lang}
  If HDAs $X$ and $Y$ are hd-bisimilar, then $L(X)= L(Y)$.
\end{therm}

\begin{proof}
  It suffices to assume an open HDA map $f: X\to Y$; the inclusion
  $L(X)\subseteq L(Y)$ is then clear by Lemma \ref{le:HDAmap-L}.  For
  the reverse inclusion, let $\sigma=( y_1,\dotsc, y_m)$ be an
  accepting track in $Y$.  By bijectivity
  of $f$ on initial cells there is $x_1\in I_X$ such that
  $f( x_1)= y_1$, and then inductive application of the zig-zag
  property yields a track $\rho=( x_1,\dotsc, x_m)$ in $X$ with
  $f( x_i)= y_i$ for all $i$ and
  $\lambda_X( \rho)= \lambda_Y( \sigma)$, with $x_m\in F_X$ because
  $f$ is bijective on accepting cells.
\end{proof}

In \cite{DBLP:journals/tcs/Glabbeek06}, van Glabbeek introduces a
notion of \emph{ST-bisimilarity} for HDAs which in our notation is
given as follows.  An \emph{ST-bisimulation} between HDAs $X$ and $Y$
is a relation $R$ between tracks in $X$ and $Y$ such that
\begin{enumerate}
\item \label{en:stbis.i} $R$ is a bijection between initial unit
  tracks $\{( x)\mid x\in I_X\}$ and $\{( y)\mid y\in I_Y\}$;
\item \label{en:stbis.f} $R$ respects accepting cells: for all
  $( \rho, \sigma)\in R$ such that $\rho: x\leadsto x'$ and
  $\sigma: y\leadsto y'$, $x'\in F_X$ iff $y'\in F_Y$;
\item \label{en:stbis.l} $R$ respects labels: for all
  $( \rho, \sigma)\in R$, $\ell_X( \rho)= \ell_Y( \sigma)$;
\item \label{en:stbis.e1} for all $( \rho, \sigma)\in R$ and track
  $\rho'$ in $X$ such that $\rho$ and $\rho'$ may be concatenated,
  there exists a track $\sigma'$ in $Y$ such that $( \rho* \rho',
  \sigma* \sigma')\in R$;
\item \label{en:stbis.e2} for all $( \rho, \sigma)\in R$ and track
  $\sigma'$ in $Y$ such that $\sigma$ and $\sigma'$ may be
  concatenated, there exists a track $\rho'$ in $X$ such that
  $( \rho* \rho', \sigma* \sigma')\in R$.
\end{enumerate}
That is, whenever a track in $X$ can be extended, then a related
extension is available in $Y$ and vice versa.  Finally, $X$ and $Y$
are \emph{ST-bisimilar} if there exists an ST-bisimulation $R$ between
them; this is an equivalence relation.

\begin{therm}
  \label{th:stbis-lang}
  If HDAs $X$ and $Y$ are ST-bisimilar, then $L(X)= L(Y)$.
\end{therm}

\begin{proof}
  By symmetry it suffices to show the inclusion $L(X)\subseteq L(Y)$.
  Let $P\in L(X)$, then there is a track $\rho: x\leadsto x'$ in $X$
  with $x\in I_X$, $x'\in F_X$ and $\ell(\rho)= P$.  By
  \eqref{en:stbis.i}, there is $y\in I_Y$ such that the unit tracks
  $(( x),( y))\in R$.  Now $\rho=( x)* \rho$, so using
  \eqref{en:stbis.e1} there exists a track $\sigma: y\leadsto y'$ in
  $Y$ such that $( \rho, \sigma)\in R$, but then by
  \eqref{en:stbis.f}, $y'\in F_Y$.  Hence $\sigma$ is an accepting
  track in $Y$, and by \eqref{en:stbis.l},
  $\ell( \sigma)= \ell( \rho)= P$, so that $P\in L(Y)$.
\end{proof}

In \cite{DBLP:journals/tcs/Glabbeek06}, other notions of
\emph{history-preserving} and \emph{hereditary} history-preserving
bisimilarity for HDAs are introduced; both imply ST-bisimilarity and,
thus, language equivalence.

\bibliographystyle{plain}
\bibliography{mybib}

\end{document}